\documentclass{article}
\usepackage[margin=0.75in]{geometry}
\title{EigenPrism: Inference for High-Dimensional Signal-to-Noise Ratios}
\author{Lucas Janson, Rina Foygel Barber, Emmanuel Cand\`{e}s}
\date{}

\usepackage{color}
\usepackage{graphicx}
\usepackage{amsmath}
\usepackage{amssymb}
\usepackage{amsthm}
\usepackage{subfigure}
\usepackage{rotating}
\usepackage{verbatim}
\usepackage{hyperref}
\usepackage{bm}
\usepackage[authoryear]{natbib}
\usepackage{bbm}
\usepackage{xr}

\newcommand{\one}[1]{{\mathbbm{1}}_{{#1}}}
\newcommand{\inner}[2]{\langle{#1},{#2}\rangle} 
\newcommand{\norm}[1]{\left\lVert{#1}\right\rVert}
\newcommand{\PP}[1]{\mathbb{P}\left\{{#1}\right\}} 
\newcommand{\Pp}[2]{\mathbb{P}_{#1}\left({#2}\right)} 
\newcommand{\EE}[1]{\mathbb{E}\left({#1}\right)} 
\newcommand{\Ep}[2]{\mathbb{E}_{#1}\left({#2}\right)}

\newcommand{\EEst}[2]{\mathbb{E}\left({#1}\ \middle| \ {#2}\right)} 
\newcommand{\PPst}[2]{\mathbb{P}\left({#1}\ \middle| \ {#2}\right)} 

\def\R{\mathbb{R}}

\newcommand{\ident}[0]{\mathbf{I}}
\def\independenT#1#2{\mathrel{\rlap{$#1#2$}\mkern2mu{#1#2}}}
\newcommand\independent{\protect\mathpalette{\protect\independenT}{\perp}}
\newcommand{\iidsim}{\stackrel{\mathrm{iid}}{\sim}}
\newcommand{\indsim}{\stackrel{\independent}{\sim}}

\newcommand{\ignore}[1]{}

\newcommand{\eps}{\epsilon}

\newcommand{\pr}[1]{\mbox{$\mathbb{P}\left(#1\right)$}} 
\newcommand{\E}[1]{\mbox{$\mathbb{E}\left(#1\right)$}} 
 
\newcommand{\Ec}[2]{\mbox{$\mathbb{E}\left(\left.#1 \right| \,#2\right)$}} 

\newcommand{\iid}{\stackrel{\emph{i.i.d.}}{\sim}}
\newcommand{\bs}{\boldsymbol}

\newcommand{\var}{\ensuremath{\operatorname{Var}}}
\newcommand{\sd}{\ensuremath{\operatorname{SD}}}
\newcommand{\cov}{\ensuremath{\operatorname{Cov}}}
\newcommand{\val}{\ensuremath{\operatorname{val}}}
\newtheorem{2}{Theorem}
\newtheorem{1}[2]{Theorem}
\newtheorem{lemma}{Lemma}
\DeclareMathOperator*{\argmin}{arg\,min}
\newcommand{\eqnref}[1]{\eqref{eqn:#1}}
\newcommand{\Aset}{\mathcal{A}}

\graphicspath{{./fig/}}

\numberwithin{equation}{section}

\definecolor{darkred}{rgb}{0.6, 0.0, 0.0}
\usepackage[normalem]{ulem}

\newcommand{\rev}[1]{{\color{black}#1}}

\begin{document}
\maketitle

\begin{abstract}
  Consider the following three important problems in statistical
  inference, namely, constructing confidence intervals for (1) the
  error of a high-dimensional ($p>n$) regression estimator, (2) the
  linear regression noise level, and (3) the genetic signal-to-noise
  ratio of a continuous-valued trait (related to the
  heritability). All three problems turn out to be closely related to
  the little-studied problem of performing inference on the
  $\ell_2$-norm of the \rev{signal} in high-dimensional linear
  regression. We derive a novel procedure for this, which is
  asymptotically correct \rev{when the covariates are
  multivariate Gaussian} and produces valid confidence intervals in
  finite samples as well. The procedure, called \emph{EigenPrism}, is
  computationally fast and makes no assumptions on coefficient
  sparsity or knowledge of the noise level. We investigate the width
  of the EigenPrism confidence intervals, including a comparison with
  a Bayesian setting in which our interval is just 5\% wider than the
  Bayes credible interval. We are then able to unify the three
  aforementioned problems by showing that the EigenPrism procedure
  with only minor modifications is able to make important
  contributions to all three. We also investigate the robustness of 
  coverage and find that the method applies in practice and in finite
  samples much more widely than \rev{just the case of multivariate Gaussian
  covariates}. Finally, we apply EigenPrism to a genetic dataset to
  estimate the genetic signal-to-noise ratio for a number of
  continuous phenotypes.

\smallskip
\noindent \textbf{Keywords.} EigenPrism, Heritability, Regression
error, Signal-to-noise ratio, Variance estimation
\end{abstract}

\section{Introduction}
\subsection{Problem Statement}
Throughout this paper we will assume the linear model
\begin{equation}
\label{linmod}
\bs{y} = \bs{X}\bs{\beta} + \bs{\varepsilon},
\end{equation}
where $\bs{y}$, $\bs{\varepsilon} \in \mathbb{R}^n$, $\bs{\beta} \in
\mathbb{R}^p$, and $\bs{X} \in \mathbb{R}^{n \times p}$. Denote the
$i^{th}$ row and $j^{th}$ column of $\bs{X}$ by $\bs{x}_i$ and
$\bs{X}_j$, respectively. \rev{We assume the $\bs{x}_i$ are drawn
  i.i.d. from a mean-zero distribution with covariance matrix $\bs{\Sigma}$.}

Our goal is to construct a two-sided confidence interval (CI) for the \rev{expected}
signal squared magnitude $\theta^2 := \|\rev{\bs{\Sigma^{1/2}}}\bs{\beta}\|_2^2$ (or
equivalently just $\theta$). Explicitly, for a given significance
level $\alpha \in (0,1)$, we want to produce statistics $L_{\alpha}$
and $U_{\alpha}$, computed from the data, obeying
\begin{equation}
\label{uncondprob}
\begin{split}
\pr{\theta^2 < L_{\alpha}} \le
\alpha/2, \\
\pr{\theta^2 > U_{\alpha}} \le
\alpha/2. \\
\end{split}
\end{equation}
In words, we want to be able to make the following statement: ``with $100(1-\alpha)$\%
confidence, $\theta^2$ lies between $L_{\alpha}$ and
$U_{\alpha}$.''

\subsection{Motivation}
\label{motivation}
This problem can be motivated first from a high level as an approach
to performing inference on $\bs{\beta}$ in high dimensions. Since
$p>n$, we cannot hope to perform inference on the individual elements
of $\bs{\beta}$ directly (without further assumptions, such as
sparsity), but there is hope for the one-dimensional parameter
$\theta$. Although $\theta$ is not often considered a parameter of
inference in regression problems, it turns out to be closely related
to a number of well-studied problems.

Suppose one has an estimator $\hat{\bs{\beta}}$ for
$\bs{\beta}$. Perhaps the most important question to be asked is: how
close \rev{is} $\hat{\bs{\beta}}$ to $\bs{\beta}$? This question can be
answered statistically by estimating and/or constructing a CI for the
error of that estimate, namely,
$\|\hat{\bs{\beta}}-\bs{\beta}\|_2$. This is a fundamental statistical
problem arising in many applications. Consider, for example, a
compressed sensing (CS) experiment in which a doctor performs an MRI
on a patient. In MRI, the image is observed not in the spatial domain,
but in the frequency domain. If as many observations as pixels are
made, the result is the Fourier transform (with some added noise) of
the image, from which the original spatial pixels can be inferred. CS
theory suggests that one can instead use a number of observations
(rows of the Fourier matrix) that is a fraction of the number of
pixels, and still get very good recovery of the original image using
perhaps sophisticated $\ell_1$ methods \citep{candes06}. However, for
a specific instance, there is no good way to estimate how ``good'' the
recovery is. This can be important if the doctor is looking for a
specific feature on the MRI, such as a small tumor, and needs to know
if what he or she sees on the reconstructed image is accurate. In the
authors' experience, this is the most common question asked by
end-users of CS algorithms. Put another way, when the Nyquist sampling
theorem is violated, there is always a possibility of missing some of
the signal, so what reassurances can we make about the quality of the
reconstruction?

The estimation of the noise level $\sigma^2$ in a linear model is
another important statistical problem. Consider, for example,
performing inference on individual coefficients in the linear
model. When $n>p$, OLS theory provides an answer that depends on
$\sigma^2$ or at least an estimate of it. Indeed, one can find in almost
any introductory statistics textbook both estimation and inference
results for $\sigma^2$ in the case of $n>p$. However much recent work
has investigated the problem of performing inference on individual
coefficients in the high-dimensional setting of $n\le p$
\citep{berk2013, lockhart2014, taylor2014, javanmard2014,
  vandegeer2014, zhang2014, Lee2015}, and they all require knowledge of
$\sigma^2$. Unfortunately very few such results exist for the
high-dimensional setting of $n\le p$. Beyond regression coefficient
inference, $\sigma^2$ can be useful for benchmarking prediction
accuracy and for performing model selection, for instance using AIC,
BIC, or the Lasso. It also may be of independent interest to know
$\sigma^2$, for instance to understand the variance decomposition of $\bs{y}$.

A third topic is the study of genetic heritability
\citep{visscher2008}, which can be characterized by the following
question: what fraction of variance in a trait (such as height) is
explained by our genes, as opposed to our environment? Colloquially,
this can be considered a way of quantifying the nature versus nurture
debate.

It turns out that all three of these problems can be solved by
connection with our original problem of estimating and constructing
CIs for $\theta^2$. Indeed, in the MRI example, the doctor may split
the collected observations into two independent subsamples,
$(\bs{y}^{(0)},\bs{X}^{(0)})$ and $(\bs{y}^{(1)},\bs{X}^{(1)})$, and
construct an estimator $\hat{\bs{\beta}}$ from just $(\bs{y}^{(0)},\bs{X}^{(0)})$. Then
the vector $\tilde{\bs{y}} := \bs{y}^{(1)}-\bs{X}^{(1)}\hat{\bs{\beta}}$ follows
a linear model,
\begin{equation}
\label{HDreg}
\tilde{\bs{y}} = \bs{X}^{(1)}(\bs{\beta}-\hat{\bs{\beta}}) + \bs{\varepsilon},
\end{equation}
so that \rev{if $\bs{\Sigma}=I$,} inference on $\theta$ in this linear model corresponds exactly
to inference on the $\ell_2$ regression error of
$\hat{\bs{\beta}}$. Note that since the analysis is \emph{conditional}
on $\hat{\bs{\beta}}$, there is no restriction on how
$\hat{\bs{\beta}}$ is computed from $(\bs{y}^{(0)},\bs{X}^{(0)})$,
and so the method applies to \emph{any} coefficient estimation
technique. We defer the
connection between inference for $\theta^2$ and inference for
$\sigma^2$ and genetic variance decomposition to Section~\ref{relatedproblems}.

\subsection{Main Result}
Although we will ultimately
argue that our method applies more broadly, we will begin with
the following distributional assumptions,
\begin{equation}
\label{normassump}
\bs{x}_i \iid N(0, \bs{I}_p), \qquad \varepsilon_i \iid N(0, \sigma^2),
\end{equation}
with $\bs{X}$ independent of $\bs{\varepsilon}$. Note that $p>n$
ensures the design matrix will have a nontrivial null space, and thus
conditional on $\bs{X}$, the linear model~\eqref{linmod} (including
$\theta$) is unidentifiable (since any vector in the null space of
$\bs{X}$ can be added to $\bs{\beta}$ without changing the
data-generating process). This necessitates a random design
framework. The assumption of independence on the rows of the design
matrix is often satisfied in realistic settings when observations are
drawn independently from a population. However, the independence (and
multivariate Gaussianity) of the columns is rather stringent and just
a starting point\rev{---Sections~\ref{heritability}, \ref{robust}, and
  \ref{nfbc} demonstrate in simulations and on real data that in practice EigenPrism
  achieves nominal coverage even when the marginal distribution of the
entries of $\bs{X}$ are far from Gaussian, as well as in some cases
when $\bs{\Sigma}\neq I$}. We are treating the coefficient vector $\bs{\beta}$ as fixed,
not random.

Under these assumptions, we will develop in Section~\ref{method} an
estimator that is unbiased for $\theta^2$, is \rev{asymptotically} normally
distributed, and has an estimable tight bound on its variance. None of
these properties, including estimability of the variance, require
knowledge of the noise level $\sigma^2$ or any assumption, such as
sparsity, on the structure of the coefficient vector $\bs{\beta}$. From
these results, it is easy to generate valid CIs for $\theta^2$ (or
$\theta$), and we will show that such CIs are nearly as short as they can
be, and provide nominal coverage in finite samples under a variety of
circumstances (even beyond the assumptions made here).

\subsection{Related Work}
When $n > p$, ordinary least squares (OLS) theory gives us inference
for $\bs{\beta}$ and thus also for $\theta$. When $n \le p$, the
problem of estimating $\theta^2$ has been studied in
\cite{Dicker2014}. \cite{Dicker2014} uses the method of moments on two
statistics to estimate $\theta^2$ and $\sigma^2$ without assumptions
on $\bs{\beta}$\rev{, and with the same multivariate Gaussian random
  design assumptions used here}.  \cite{Dicker2014} also derives asymptotic
distributional results, but does not explore the estimation of the
parameters of the asymptotic distributions, nor the coverage of any CI
derived from it. The main contribution of our work is to provide
tight, \emph{estimable} CIs which achieve nominal coverage even in
finite samples.

Inference for high-dimensional regression error, noise level, and
  genetic variance decomposition are each individually well-studied, so we review some
  relevant works here. To begin with, many authors have studied
  high-dimensional regression error for specific coefficient
  estimators, such as the Lasso~\citep{Tibshirani1996}, often
  providing conditions under which this regression error asymptotes to
  0 (see for example \cite{Bayati2013,KnightFu2000}). To our knowledge
  the only author who has considered inference for a general estimator
  is \cite{Ward2009}, who does so using the Johnson--Lindenstrauss
  Lemma and assuming \emph{no noise}, that is, $\varepsilon_i\equiv 0$
  in the linear model~\eqref{linmod}. Thus the problem studied there
  is quite different from that addressed here, as we allow for noise
  in the linear model. Furthermore, because the Johnson--Lindenstrauss
  Lemma is not distribution-specific, it is conservative and thus
  Ward's bounds are in general conservative, while we will show that
  in most cases our CIs will be quite tight.

There has also been a lot of recent interest in estimating the
noise level $\sigma^2$ in high-dimensional regression
problems. \cite{Fan2012} introduced a refitted cross validation method
that estimates $\sigma^2$ assuming sparsity and a model selection
procedure that misses none of the correct variables. \cite{Sun2012}
introduced the scaled Lasso for estimating $\sigma^2$ using an
iterative procedure that includes the Lasso. \cite{Stadler2010} also
use an $\ell_1$ penalty to estimate the noise level, but in a finite
mixture of regressions model. \cite{Bayati2013} use the Lasso and
Stein's unbiased risk estimate to produce an estimator for
$\sigma^2$. All of these works prove consistency of their estimators,
but under conditions on the sparsity of the coefficient
vector. \rev{Indeed, it can be shown \citep{Giraud2012} that such a
  condition is needed when $\bs{X}$ is treated as fixed (which it is
  not in the present paper).} Under
the same sparsity conditions, \cite{Fan2012} and \cite{Sun2012} also
provide asymptotic distributional results for their estimators,
allowing for the construction of asymptotic CIs. What distinguishes our treatment of
this problem from the existing literature is that our estimator and CI for $\sigma^2$
make \emph{no} assumptions on the sparsity or structure of
$\bs{\beta}$.

An unpublished paper \citep{Owen2012} estimates $\theta^2$ using
a type of method of moments, with the goal of estimating genetic
heritability by way of a variance decomposition. Although
\citeauthor{Owen2012} gives conditions for consistency of his
esimator, no inference is discussed, and he points
out that the work is only valid for estimating heritability if the SNPs
are assumed to be independent. In general, heritability is a
well-studied subject in genetics, with especially accurate estimates coming from
studies comparing a trait within and between twins
(e.g. \cite{THG:8493904}). However, in order to better understand the
genetic basis of such traits, some authors have tried to directly
predict a trait from genetic information. Since most forms of genetic
information, such as SNP data, are much higher-dimensional than the
number of samples that can be obtained, the main approaches are either
to try and find a small number of important variables through
genome-wide association studies (e.g. \cite{weedon2008genome}) before
modeling, to estimate the kinships among subjects and use maximum
likelihood, assuming independence among SNPs and random effects, on
the trait covariances among subjects to estimate the
(narrow-sense) heritability (e.g. \cite{yang2010common,Golan2011a}), or
to assume random effects and use maximum likelihood to estimate the
signal-to-noise ratio in a linear model
(e.g. \cite{Kang2008,bonnet2014heritability,OwenPC}). However, attempts to explain
heritability by genetic prediction have fallen quite short
of the estimates from twin studies, leading to the famous conundrum of
\emph{missing heritability} \citep{manolio2009finding}. Our main
contribution to this field will be to consistently
estimate and provide inference for the signal-to-noise ratio
in a linear model, which is related to the heritability, without assumptions on
the coefficient vector (such as sparsity or random effects), knowledge
of the noise variance, or feature independence. This contribution may
be especially valuable given the increased popularity of the rare
variants hypothesis \citep{Pritchard2001} for missing heritability,
which conjectures that the effects of genetic variation on a trait may
not be strong and sparse, but instead distributed and weak (and their
corresponding mutations rare).

We note that neuroscientists have also done work estimating a
signal-to-noise ratio, namely the explainable variance in functional
MRI. That problem is made especially challenging due to correlations
in the noise, making it different from the i.i.d.~noise setting
considered in this paper. For this related problem,
\cite{Benjamini2013} are able to construct an unbiased estimator in
the random effects framework by permuting the measurement vector in
such a way as to leave the noise covariance structure unchanged.

\section{Constructing a Confidence Interval for $\theta^2$}
\label{method}
In this section we develop a novel method for constructing a valid CI
for $\theta^2$. This method does not require $\sigma^2$ to be
known. However, for pedagogical reasons, we begin with the simpler
situation in which $\sigma^2$ is known, which may arise in many signal
or image processing applications.

\subsection{Known $\sigma^2$}
\label{para}
Consider a sample of size $n$ from the linear model \eqref{linmod}. Then
\begin{equation*}
y_i \iid N(0, \theta^2+\sigma^2),
\end{equation*}
which implies
\begin{equation}
\label{chi2distr}
\frac{\| \bs{y} \|_2^2}{\theta^2 +
  \sigma^2} \sim \chi_n^2.
\end{equation}
Denote the $\tau^{th}$ quantile of the $\chi_n^2$ distribution by
$Q^{(n)}_{\tau}$. Then when $\sigma^2$ is known, a valid CI can be
obtained by setting
\begin{equation*}
L_{\alpha} = \frac{\| \bs{y} \|_2^2}{Q^{(n)}_{1-\alpha/2}} - \sigma^2, \qquad
U_{\alpha} = \frac{\| \bs{y} \|_2^2}{Q^{(n)}_{\alpha/2}} - \sigma^2,
\end{equation*}
that is, \eqref{uncondprob} is satisfied under this choice of
$L_{\alpha},U_{\alpha}$.  Note that the method of \cite{Ward2009} also
assumes $\sigma^2$ is known, and equal to zero, so we may consider
comparing it to the above. In particular we want to emphasize that
\cite{Ward2009}'s inference method is conservative due to the
generality of the Johnson--Lindenstrauss lemma, while
$[L_{\alpha},U_{\alpha}]$ contitutes an \emph{exact} $100(1-\alpha)\%$
CI. The same procedure can be generalized using the bootstrap on the
unbiased estimator 
\begin{equation}\label{eqn:T1def} T_1 := \|\bs{y}\|_2^2/n-\sigma^2\;.\end{equation}
See
Appendix~\ref{nonGaussKnownSigma} for details.

\subsection{Unknown $\sigma^2$}

\subsubsection{Theory}
Consider again the linear model \eqref{linmod} with assumptions
\eqref{normassump}, in particular that $\bs{X}$
has i.i.d. standard Gaussian elements. Recall that we assume $n < p$, and let $\bs{X} = \bs{UDV}^{\top}$
be a singular value decomposition (SVD) of $\bs{X}$, so that $\bs{U}$ is $n \times n$
orthonormal, $\bs{D}$ is $n \times n$ diagonal with non-negative,
non-increasing diagonal entries, and $\bs{V}$ is $p \times n$
orthonormal. Let
$\bs{z} = \bs{U}^{\top} \bs{y}$, and denote the diagonal vector of $\bs{D}$ by
$\bs{d}$. We emphasize that the singular values in $\bs{D}$ are
arranged along the diagonal in decreasing order, so that $d_1 \ge d_2 \ge \cdots d_n \ge 0$. Then
\begin{equation*}
\bs{z} = \bs{D}(\bs{V}^{\top}\bs{\beta}) + \bs{U}^{\top}\bs{\varepsilon},
\end{equation*}
and note that
\begin{equation*}
\begin{split}
\mathbb{E}\left(z_i^2 | \bs{d}\right) =&\, \mathbb{E}\left[\Big(d_i
  \bs{V}_i^{\top} \bs{\beta} + \bs{U}^{\top}_i \bs{\varepsilon}\Big)^2 \Big| \bs{d}\right], \\
=&\, d_i^2 \mathbb{E}\left[\Big(\bs{V}^{\top}_i \bs{\beta}\Big)^2 \Big| \bs{d} \right] +\rev{2} d_i
\mathbb{E}\left(\bs{V}_i^{\top} \bs{\beta}
  \bs{U}_i^{\top} \bs{\varepsilon} \big| \bs{d} \right) +
\mathbb{E}\left[\Big(\bs{U}_i^{\top} \bs{\varepsilon}\Big)^2 \Big| \bs{d}
\right], \\
=&\, d_i^2 \theta^2/p + \sigma^2, \\
\end{split}
\end{equation*}
where the third equality follows from the fact that in our model the
columns of $\bs{V}$ are uniformly
distributed on the unit sphere, and independent of $\bs{d}$.

To give some intuition for what follows, assume $n$ is even and
consider the expectation, conditional on $\bs{d}$, of the difference
between the sum of squares of the first half of the entries of
$\bs{z}$ and the sum of squares of the second half of the entries of $\bs{z}$,
\begin{equation*}
\begin{split}
\mathbb{E}\left(\left.\sum_{i=1}^{n/2} z_i^2 - \sum_{i=n/2+1}^{n}
    z_i^2\right|\bs{d}\right) =&\, \left(\sum_{i=1}^{n/2} d_i^2 \theta^2/p +
  \frac{n}{2}\sigma^2\right) - \left(\sum_{i=n/2+1}^{n} d_i^2 \theta^2/p +
  \frac{n}{2}\sigma^2\right), \\
=&\, \frac{\theta^2}{p}\sum_{i=1}^{n/2} \left(d_i^2 - d_{i+n/2}^2\right). \\
\end{split}
\end{equation*}
Note that the terms containing $\sigma^2$ in the first line cancel
out, but because the singular values $d_i$ of $\bs{X}$ are in decreasing
order, a term proportional to $\theta^2$ remains. We generalize this
idea below.

Let $\lambda_i = d_i^2/p$ be the eigenvalues of $\bs{X}\bs{X}^{\top}/p$,
let $\bs{w} \in \mathbb{R}^n$ be a vector of weights (which need not
be nonnegative), and consider the statistic $S = \sum_{i=1}^n w_i
z_i^2$. We can compute its expectation, conditional on $\bs{d}$, as
\begin{equation}
\begin{split}
\label{condexp}
\Ec{S}{\bs{d}} & = \mathbb{E}\left(\sum_{i=1}^{n} w_i z_i^2 \Big|
  \bs{d}\right) \\ 
& = \sum_{i=1}^n w_i \Big(\lambda_i \theta^2 + \sigma^2 \Big) \\
& = \theta^2 \sum_{i=1}^n w_i \lambda_i + \sigma^2 \sum_{i=1}^n w_i. \\
\end{split}
\end{equation}
Based on this
calculation, constraining $\sum_{i=1}^{n} w_i = 0$ and
$\sum_{i=1}^{n} w_i \lambda_i = 1$ makes $S$ an unbiased estimator of $\theta^2$
 (even
conditionally on $\bs{d}$). We can
also compute its conditional variance (see Appendix~\ref{minmaxvar}
for a detailed computation),
\begin{equation}
\label{varExact}
\begin{split}
\var(S | \bs{d}) = & \,2 \sigma^4 \sum_{i=1}^n w_i^2 + 4 \sigma^2
\theta^2 \sum_{i=1}^n w_i^2 \lambda_i + 2 \theta^4 \left[\frac{p}{p+2} \sum_{i=1}^n w_i^2 \lambda_i^2 - \frac{\left(\sum_{i=1}^n w_i \lambda_i\right)^2}{p+2}\right], \\
\end{split}
\end{equation}
which, under the aforementioned constraint $\sum_{i=1}^{n} w_i
\lambda_i = 1$ can be rewritten as
\begin{equation}
\label{varUB1}
\begin{split}
= & \,2 \sigma^4 \sum_{i=1}^n w_i^2 + 4 \sigma^2
\theta^2 \sum_{i=1}^n w_i^2 \lambda_i + 2 \theta^4 \left(\frac{p}{p+2} \sum_{i=1}^n w_i^2 \lambda_i^2 - \frac{1}{p+2}\right), \\
\le & \,2 \sum_{i=1}^n w_i^2 \left(\lambda_i\theta^2 + \sigma^2\right)^2, \\
= & \,2\left(\theta^2+\sigma^2\right)^2\sum_{i=1}^n w_i^2 \left(\lambda_i \rho +
1-\rho\right)^2, \\
\end{split}
\end{equation}
where 
\begin{equation}
\label{rho}
\rho = \frac{\theta^2}{\theta^2+\sigma^2}
\end{equation}
is the fraction of the variance of the $y_i$ accounted for by the
signal (recall that $\var(y_i) = \theta^2+\sigma^2$). The inequality
will be quite tight when $p$ is large and
$\frac{1}{p\sum_{i=1}^nw_i^2\lambda_i^2} \ll 1$. By noting that this
variance bound, as a function of $\rho$, is a quadratic equation with
positive leading coefficient, it follows that it is maximized either
at $\rho=0$ or at $\rho=1$. This leads to one more upper-bound,
\begin{equation}
\label{varUB}
\var(S|\bs{d}) \le 2 \left(\theta^2 + \sigma^2
\right)^2 \cdot \max \left( \sum_{i=1}^n w_i^2,
\sum_{i=1}^n w_i^2 \lambda_i^2 \right).
\end{equation}

The above equation has two striking features. The first is that it
depends on $\theta^2$ and $\sigma^2$ only through the sum $\theta^2 + \sigma^2$, for
which we have an excellent estimator given by $||\bs{y}||_2^2/n$. The
second feature is that it separates into the product of two terms: one
term that does not depend on $\bs{w}$, and a second term that is
known (in that it contains nothing that needs to be estimated) and
(strictly) \emph{convex} in $\bs{w}$. Thus we can use convex
optimization to find the vector $\bs{w}$ that minimizes the
upper-bound \eqref{varUB} on the variance subject to the two linear
equality constraints mentioned earlier, $\sum_{i=1}^{n} w_i = 0$ and
$\sum_{i=1}^{n} w_i \lambda_i = 1$, which ensure that $S$
remains unbiased for $\theta^2$. Figure~\ref{weights} shows an example
of such an optimized weight vector when $n=200$ and $p=2000$.
\begin{figure}\centering
\includegraphics[width=0.45\textwidth]{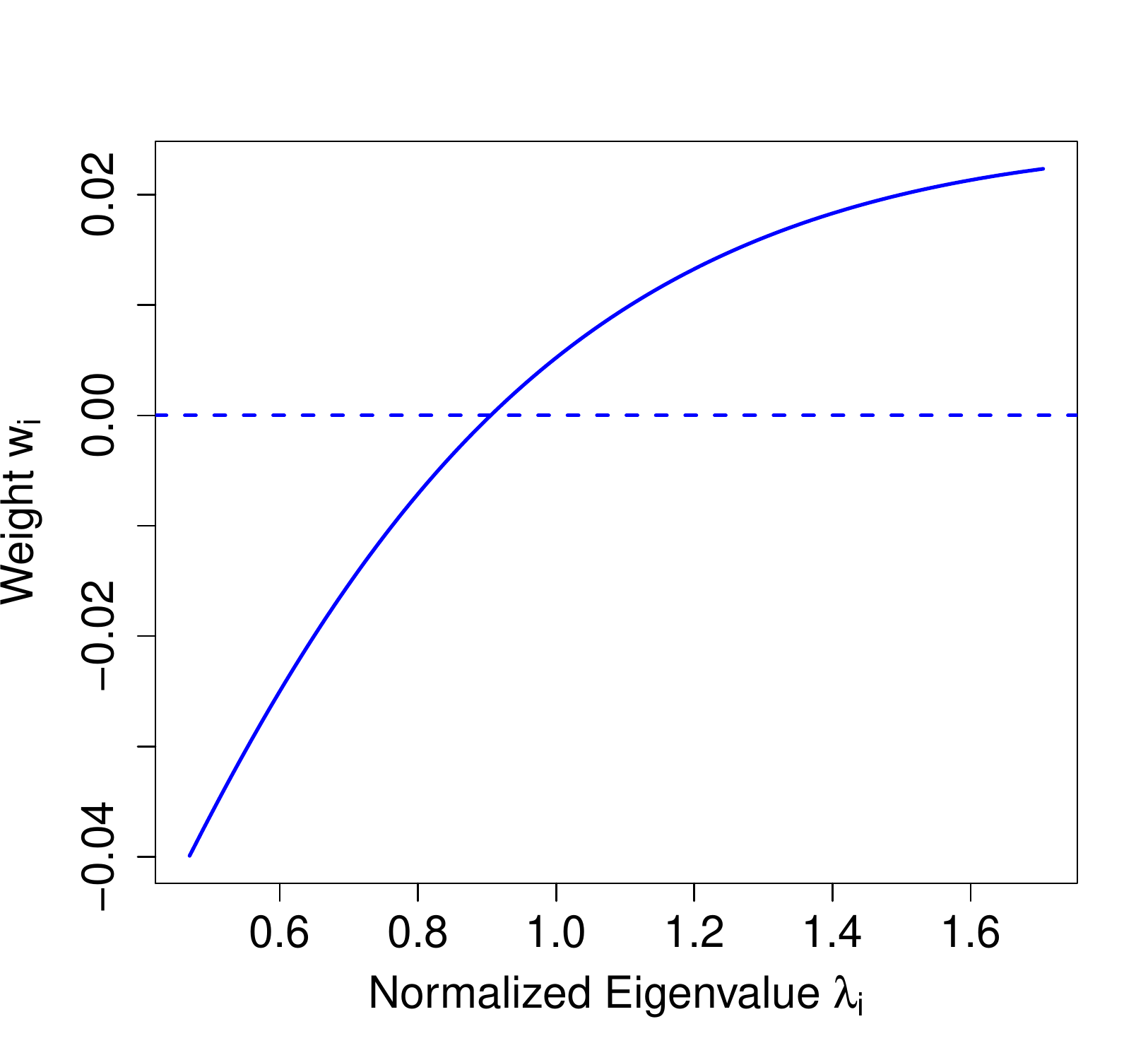}
\caption{Plot of weights $w_i$ as a function of normalized eigenvalues
  $\lambda_i$ for $n=200$ and $p=2000$.}
\label{weights}
\end{figure}
Note that instead of just giving some positive weight to large
$\lambda_i$'s and some negative weight to small $\lambda_i$'s, the
optimal weighting is a smooth function of the $\lambda_i$. This makes
sense, as the $z_i$'s with large associated $\lambda_i$ have a larger
signal-to-noise ratio, and should be given greater weight. Denote the
statistic $S$ constructed using these constrained-optimal weights by
$T_2$. Explicitly, let $\bs{w}^*$ be the
solution to the following convex optimization program $\mathcal{P}_1$:
\begin{equation}
\label{wopt}
\argmin_{\bs{w}\in\mathbb{R}^n} \; \max \left( \sum_{i=1}^n w_i^2, \sum_{i=1}^n
  w_i^2 \lambda_i^2 \right) \quad
\text{such that}  \, \sum_{i=1}^n w_i = 0,  \; \sum_{i=1}^n w_i \lambda_i = 1
\end{equation}
and denote by $\val(\mathcal{P}_1)$ the minimized objective function value. Then the
statistic for our main procedure in this paper, which we call the
\emph{EigenPrism} procedure, is the following,
\begin{equation}
\begin{split}
\label{minmax}
T_2 :=&\, \sum_{i=1}^{n} w^*_i z_i^2, \\
\Ec{T_2}{\bs{d}} =&\, \theta^2, \\
\sd(T_2|\bs{d}) \lesssim&\, \sqrt{2\val(\mathcal{P}_1)} \, \frac{\|\bs{y}\|_2^2}{n}, \\
\end{split}
\end{equation}
where the only approximation in the variance is the replacement of
$\theta^2+\sigma^2$ by its estimator $\|\bs{y}\|_2^2/n$.

With these calculations in place, we now define our
$(1-\alpha)$-confidence interval for $\theta^2$, by assuming that
$T_2$ follows an approximately normal distribution (discussed later
on).  We construct lower and upper endpoints
\[
L_{\alpha} := \max\left(T_2 - z^\star_{1-\alpha/2}\cdot  \sqrt{2\val(\mathcal{P}_1)} \, \frac{\|\bs{y}\|_2^2}{n}, 0\right)\;,
\quad
U_{\alpha}:=T_2 + z^\star_{1-\alpha/2}\cdot  \sqrt{2\val(\mathcal{P}_1)} \, \frac{\|\bs{y}\|_2^2}{n}\;,
\]
where the value of $L_{\alpha}$ is clipped at zero since it holds
trivially that $\theta^2>0$, and where $z^\star_{1-\alpha/2}$ is the
$(1-\alpha/2)$ quantile of the standard normal distribution.

{\bf Remark.} The idea of constructing the $z_i$'s as contrasts has
been used in the heritability literature before,
e.g.~\cite{Kang2008,bonnet2014heritability,OwenPC}, but in a strict
random effects framework. In particular, when the entries of
$\bs{\beta}$ are i.i.d. Gaussian, the $z_i$'s become independent. With
independent $z_i$'s whose distribution depends only on the signal
($\theta^2$) and noise ($\sigma^2$) parameters, the authors are able
to apply maximum likelihood estimation, with associated asymptotic
inference results for the signal, noise, or signal-to-noise ratio (we
note that \cite{bonnet2014heritability} generalize such estimators
somewhat to the case of a Bernoulli-Gaussian random effects
model). The crucial difference between our work and theirs is that we
make no assumptions (e.g., Gaussianity, sparsity) on the coefficient
vector, and thus not only are the $z_i$'s not independent in our
setting, but their dependence (and thus the full likelihood) is a
function of the products $\beta_i\beta_j$, and thus a maximum
likelihood approach in this setting would still be overparameterized.

Next, we discuss the coverage and width properties of this constructed
confidence interval.

\subsubsection{Coverage}
Now that we are equipped with an unbiased estimator and a computable
variance (upper-bound), and have constructed a confidence interval
(CI) using a normal approximation, there are two main questions to
answer in order to determine whether these CIs will exhibit the
desired coverage properties.  In particular, we would like to know if
substituting $\theta^2+\sigma^2$ with $\|\bs{y}\|_2^2/n$ substantially
affects the variance formula, and we would like to know if $T_2$ is
approximately normally distributed (so that we can construct arbitrary
CIs from just the second moment). For the first question, since
$\|\bs{y}\|_2^2$ is a rescaled $\chi_n^2$ random variable, \rev{for
  nominal coverage of $1-\alpha$, the coverage actually achieved can
  be closely approximated by $\PP{|N(0,1)|\le
    z_{1-\alpha/2}\cdot\chi_n^2/n}$ (where the $N(0,1)$ and the
  $\chi_n^2$ are independent), assuming exact
  normality. Table~\ref{tab:coverage} shows that for nominal 95\%
  coverage, one would need fewer than 20 samples to achieve less than 90\% coverage.}
\begin{table}[ht]\centering
\begin{tabular}{|r|c|c|c|c|c|c|c|}
\hline
$n$ & 10 & 20 & 50 & 100 & 500 & 1000 & 5000 \\
\hline
Coverage & 87.5\% & 91.0\% & 93.3\% & 94.1\% & 94.8\% & 94.9\% & 95.0\% \\
\hline
\end{tabular}
\caption{Values of $\PP{|N(0,1)|\le z_{0.975}\cdot\chi_n^2/n}$ for a range
of $n$.}
\label{tab:coverage}
\end{table}
For the second question, \rev{the following theorem establishes the
  asymptotic normality of $T_2$.
\begin{2}\label{T2normal}
Under the linear model~\eqref{linmod} with Gaussian random design and
errors given in Equation~\eqref{normassump}, the estimator $T_2$ as
defined in Equation~\eqref{minmax} is asymptotically
normal as $n,p\rightarrow\infty$ and $n/p \rightarrow \gamma\in
(0,1)$.  This holds for any values of $\theta^2$ and $\sigma^2$,
including values that vary with $n$. Explicitly,
\[\frac{T_2-\theta^2}{\sd(T_2|d)} \stackrel{d}{\longrightarrow} N(0,1).\]
\proof The proof is given in Appendix~\ref{asymptoticnormality}.
\end{2}
For finite-sample results,} we defer
to the simulation results of Section~\ref{regrerror} to show that for problems of
reasonable size ($\min\{n,p\} \gtrsim 100$), CIs constructed as if
$T_2$ were \emph{exactly} normal with variance \emph{exactly} given by
Equation~\eqref{minmax} \emph{never} result in below-nominal coverage.

\subsubsection{Width}
\label{sec:width}
Once we have confirmed that our CIs provide the proper coverage, the
next topic of interest is their widths. It is not hard to obtain a
closed-form asymptotic upper-bound for $\var(T_2)$ (the
details are worked out in Appendix~\ref{CFvarUB}). In particular, letting
$Y_{\gamma}$ denote a random variable with Mar{\v c}enko--Pastur (MP) distribution
with parameter $\gamma$ \citep{Marcenko1967}, and $M_{\gamma}$ denote the median of
$Y_{\gamma}$, define the constants,
\begin{equation}\label{eqn:ABgamma}
\begin{split}
A_{\gamma} =&\, \mathbb{E}\left[Y_{\gamma} \cdot (\mathbbm{1}_{Y_{\gamma} \ge M_{\gamma}}-\mathbbm{1}_{Y_{\gamma}< M_{\gamma}})\right], \\
B_{\gamma} =&\, \E{Y_{\gamma}^2}. \\
\end{split}
\end{equation}
Then in the limit as $n,p \rightarrow \infty$ and $n/p \rightarrow \gamma \in (0,1)$,
\begin{equation}
\label{varbound}
\sqrt{n} \cdot \frac{\sqrt{\var(T_2)}}{\theta^2+\sigma^2} \le \sqrt{2}
\cdot \max\left(\frac{1}{A_{\gamma}},\frac{\sqrt{B_{\gamma}}}{A_{\gamma}}\right).
\end{equation}

We can draw a few conclusions from Equation~\eqref{varbound}. The most
obvious is that for $n,p \rightarrow \infty$, $n/p \rightarrow \gamma
\in (0,1)$, $\sigma^2$ asymptotically bounded above and $\theta^2$
asymptotically bounded \emph{below}, the error of $T_2$, as a fraction
of its estimand $\theta^2$, converges to 0 in probability at a rate of
$n^{-1/2}$. Note that we make no assumptions at all on the structure
of $\bs{\beta}$, and just require that $\theta^2$ does not asymptote
at 0. The equation also lets us compute a conservative upper-bound on
the asymptotic relative efficiency (ARE), defined as the asymptotic
ratio of standard deviations (although it is often defined by
variances elsewhere), of $T_2$ with respect to $T_1$ from
Section~\ref{para} (see~\eqref{eqn:T1def}), the latter of which uses exact knowledge of
$\sigma^2$ and has standard deviation characterized \rev{by the
  $\chi^2_n$ distribution}. While we may not be able to formulate a
closed-form expression for it in terms of expectations due to the
constrained minimization functional, the standard deviation bound for
$T_2$ in Equation~\eqref{minmax} will also converge to a constant
times $\sd(T_1)$ under the same asymptotic conditions, where the
constant depends only on the MP distribution. This is because the
optimal weights are a smooth function of the $\lambda_i$. Due to fast
convergence to the MP distribution, we can numerically approximate
this \emph{exact} asymptotic ratio. Figure~\ref{T1T2are} shows this
estimate of the ARE of $T_2$ to $T_1$ as a function of $\gamma$. Note
that the standard deviation bound for $T_2$ in
Equation~\eqref{minmax}, used to compute the curve in
Figure~\ref{T1T2are}, is still an upper-bound for the ARE of $T_2$
with respect to $T_1$, but it reflects the ratio of CI widths between
the EigenPrism procedure and a CI constructed from $T_1$ with
knowledge of $\sigma^2$.
\begin{figure}\centering
\includegraphics[width=0.45\textwidth]{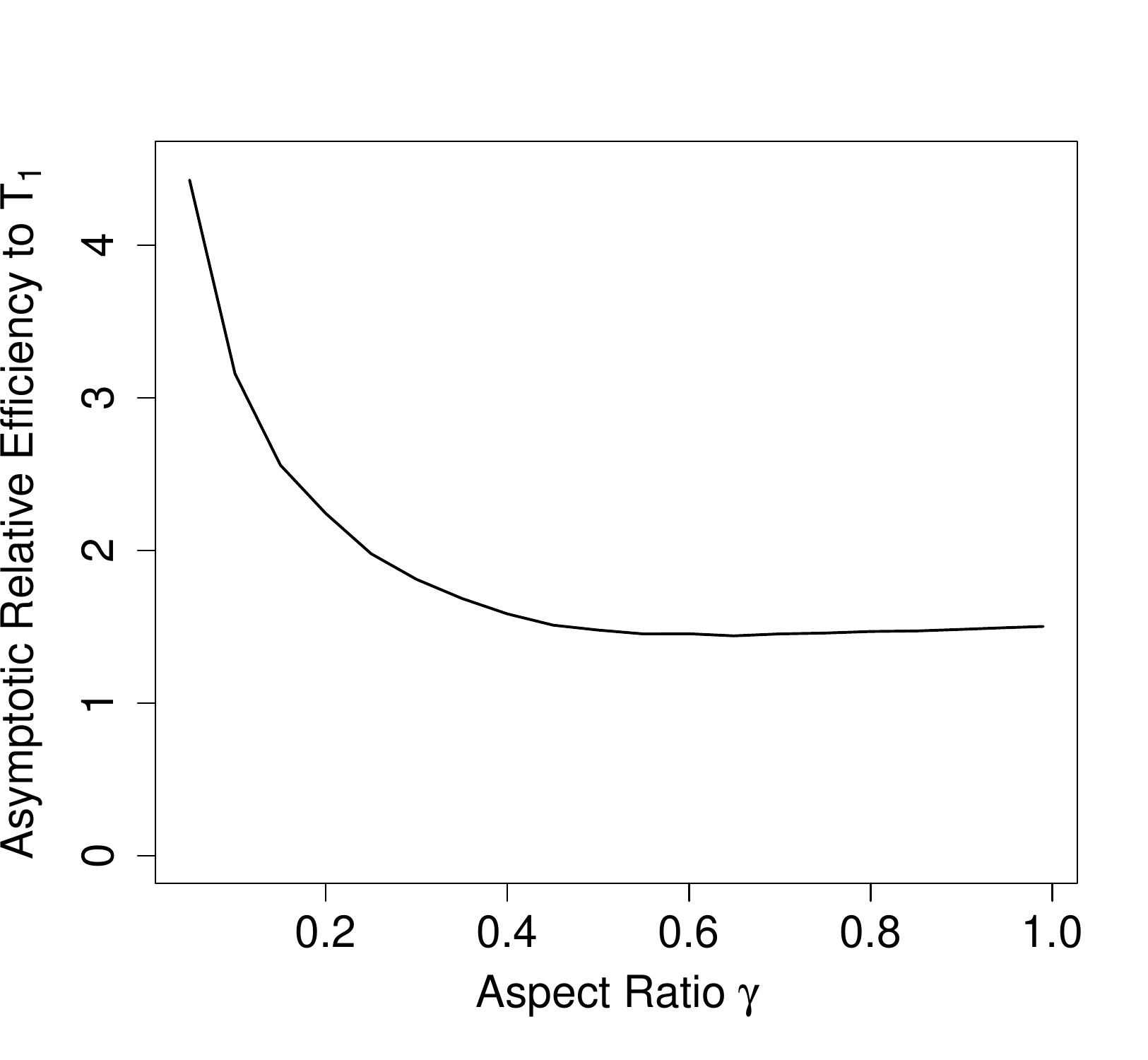}
\caption{Estimate of the asymptotic relative efficiency of
  $T_2$ to $T_1$.}
\label{T1T2are}
\end{figure}
The figure demonstrates how close in width the EigenPrism procedure
comes to an exact CI for $T_1$ which knows $\sigma^2$. In particular,
for $\gamma \gtrsim 0.25$, the EigenPrism CIs are at most twice as
wide as those for $T_1$.

 Another notable feature of
  Figure~\ref{T1T2are} is how large the ARE becomes as $\gamma
  \rightarrow 0$. This is a symptom of an important property of not
  just our procedure, but the frequentist problem as a whole. First,
  it is clear that if all the $\lambda_i\equiv 1$, our procedure
  fails, as the $z_i^2$ no longer provide any contrast between
  $\theta^2$ and $\sigma^2$, and no linear combination of them will
  produce an unbiased statistic for $\theta^2$. Intuitively, note that
  $\E{z_i^2} \propto \lambda_i\rho + (1-\rho)$, so that the problem of
  estimating $\rho$ is that of estimating the slope and intercept of a
  regression line. But in regression, when the predictor variable
  assumes a constant value, as it would when $\lambda_i\equiv 1$, it
  becomes impossible to estimate the slope and intercept. To understand better
  how our procedure performs when the spread of the $\lambda_i$
  approaches zero, consider the case when
  $\lambda_1=\cdots=\lambda_{n/2}=1+a$ and
  $\lambda_{n/2+1}=\cdots=\lambda_n=1-a$. In this case $\sd(\lambda_i)
  = a$, and it is easy to show that
\begin{equation*}
\val(\mathcal{P}_1) = \frac{1+a^2}{a^2n},
\end{equation*}
so if $\sqrt{n}\sd(\lambda_i)\rightarrow 0$, then $a^2n\rightarrow 0$
and so $\val(\mathcal{P}_1)\rightarrow\infty$.

Returning to our original
model in which $\bs{X}$ is i.i.d. $N(0,1)$, the $\lambda_i$'s will be
approximately MP-distributed with parameter
$\gamma=n/p$. 
\begin{figure}\centering
\includegraphics[width=0.45\textwidth]{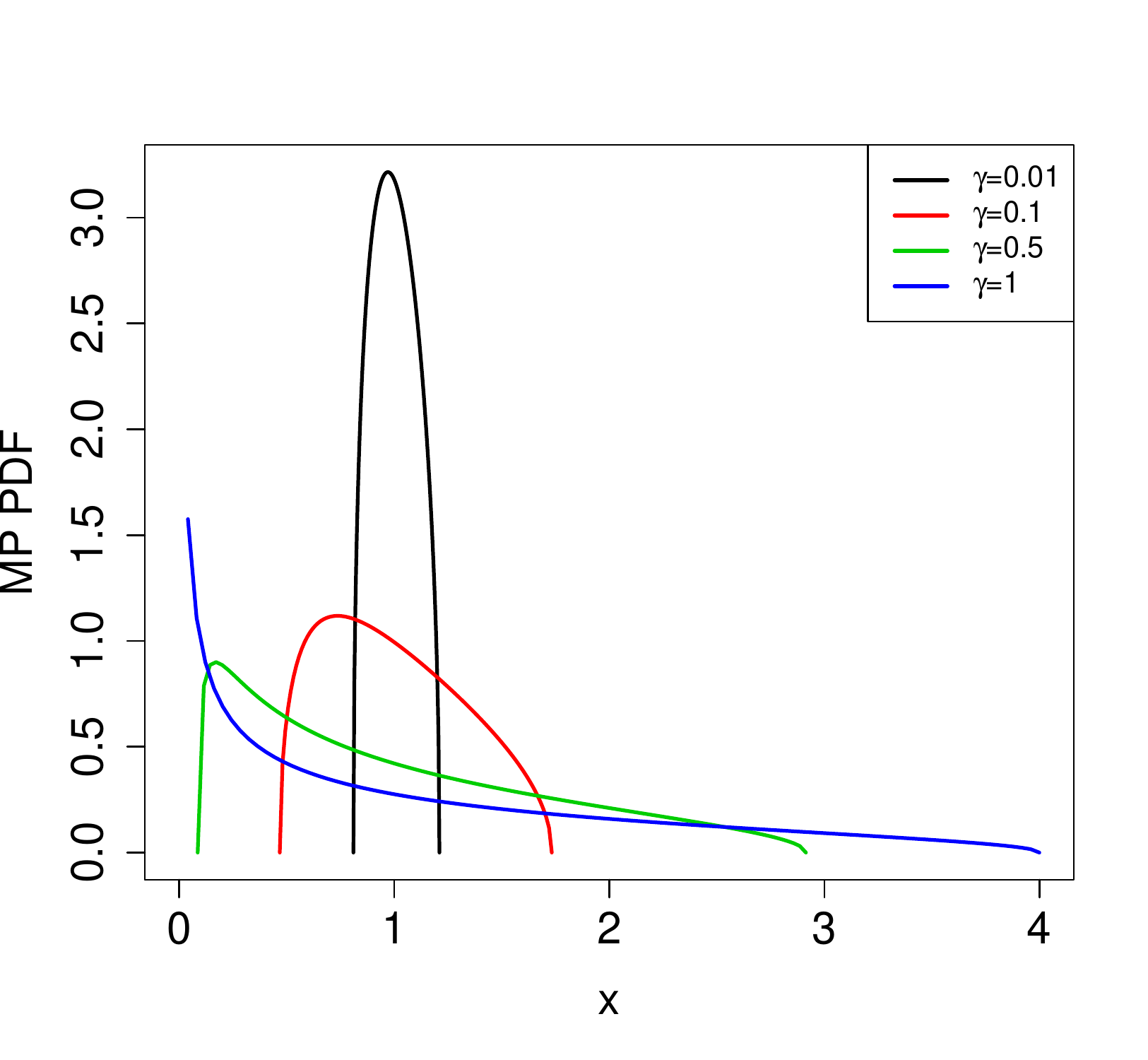}
\caption{Probability density function (PDF) of Mar{\v c}enko--Pastur
  distribution for various values of $\gamma$.}
\label{MPdistr}
\end{figure}
Figure~\ref{MPdistr} shows visually how the width of the MP
distribution depends on $\gamma$, and analytically,
$\sd(Y_{\gamma})=\sqrt{\gamma}$. We show in the following theorem
(proved in Appendix~\ref{lambdasone}) that if the $\lambda_i$'s are
too close to 1 and $n\ll p$, it is impossible for \emph{any procedure}
to reliably distinguish between the case of $\rho=0$ (pure noise) and
$\rho=1/2$ (variance equally split between signal and noise).

\begin{1}
\label{thm:lambdas}
Let $p\geq n\geq 1$. Suppose that
\begin{equation}\label{eqn:Z_model}\bs{Z} = \theta \cdot \bs{DV}^{\top} \bs{a} + \sigma \cdot \bs{\varepsilon}\end{equation}
where $\theta>0$, $\sigma>0$, and a unit vector $\bs{a}$ are all fixed
but unknown, $\bs{D}\in\mathbb{R}^{n\times n}$ is a known nonnegative
diagonal matrix with $\sum_{i=1}^n \lambda_i:= \sum_{i=1}^n
D_{ii}^2/p= n$, $\bs{V}\in\mathbb{R}^{p\times n}$ is a random
Haar-distributed orthonormal matrix, and $\bs{\varepsilon}\sim N(0,\ident_n)$
independent of $\bs{V}$.  Consider the simple scenario where we are
trying to distinguish between only two possibilities, denoted by
distributions $P_0$ and $P_1$:
\[P_0: \text{ $\bs{Z}$ follows \eqnref{Z_model} with }\theta=0\text{ and }\sigma=1,\text{\quad vs.\quad} P_1: \text{ $\bs{Z}$ follows \eqnref{Z_model} with }\theta=\sigma=\frac{1}{\sqrt{2}}\;.\]
Then for any test $\psi: \mathbb{R}^n\rightarrow \{0,1\}$, the power to correctly distinguish between these two distributions is bounded as
\[\mathbb{P}_{\bs{Z}\sim P_0}\left[\psi(\bs{Z})=1\right] + \mathbb{P}_{\bs{Z}\sim P_1}\left[\psi(\bs{Z})=0\right] \geq 1 -\left(
  \sqrt{\frac{\sum_{i=1}^n (\lambda_i-1)^2}{8}} +
  \sqrt{\frac{n/p}{4\pi}}\right)\;.\]
\end{1}

In other words, every test $\psi$ has high error, with
\[\mathbb{P}(\text{Type I error}) + \mathbb{P}({\text{Type II error}}) \geq 1 - \left(
  \sqrt{\frac{\sum_{i=1}^n (\lambda_i-1)^2}{8}} +
  \sqrt{\frac{n/p}{4\pi}}\right)\;,\] so that if the $\lambda_i$ are
tightly distributed around 1 and $n\ll p$, the problem of estimating
$\rho$, and thus $\theta$, is extremely difficult. Note that for
approximately MP-distributed $\lambda_i$ with $\gamma=n/p \approx 0$,
both $\sum_{i=1}^n (\lambda_i-1)^2$ and $n/p$ are quite small,
explaining the spike in ARE in Figure~\ref{T1T2are} as
$\gamma\rightarrow 0$.

Another way to evaluate how short the EigenPrism CIs are, compared to how
short they could be, is to compare to a Bayesian procedure on a
Bayesian problem. This is done in Section~\ref{regrerror}.

\subsubsection{Computation}
As a procedure intended for use in high-dimensional settings, it is of
interest to know how the EigenPrism procedure scales with large
problem dimensions. There are essentially two parts to the procedure:
the SVD, and the optimization~\eqref{wopt} to choose $\bs{w}^*$. Due
to the strict convexity of the optimization problem, it is extremely
fast to solve \eqref{wopt} and in all of our simulations the runtime
was dominated by the SVD computation. In Appendix~\ref{cvx} we include
a snippet of Matlab code in the popular convex optimization language
CVX \citep{cvx,gb08} that reformulates the optimization
problem~\eqref{wopt} as a second-order cone problem. Even if the
optimization becomes extremely high-dimensional, note that the optimal
weights $\bs{w}^*$ are a smooth function of their associated
eigenvalues $\lambda_i$. Thus we can approximate $\bs{w}^*$ extremely
well by subsampling the $\lambda_i$, computing a lower-resolution
optimal weight vector, and then linearly interpolating to obtain the
higher-resolution, high-dimensional $\bs{w}^*$. For the SVD, note that
$\bs{V}$ never needs to be computed. Thus, the computation scales as
$n^2p$ with a small constant of proportionality, as the SVD of
$\bs{X}\bs{X}^{\top}$ is all that is needed.

\section{Derivative Procedures}
\label{relatedproblems}
In this section, we go into more detail about the three related problems of
performing inference on estimation error of a high-dimensional
regression estimator, noise level in a high-dimensional linear model, and genetic
signal-to-noise ratio, including simulation results. MATLAB code for
the numerical results in this paper is available on the first author's
website.

\subsection{High-Dimensional Regression Error}
\label{regrerror}
We have already shown in Section~\ref{motivation} that the problem of
inference for high-dimensional regression error is equivalent, with a
change of variables, to that of inference on $\theta$. Under
assumptions~\eqref{normassump}, our framework even allows for selection of a subset of
$\hat{\bs{\beta}}$, for instance if the doctor sees an anomaly in a
region of the reconstructed image, he or she may only care about error
in that region. In that case, for a subset of indices $R$ (with
corresponding complement $R^c$), Equation~\eqref{HDreg} can be
rewritten as
\begin{equation*}
\tilde{\bs{y}} = \bs{X}_R^{(1)}(\bs{\beta}_R-\hat{\bs{\beta}}_R) +
\bs{X}_{R^c}^{(1)}(\bs{\beta}_{R^c}-\hat{\bs{\beta}}_{R^c})+
\bs{\varepsilon} = \bs{X}_R^{(1)}(\bs{\beta}_R-\hat{\bs{\beta}}_R) + \tilde{\bs{\varepsilon}},
\end{equation*}
where $\tilde{\bs{\varepsilon}}$ is an i.i.d.~Gaussian vector
independent of $\bs{X}_R^{(1)}$, so that defining
$\theta=\|\bs{\beta}_R-\hat{\bs{\beta}}_R\|_2$ puts this problem
squarely into the EigenPrism framework, \emph{regardless} of the fact
that $R$ may be chosen after observing $\hat{\bs{\beta}}$ (recall that $\hat{\bs{\beta}}$
was fitted on an independent subset of the data, $(\bs{X}^{(0)},\bs{y}^{(0)})$).

We note that
the requirement that the columns of $\bs{X}$ be independent in order
to perform inference on $\|\hat{\bs{\beta}}-\bs{\beta}\|_2^2$ cannot
be relaxed. However, with a known
covariance $\bs{\Sigma}$, one could instead perform inference on
$\|\bs{\Sigma}^{1/2}(\hat{\bs{\beta}}-\bs{\beta})\|_2^2$. Of course,
inference for either $\|\hat{\bs{\beta}}-\bs{\beta}\|_2^2$ or
$\|\bs{\Sigma}^{1/2}(\hat{\bs{\beta}}-\bs{\beta})\|_2^2$ is sufficient if
the ultimate goal is to invert the CI to test a global null
hypothesis on the coefficient vector.

What remains to be seen then is (1) that coverage is not lost by
approximating $\theta^2+\sigma^2$ by $\|\bs{y}\|_2^2$ and by assuming
$T_2$ is normal, and (2) how short the resulting CIs are relative
to how short they could be. To investigate (1), we fixed $p$ at
$10^4$, $\theta^2+\sigma^2 = 10^4$, varied $n$ on a log scale between
0 and $p$, and varied $\rho$ (recall Equation~\ref{rho}) between 0 and
1 by taking equally spaced values of $\log(\rho/(1-\rho))$. Note
that due to rotational symmetry, the direction of $\bs{\beta}$ is
irrelevant. We ran $10^4$ simulations of the EigenPrism procedure to
generate 95\% CIs and compared coverage across the settings in
Figure~\ref{minmaxcoverage}. \rev{We also simulated CIs using
  the results of \cite{Dicker2014} by simply plugging in its
  estimators for $\theta^2$ and $\sigma^2$ to its asymptotic variance
  formula (which depends on the exact parameters).}
\begin{figure}[!t]
  \centering
    \subfigure[]{\label{minmaxcoverage}
      \includegraphics[width=0.45\textwidth]{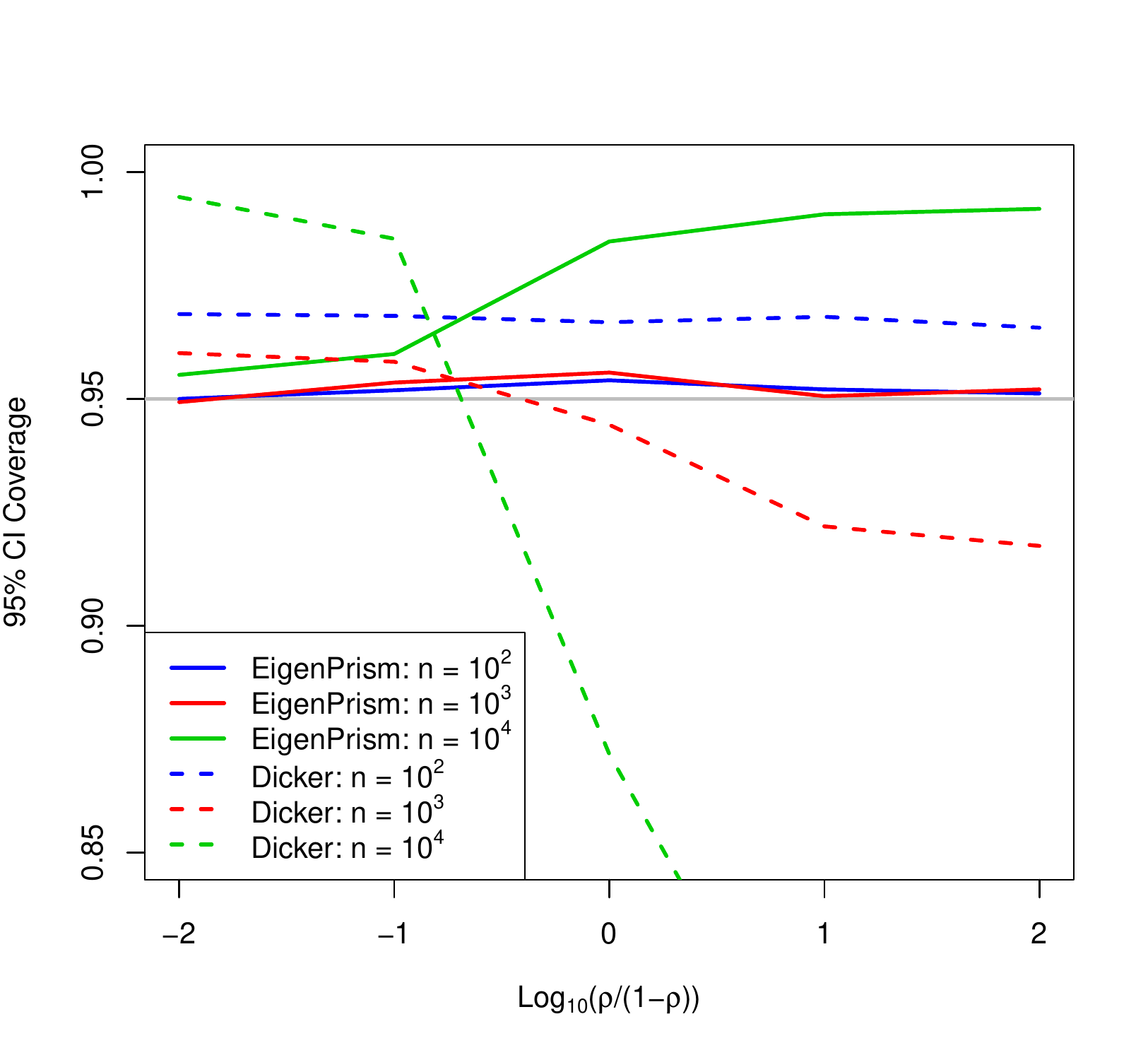}
    }
    \qquad
    \subfigure[]{\label{figbayes}
      \includegraphics[width=0.45\textwidth]{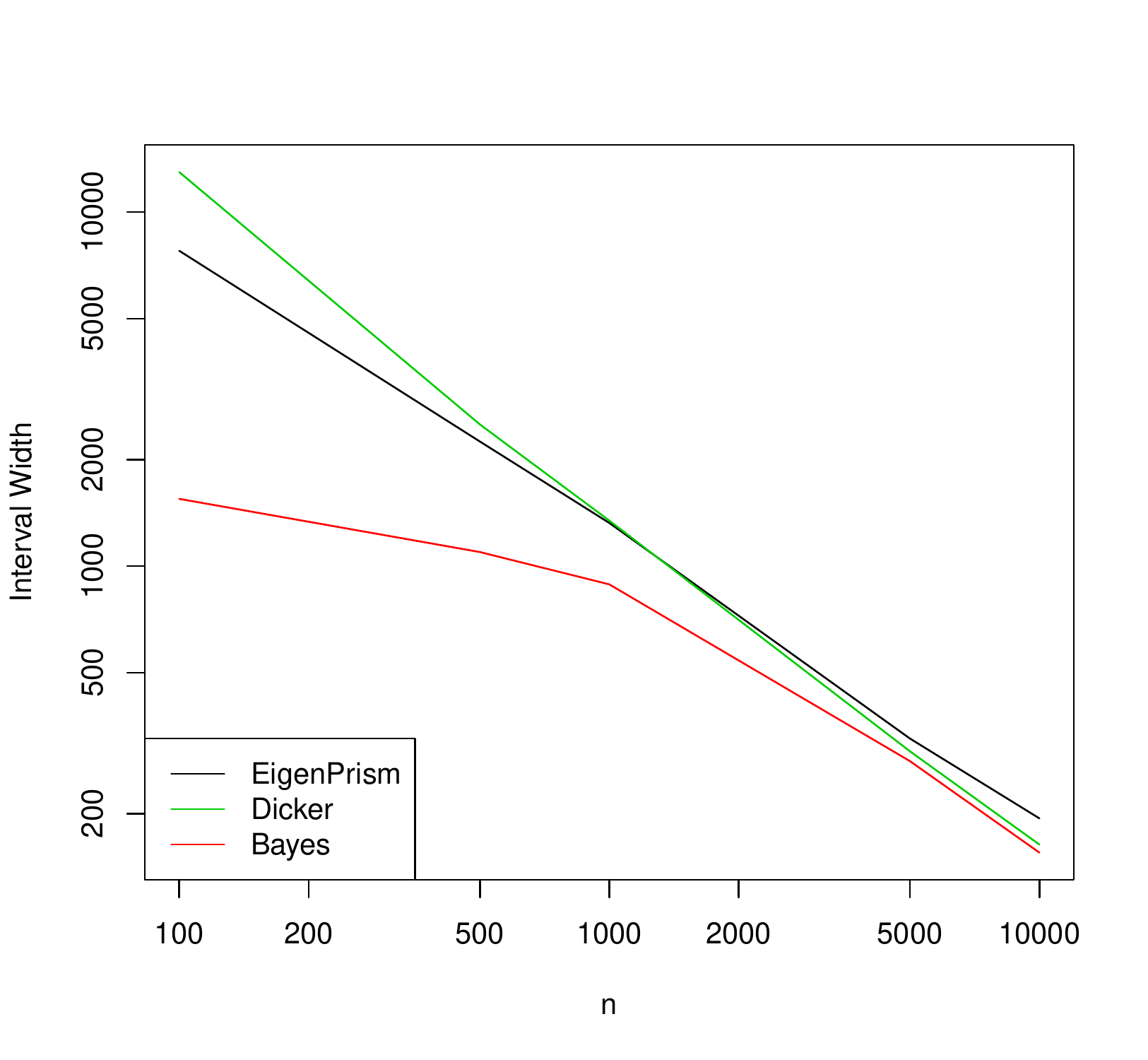}
    }
    \caption{(a) Coverage of 95\% EigenPrism \rev{and Dicker} confidence
      intervals as a function of $\rho$ for $p=10^4$ and
      $\theta^2+\sigma^2=10^4$. (b) EigenPrism confidence interval\rev{,
      Dicker confidence interval,} and
      Bayes credible interval widths as a function of $n$ for $p=10^4$
      and $\bs{\beta}$ sampled according to the Bayesian model in
      Equation~\eqref{bayes}.}
\label{coveragewidth}
\end{figure}
Note that the \rev{EigenPrism} CIs achieve at least nominal coverage
in all cases\rev{, while the Dicker procedure is less reliable, especially
for large $\rho$}. One
setting in which we see \rev{EigenPrism} over-cover is when $n \approx p$ and
$\rho\approx 1$. This can be explained by the variance upper-bound for
$T_2$ in Equation~\eqref{varUB1}, which is tight when $p$ is large and
$(p\sum_{i=1}^n(w^*_i\lambda_i)^2)^{-1} \ll
1$. Figure~\ref{tightBound} shows that, except when $n\approx p$, we
indeed have $(p\sum_{i=1}^n(w^*_i\lambda_i)^2)^{-1} \ll 1$.
\begin{figure}\centering
\includegraphics[width=0.45\textwidth]{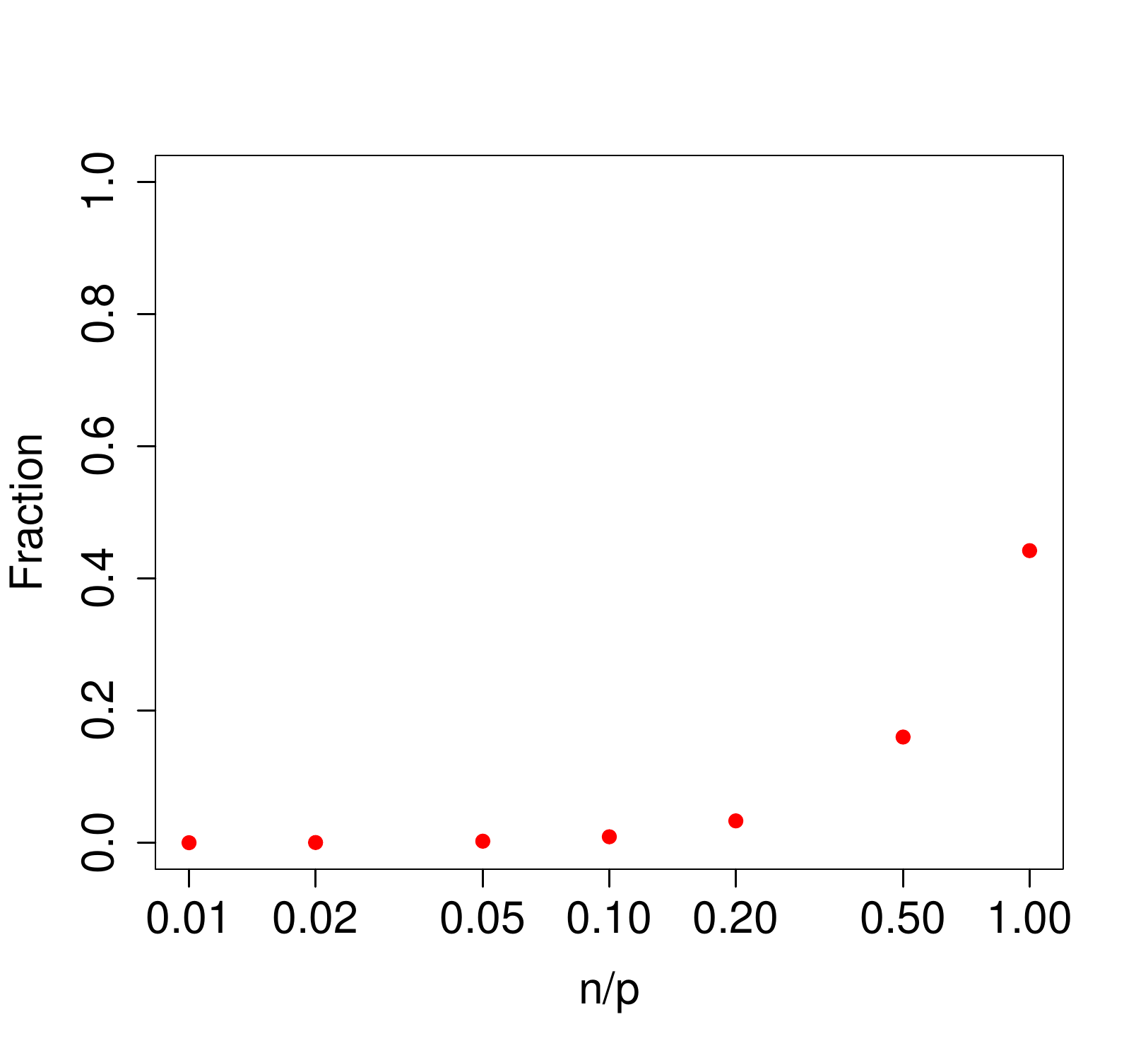}
\caption{Plot of the fraction $\frac{1}{p\sum_{i=1}^n(w^*_i\lambda_i)^2}$ as a
  function of $n/p$ (on the log scale) for $p=10^4$.}
\label{tightBound}
\end{figure}

To investigate (2), we simulated the EigenPrism procedure on a Bayesian
model and compared the EigenPrism widths to those obtained by computing equal-tailed
Bayes credible intervals (BCI) from a Gibbs-sampled posterior. The
details of the Bayesian setup are given in Appendix~\ref{appbayes},
but the resulting CI widths are summarized in Figure~\ref{figbayes}
for $p = 10^4$ and a range of $n$. \rev{Again, we also compared to
  Dicker CIs.} Each point on the plot represents 1000
simulations. \rev{Although the Dicker CIs become slightly shorter than
  EigenPrism's for large $n$, we note (as evidenced by Figure~\ref{minmaxcoverage})
  that this is exactly the regime in which the Dicker CIs have
  unreliable coverage.
We will see later in Section~\ref{robust} that even for small $n$ and $\rho$, the
Dicker CIs quickly lose coverage as correlations are added to the
design matrix, while EigenPrism's coverage is in fact quite robust. 
} The \rev{other}
salient features of this plot are that the EigenPrism CI widths decrease
at a steady $\sqrt{n}$-rate, while the BCI widths start much
lower and appear to asymptote around the EigenPrism CI width curve. The
fact that the BCI widths are much shorter for small $n$ can be
explained by the information contained in the priors, which is
important for two reasons. In any frequentist-Bayesian comparison of
methods, there is always the phenomenon that small $n$ means the data
contains little information, so the prior information given to the
Bayesian method makes it heavily favored over the frequentist
method. However, as we saw in Section~\ref{sec:width}, the
frequentist problem is fundamentally limited not just by $n$ but by
$\sd(\lambda_i)$ as well, and here since $p$ is fixed, small $n$
corresponds to small $\sd(\lambda_i)$ as well, adding an extra layer
of challenge for the EigenPrism procedure. As $n$ increases though, the BCIs rely more
heavily on the data, and come much closer in width to the EigenPrism CIs,
with the \rev{average} relative width increase bottoming-out at about 5\% for $n =
5000$. The relative uptick in the EigenPrism CI widths for $n
\approx p$ can again be explained by the upper-bound in Equation~\eqref{varUB1}.

\subsection{Inference on $\sigma^2$}
\label{estsigma2}
We can use almost exactly the same EigenPrism procedure for $\sigma^2$
as we did for $\theta^2$. Recall Equation~\eqref{condexp},
\begin{equation*}
\Ec{S}{\bs{d}} = \theta^2\sum_{i=1}^nw_i\lambda_i + \sigma^2\sum_{i=1}^nw_i.
\end{equation*}
To make $S$ unbiased for $\theta^2$, we constrained
$\sum_{i=1}^nw_i=0$ and $\sum_{i=1}^nw_i\lambda_i=1$. However by
switching these linear constraints, so that $\sum_{i=1}^nw_i=1$ and
$\sum_{i=1}^nw_i\lambda_i=0$, we make $S$ unbiased for $\sigma^2$. The
variance formulae and upper-bounds in
Equations~\eqref{varExact}--\eqref{varUB} still hold, so that we can
construct $T_3$ (and an associated CI). Let $\bs{w}^{**}$ be the
solution to the following convex optimization program $\mathcal{P}_2$:
\begin{equation*}
\argmin_{\bs{w}\in\mathbb{R}^n} \; \max \left( \sum_{i=1}^n w_i^2, \sum_{i=1}^n
  w_i^2 \lambda_i^2 \right) \quad
\text{such that}  \, \sum_{i=1}^n w_i = 1, \; \sum_{i=1}^n w_i \lambda_i = 0
\end{equation*}
and denote by $\val(\mathcal{P}_2)$ the minimized objective function
value. Then the EigenPrism procedure for performing inference on
$\sigma^2$ reads
\begin{equation*}
\begin{split}
T_3 :=&\, \sum_{i=1}^{n} w^{**}_i z_i^2, \\
\Ec{T_3}{\bs{d}} =&\, \sigma^2, \\
\sd(T_3|\bs{d}) \lesssim&\, \sqrt{2\val(\mathcal{P}_2)}\, \frac{\|\bs{y}\|_2^2}{n}, \\
\end{split}
\end{equation*}
where again, the only approximation in the variance is the replacement of
$\theta^2+\sigma^2$ by its estimator $\|\bs{y}\|_2^2/n$. \rev{The analogue
to Theorem~\ref{T2normal} holds and is proved in
Appendix~\ref{asymptoticnormality}:
\[\frac{T_3-\sigma^2}{\sd(T_3|d)} \stackrel{d}{\longrightarrow} N(0,1).\]}
Finally, as before, we construct the lower and upper endpoints to obtain an approximate
$(1-\alpha)$-CI for $\sigma^2$ via
\[
L_{\alpha} := \max\left(T_3 - z^\star_{1-\alpha/2}\cdot  \sqrt{2\val(\mathcal{P}_2)} \, \frac{\|\bs{y}\|_2^2}{n}, 0\right)\;,
\quad
U_{\alpha}:=T_3 + z^\star_{1-\alpha/2}\cdot  \sqrt{2\val(\mathcal{P}_2)} \, \frac{\|\bs{y}\|_2^2}{n}\;.
\]

 Note that if
the columns of $\bs{X}$ have a known covariance matrix $\bs{\Sigma}$,
the exact same machinery goes through by replacing $\bs{X}$ by
$\bs{X}\bs{\Sigma}^{-1/2}$ and replacing $\bs{\beta}$ by
$\bs{\Sigma}^{1/2}\bs{\beta}$.

Turning to simulations, we aim to show that the EigenPrism CIs for
$\sigma^2$ have at least nominal coverage. We take the same setup as
in Figure~\ref{minmaxcoverage} but instead construct 95\% CIs for
$\sigma^2$. Figure~\ref{sg2minmaxcoverage} shows the result, and as
before we see that \rev{EigenPrism's} coverage never dips below nominal levels in any
of the settings\rev{, while for small $\rho$ the Dicker CI's coverage
  can be unreliable, especially for large $n$}.
\begin{figure}[!t]
  \centering
  \includegraphics[width=0.45\textwidth]{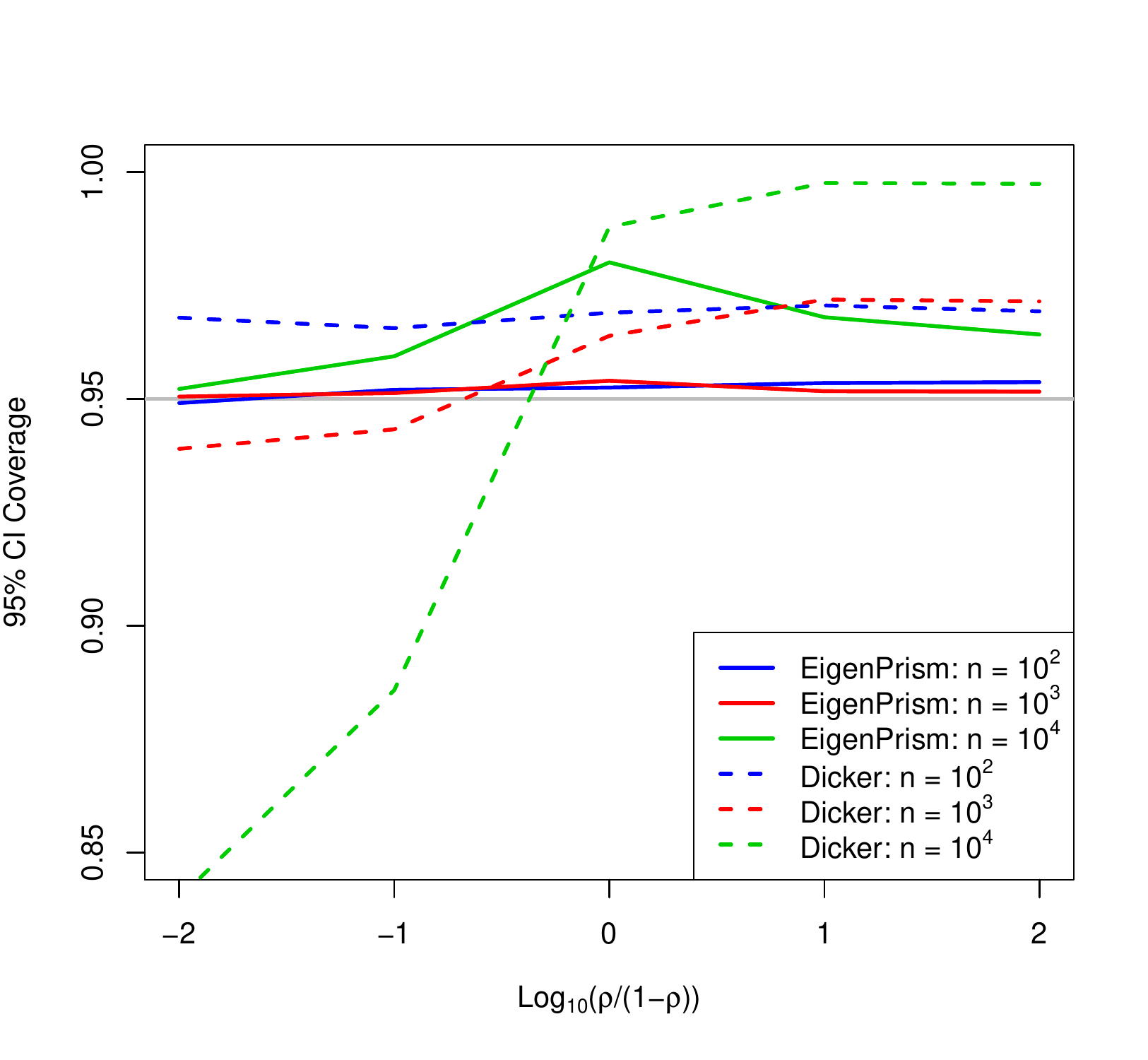}
  \caption{Coverage of 95\% EigenPrism \rev{and Dicker} confidence
    intervals for $\sigma^2$ as a function of $\rho$ for $p=10^4$ and
    $\theta^2+\sigma^2=10^4$. Each point represents $10^4$ simulations, and the grey
    line denotes nominal coverage.}
\label{sg2minmaxcoverage}
\end{figure}
\begin{figure}[!t]
  \centering
    \subfigure[]{\label{sg2SL_RCV}
      \includegraphics[width=0.45\textwidth]{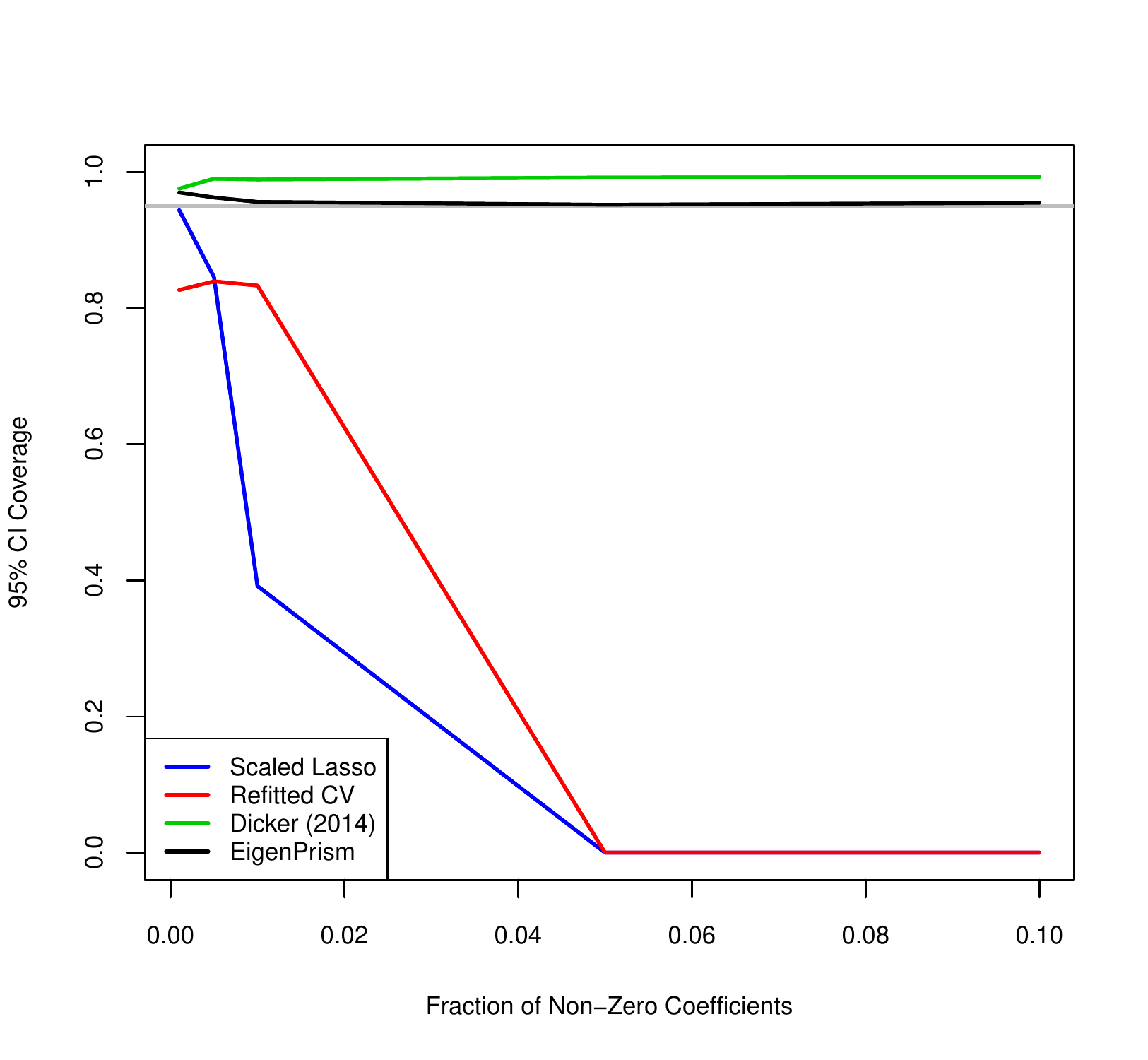}
    }
    \qquad
    \subfigure[]{\label{sg2SL_RCV_widths}
      \includegraphics[width=0.45\textwidth]{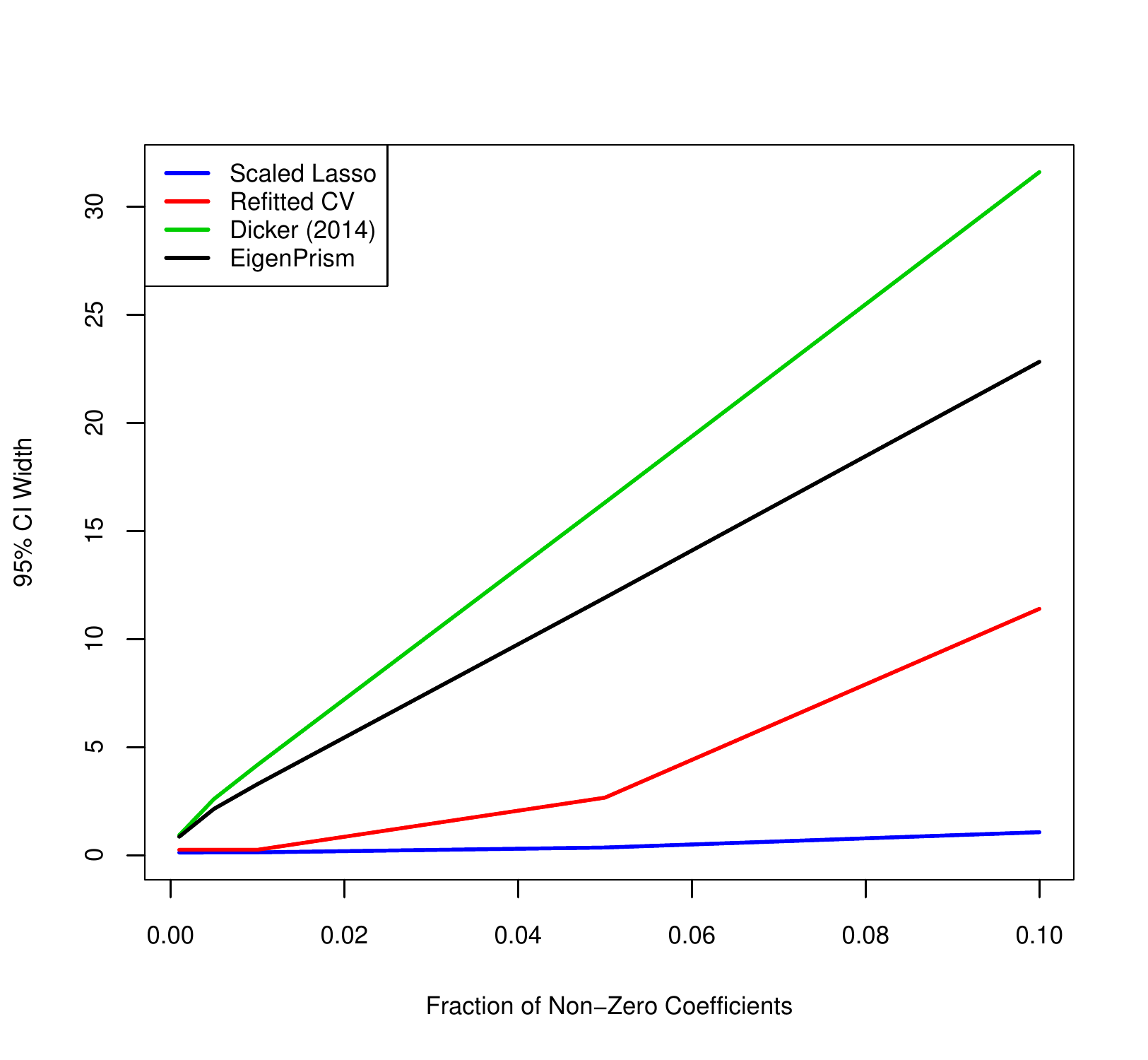}
    }
    \caption{(a) Coverage and (b) width of scaled lasso, refitted
    cross validation, \rev{plug-in CI from \cite{Dicker2014} described
      in Section~\ref{regrerror},} and EigenPrism confidence intervals when $n=500$,
    $p=1000$, $\sigma^2=1$, and non-zero entries of $\bs{\beta}$ equal
    to 1. Each point represents $10^4$ simulations, and the grey
    line denotes nominal coverage.}
\label{sg2_compare}
\end{figure}
We performed a similar experiment with a Bayesian model to compare
EigenPrism CI widths for $\sigma^2$ with those of equal-tailed BCIs, but
found a less-desirable comparison than in the $\theta$
case. In particular, the most favorable simulations showed the EigenPrism
CI approximately 30\% wider than the BCI, which can likely be attributed to the more-informative prior (Inverse Gamma) on $\sigma^2$
than that on $\theta^2$ (nearly Exponential) in the Bayesian
model \eqref{bayes}. Although we would have liked to try an
Exponential prior for $\sigma^2$, due to a lack of conjugacy the
resulting Gibbs sampler was computationally intractable. We note that
except in special cases, it can be very computationally challenging to
construct BCIs, especially in high dimensions.

We point out that only two other $\sigma^2$ estimators in the
literature provide any inference results, namely the scaled Lasso
\citep{Sun2012} and the refitted cross validation (CV) method of
\cite{Fan2012}. In particular, under some sparsity conditions on the
coefficient vector, the authors find aymptotic normal approximations
to their estimators. To compare our CIs with theirs, we compared them
on the same simulations, but quickly found that scaled Lasso and
refitted CV CIs only achieve nominal coverage in extremely sparse
settings. \rev{We also compared the plug-in CI for the estimator in \cite{Dicker2014}. This coverage comparison is shown in Figure~\ref{sg2SL_RCV}.} 
The scaled Lasso CIs
only achieve nominal coverage when 1 out of the 1000 coefficients are
non-zero, and quickly drop off to less than half of nominal coverage
by 1\% sparsity. The refitted CV CIs undercover by about 10\% even in
the sparsest settings, and also fall off further in coverage as
sparsity decreases. The EigenPrism \rev{and Dicker} CIs achieve at least nominal
coverage at all sparsity levels
examined. \rev{Figure~\ref{sg2SL_RCV_widths} shows average CI widths
  for the same simulations. The much smaller widths of the scaled Lasso and
  refitted CV CIs align with their lack of coverage, reflecting the
  fact that the bias and variance of their
  estimators can be poorly characterized in finite samples. The
  Dicker CIs are consistently wider than EigenPrism's, with
  the inflation factor nearly 40\% at the right-hand side of the plot.}

\subsection{Genetic Variance Decomposition}
\label{heritability}
Consider a linear model for a centered continuous phenotype ($y_i$)
such as height, as a function of a centered SNP array ($\bs{x}_i$). The variance
can be decomposed as
\begin{equation}
\label{vardecomp}
\E{y_i^2} = \mathbb{E}\left[(\bs{x}_i^{\top}\bs{\beta})^2\right] + \sigma^2.
\end{equation}
Under linkage disequilibrium, assuming
column-independence is unrealistic. However, a wealth of genomic data
has resulted in this column dependence possibly being estimable from
outside data sets (e.g. \cite{Abecasis2012}), so
we may instead take $\bs{x}_i\iid N(\bs 0,\bs{\Sigma})$ with
$\bs{\Sigma}$ known (we will discuss a relaxation of the normality in Section~\ref{robust}). Then Equation~\eqref{vardecomp} reduces to
\begin{equation*}
\E{y_i^2}= \|\bs{\Sigma}^{1/2}\bs{\beta}\|_2^2 + \sigma^2,
\end{equation*}
which provides a formula for the linear model's signal-to-noise ratio,
\begin{equation*}
\textsc{SNR} =
\frac{\|\bs{\Sigma}^{1/2}\bs{\beta}\|_2^2}{\|\bs{\Sigma}^{1/2}\bs{\beta}\|_2^2 +
  \sigma^2}.
\end{equation*}
The \textsc{SNR} is connected to the genetic heritability in that, for
the simplified approximation to a linear model with additive
i.i.d. noise, it quantifies what fraction of a continuous phenotype's
variance can be explained by SNP data. We note that there are many
different definitions of heritability, and the
\textsc{SNR} aligns most closely with the narrow-sense, or additive,
heritability, as we do not allow for interactions or dominance
effects. The extent of the connection between the two definitions
depends on how complete the SNP array is---if every SNP is measured,
they correspond exactly.

Although until now we have been working with $\|\bs{\beta}\|_2^2$,
while the \textsc{SNR} estimation problem seems to call for
$\|\bs{\Sigma}^{1/2}\bs{\beta}\|_2^2$, the above problem turns out to
fit right into our framework. Explicitly, the linear model can be
rewritten as
\begin{equation*}
\bs{y} = \left(\bs{X}\bs{\Sigma}^{-1/2}\right)\left(\bs{\Sigma}^{1/2}\bs{\beta}\right) + \bs{\varepsilon},
\end{equation*}
where now the rows of $(\bs{X}\bs{\Sigma}^{-1/2})$ are
i.i.d.~$N(\bs{0},\bs{I})$, and $\theta^2$ corresponds to the new
quantity of interest: $\|\bs{\Sigma}^{1/2}\bs{\beta}\|_2^2$. Since
$\E{\|\bs{y}\|_2^2/n} = \|\bs{\Sigma}^{1/2}\bs{\beta}\|_2^2 +
\sigma^2$ now, applying our methodology to $\bs{X}\bs{\Sigma}^{-1/2}$
gives a natural estimate for \textsc{SNR}, namely,
\begin{equation*}
\widehat{\textsc{SNR}} := \frac{T_2}{\|\bs{y}\|_2^2/n}.
\end{equation*}
Continuing, as we have done throughout this paper, to treat
$\|\bs{\Sigma}^{1/2}\bs{\beta}\|_2^2 + \sigma^2$ as if it is known and
equal to $\|\bs{y}\|_2^2/n$, our distributional results for $T_2$
extend to give us an approximate confidence interval for
\rev{$\widehat{\textsc{SNR}}$}.

We turn again to simulations to demonstrate the performance of the
EigenPrism procedure described above for constructing \textsc{SNR}
CIs. One major consideration is that of course, SNP data is discrete,
not Gaussian. However, we will show in Section~\ref{robust} that the
EigenPrism procedure works well empirically even under non-Gaussian
marginal distributions. Here, we run experiments for $n=10^5$, $p =
5\times10^5$, $\theta^2+\sigma^2 = 10^4$, Bernoulli(0.01) design with
independent columns, $\bs{\beta}$ having 10\% non-zero entries, and
$\textsc{SNR}$ varying from nearly 0 to nearly
1. Figure~\ref{fig:heritability} shows the EigenPrism CI coverage and
average widths.
\begin{figure}[!t]
  \centering
    \subfigure[]{\label{heritcover}
      \includegraphics[width=0.45\textwidth]{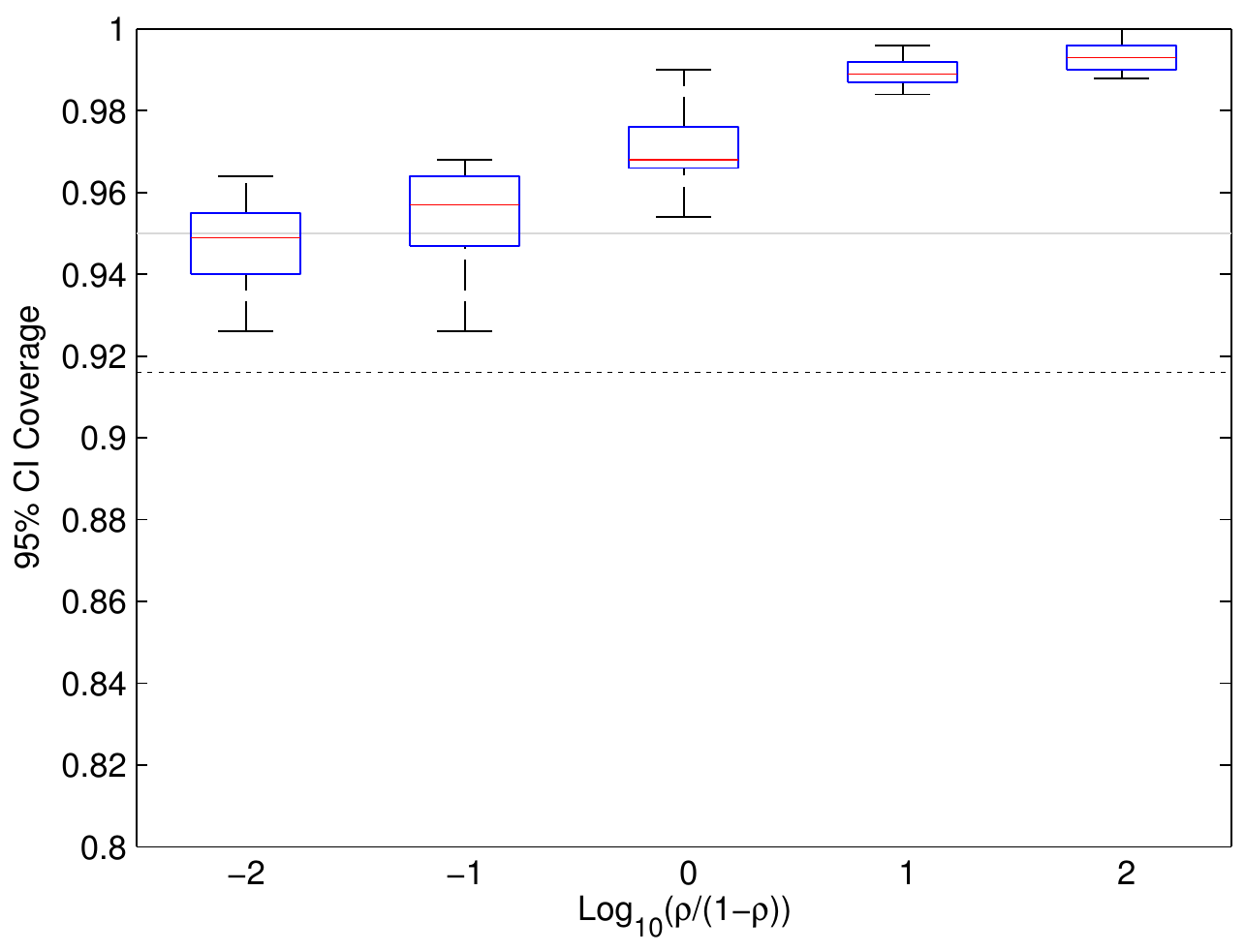}
    }
    \qquad
    \subfigure[]{\label{heritwidths}
      \includegraphics[width=0.45\textwidth]{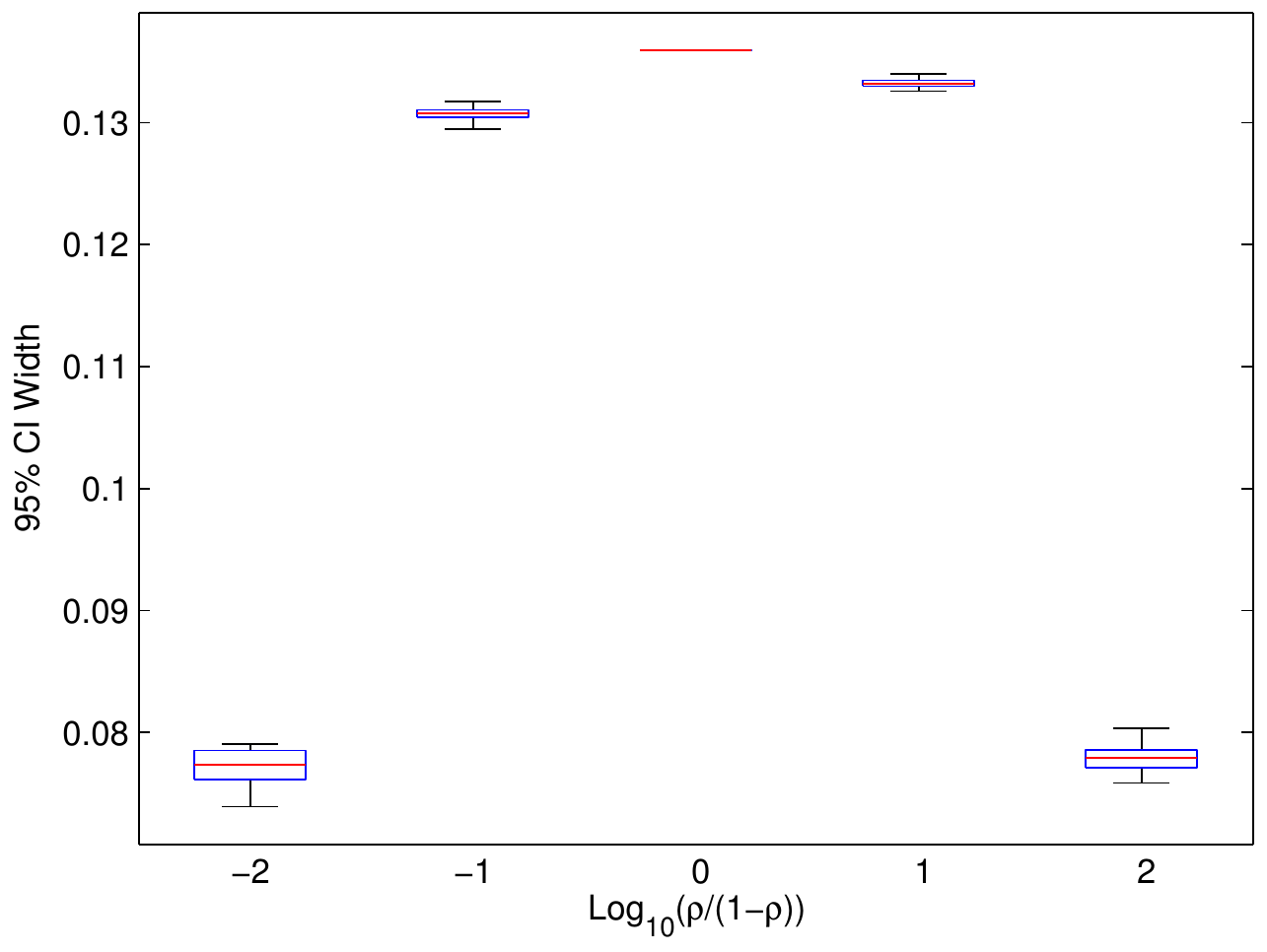}
    }
    \caption{(a) Coverage, and (b) average widths, of EigenPrism
      \textsc{SNR} 95\% confidence intervals. Experiments used
      $n=10^5$, $p = 5\times10^5$, $\theta^2+\sigma^2 = 10^4$,
      Bernoulli(0.01) design, and $\bs{\beta}$ with 10\% non-zero, Gaussian
      entries. Each boxplot summarizes the (a) coverage and (b) width
      for 20 different $\bs{\beta}$'s, each of which is estimated with
      500 simulations. The whiskers of the boxplots extend to the
      maximum and minimum points. The black dotted line in (a) is the
      95\% confidence lower-bound for the lowest whisker in each plot
      assuming all CIs achieve exact coverage, and the grey line shows nominal
      coverage.}
\label{fig:heritability}
\end{figure}
Note that although our CIs are conservative, we never lose coverage,
and at worst our 95\% CI would give the \textsc{SNR} to within an
error of $\pm 6.\rev{8}\%$.

\section{Robustness and 2-Step Procedure}
\label{minmaxexper}
In this section we follow up our investigation of the EigenPrism
framework by considering its robustness to model misspecification and
presenting a 2-step procedure that can improve the CI widths of the
vanilla EigenPrism procedure.

\subsection{Robustness}
\label{robust}
An important practical question is how robust the EigenPrism CI is to
model misspecification. In particular, our theoretical calculations
made some fairly stringent assumptions, and we explore here their
relative importances. Some standard assumptions
that we rely
on are that the model is indeed linear and the noise is
i.i.d.~Gaussian and independent of the design matrix. These
assumptions are all present, for instance, in OLS theory, and we
assume that problems substantially deviating from satisfying them are
not appropriate for our procedure. As explained in
Section~\ref{motivation}, the random design assumption is necessitated
by the high-dimensionality ($p>n$) of our problem, and within the
random design paradigm, the assumption of i.i.d.~rows is still broadly
applicable, for instance whenever the rows represent samples drawn
independently from a population.

The not-so-standard assumption we make is that the columns of $\bs{X}$
are also independent, and all of $\bs{X}$'s entries are $N(0,1)$ (note
that each column of a real design matrix can always be standardized so
that at least the first two marginal moments match this
assumption). These assumptions are important because they ensure that
the columns of $\bs{V}$ are uniformly distributed on the unit sphere,
so that we can characterize both the expectation and variance of their
inner product with $\bs{\beta}$. Although we will see that the
marginal distribution of the elements of $\bs{X}$ is not very
important as long as $n$ and $p$ are not small, in general the
independence of the columns is crucial. We note that there is work in
random matrix theory showing that for certain random matrices which
are not i.i.d.~Gaussian, the eigenvectors are still in some sense
asymptotically uniformly distributed on the unit sphere (see for
example \cite{bai2007}). This suggests that EigenPrism CIs, at
least asymptotically, may work well in a broader context than shown so
far.

Before explaining further,
we feel it is important to recall that for two of the three inference
problems this work addresses (inference for $\sigma^2$ and
signal-to-noise ratio), the EigenPrism procedure extends to easily account for any
known covariance matrix among the columns of
$\bs{X}$. However in the vanilla example of simply constructing CIs for $\theta^2$,
correlation among the columns of $\bs{X}$ can cause serious
problems. To first order, we need
$\Ec{\|\bs{V}^{\top}\bs{\beta}\|_2^2}{\bs{d}} \approx \rev{\theta^2}n/p$, or else $T_2$
wil\rev{l} be biased and the resulting shifted interval will have poor
coverage. From a practical perspective, unless $\bs{\beta}$ is
adversarially chosen, it may seem unlikely that $\bs{\beta}$ will be
particularly aligned or misaligned (orthogonal) to the directions in
which $\bs{X}$ varies. In particular, if we make a random effects
assumption and say that the entries of $\bs{\beta}$ are
i.i.d.~$N(0,\tau^2)$, then the EigenPrism procedure will achieve nominal
coverage. A slightly more subtle problem occurs if $\bs{\beta}$ is
chosen not adversarially, but sparse in the basis of $\bs{X}$'s
principal components. In this case, although $T_2$ is approximately
unbiased, the variance estimate could be far too small, resulting
again in degraded coverage.

\begin{figure}[!htbp]
  \centering
    \subfigure[]{\label{moments_dense}
      \includegraphics[trim=0cm 0.16cm 0cm 0.3cm, clip=true, width=0.45\textwidth]{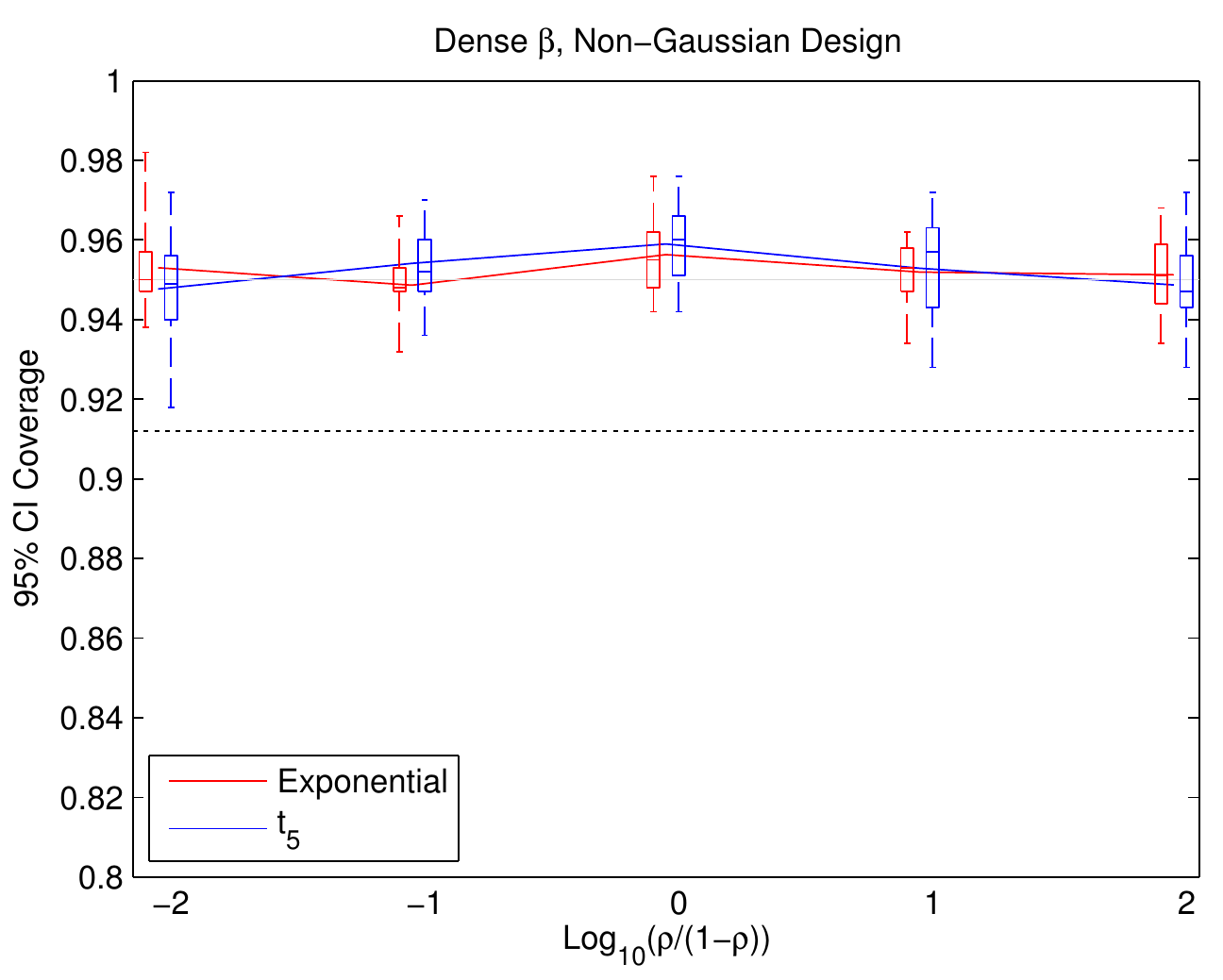}
    }
    \qquad
    \subfigure[]{\label{moments_sparse}
      \includegraphics[trim=0cm 0.16cm 0cm 0.3cm, clip=true, width=0.45\textwidth]{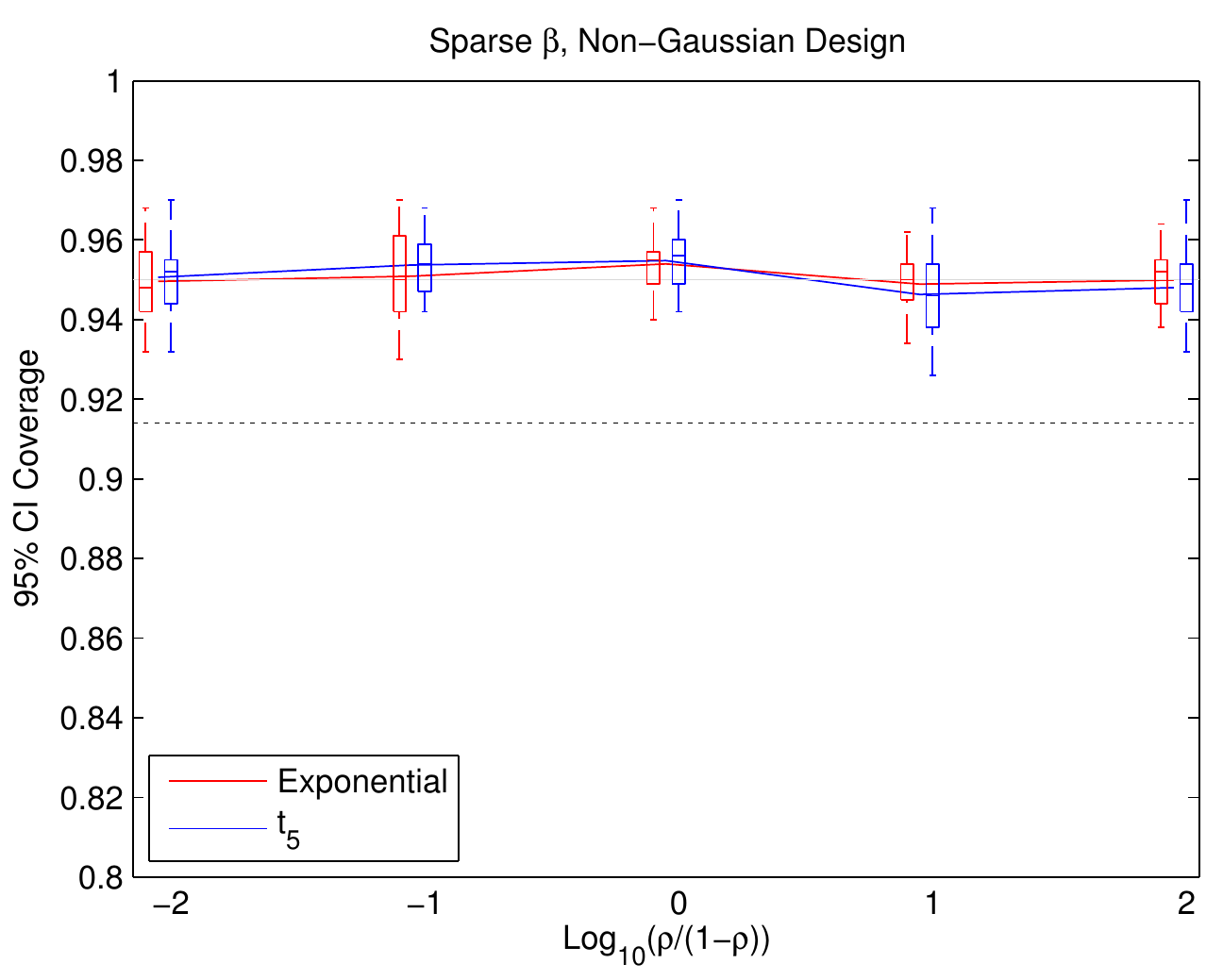}
    }
    \subfigure[]{\label{sparse_dense}
      \includegraphics[trim=0cm 0.16cm 0cm 0.3cm, clip=true, width=0.45\textwidth]{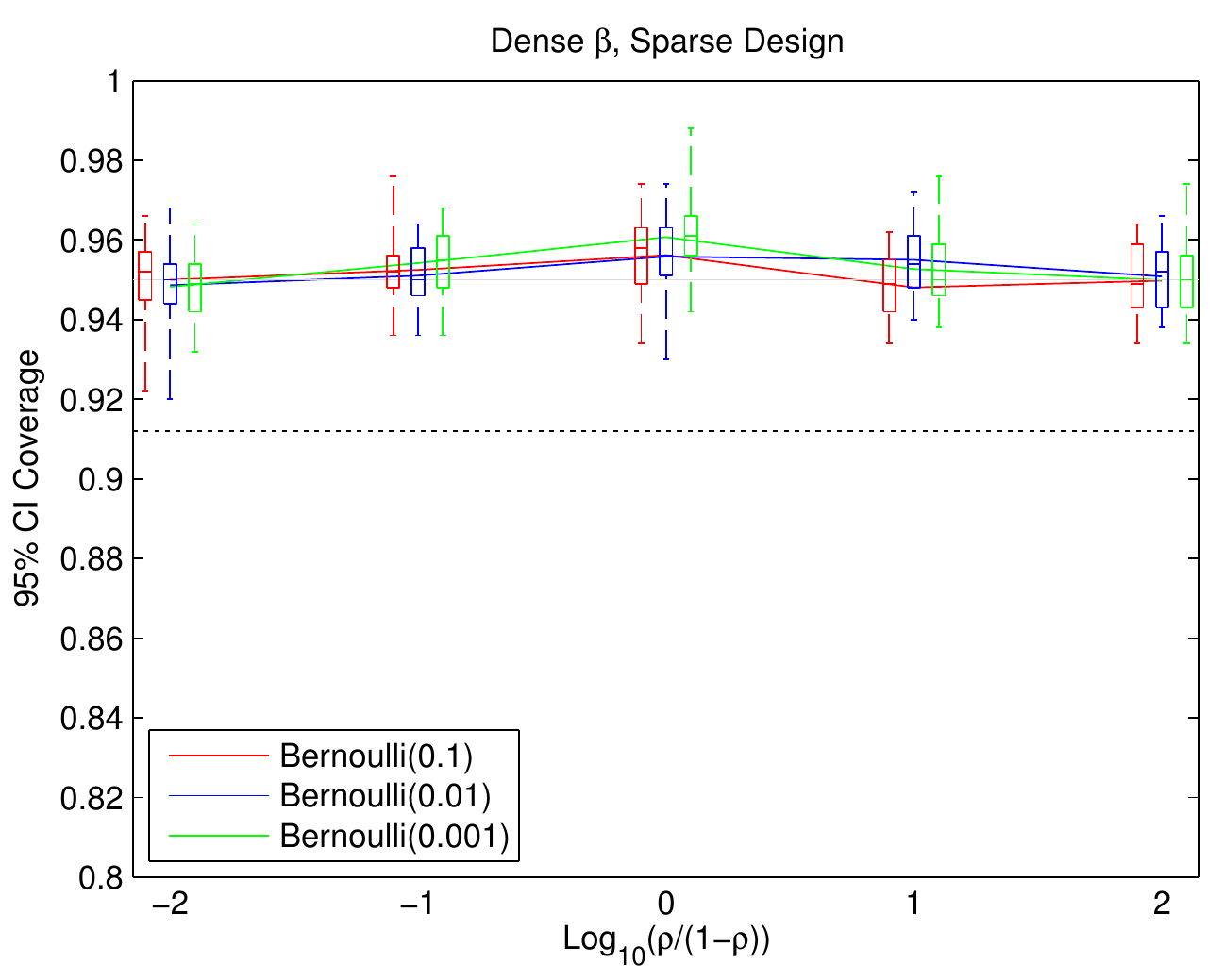}
    }
    \qquad
    \subfigure[]{\label{sparse_sparse}
      \includegraphics[trim=0cm 0.16cm 0cm 0.3cm, clip=true, width=0.45\textwidth]{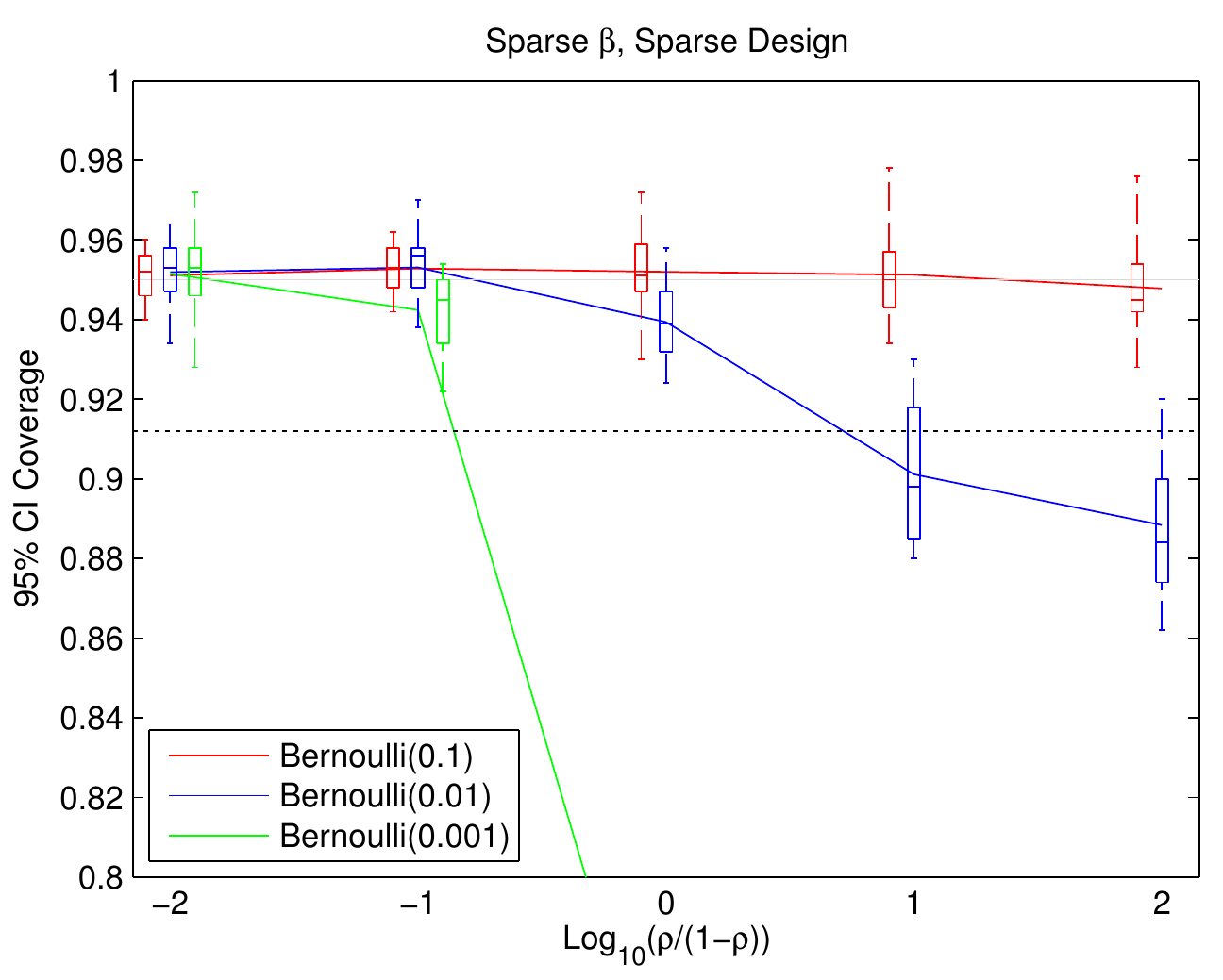}
    }
    \subfigure[]{\label{correl_dense}
      \includegraphics[trim=0cm 0.16cm 0cm 0.3cm, clip=true, width=0.45\textwidth]{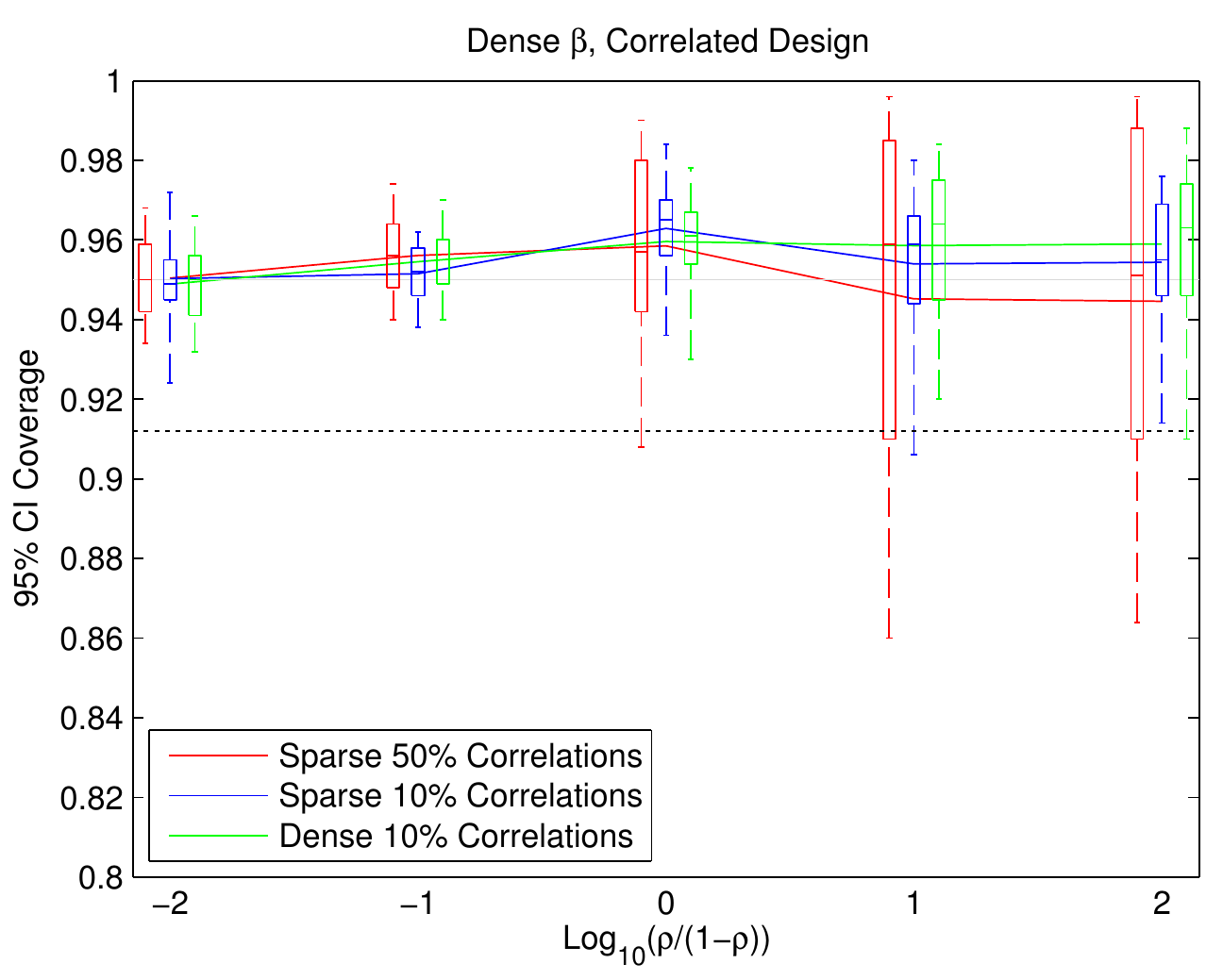}
    }
    \qquad
    \subfigure[]{\label{correl_sparse}
      \includegraphics[trim=0cm 0.16cm 0cm 0.3cm, clip=true, width=0.45\textwidth]{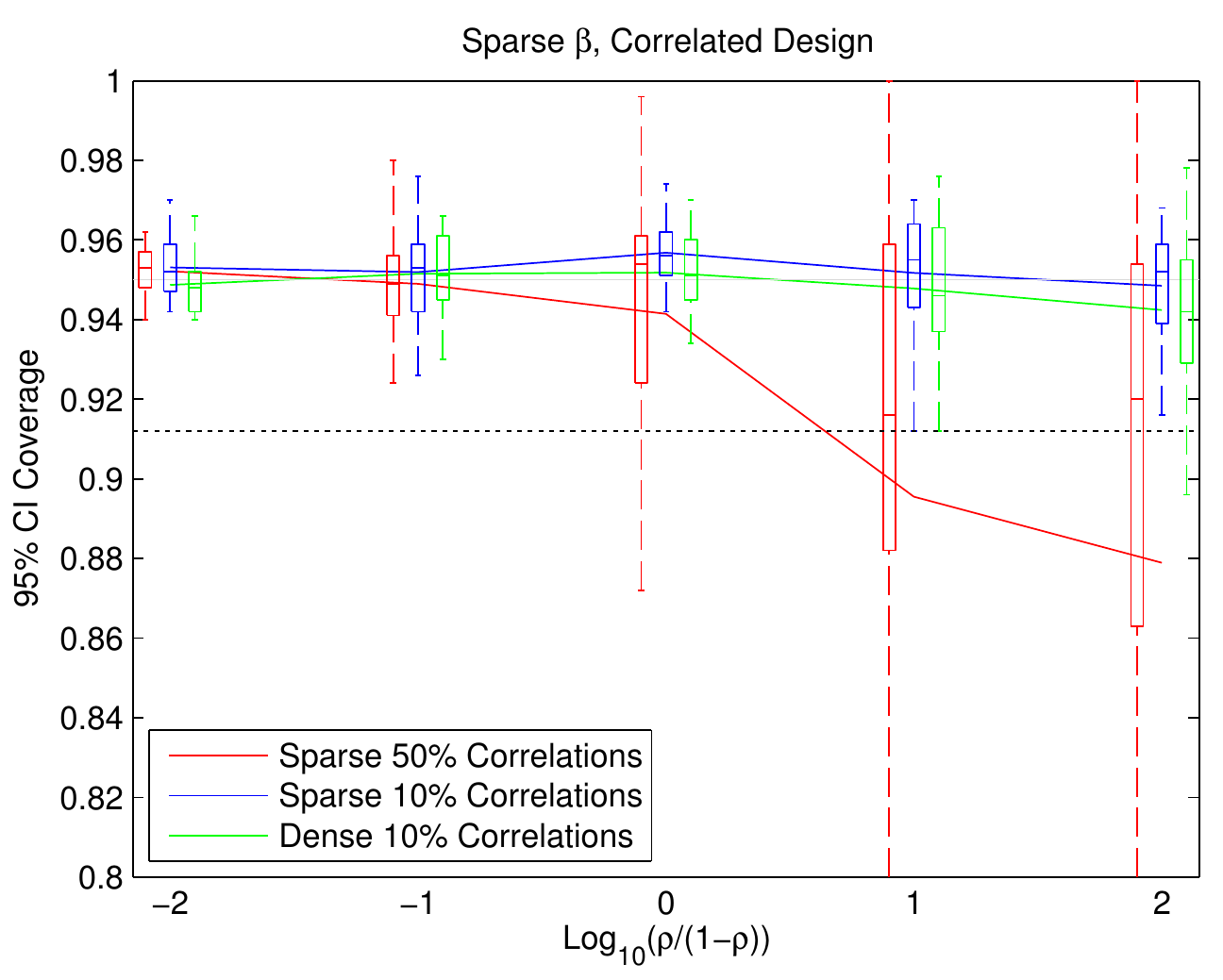}
    }
    \caption{The first column of
      plots ((a), (c), (e)) generates $\beta_i\iid N(0,1)$ and
      renormalizes to control $\theta^2$, while the
      second column of plots ((b), (d), (f)) does the same but then
      sets 99\% of the $\beta_i$ to zero before renormalizing. The
      first two rows of plots ((a), (b), (c), (d)) use $X_{ij}$ i.i.d.~from
      some non-Gaussian distribution renormalized to have mean 0 and
      variance 1. The third row of plots ((e), (f))
      uses marginally standard Gaussian $\bs{X}$ but with correlations
      among the columns; see Appendix~\ref{corrMat} for detailed constructions.
      See text for detailed boxplot constructions and interpretation
      of the dashed line.}
\label{robustness}
\end{figure}
To investigate how wrong the model has to be to make our CIs
undercover, we construct EigenPrism CIs on data coming from models not
satisfying our assumptions. In particular, we ran the EigenPrism
procedure on design matrices with either i.i.d.~entries with very
different higher-order moments than a Gaussian, i.i.d.~entries that
were sparse, or Gaussian entries and correlated columns. Since the
direction of $\bs{\beta}$ becomes relevant in all these cases, we
performed experiments with both dense and sparse $\bs{\beta}$, and in
each regime measured coverage for 20 different $\bs{\beta}$'s. The
results of simulations with $n=10^3$, $p=10^4$, and
$\theta^2+\sigma^2=10^4$ are plotted in
Figure~\ref{robustness}. Each
boxplot summarizes the coverage for 20 different $\bs{\beta}$'s, each
of which is estimated with 500 simulations. The whiskers of the
boxplots extend to the maximum and minimum points, and the black
dotted line is the 95\% confidence lower-bound for the lowest whisker
in each plot assuming all CIs achieve exact coverage. As can be seen
from Figures~\ref{moments_dense} and \ref{sparse_dense}, when
$\bs{\beta}$ is dense, the marginal moments and sparsity of the
entries of $\bs{X}$ do not affect coverage. Figures~\ref{correl_dense}
and \ref{correl_sparse} show that even small unaccounted-for
correlations among the columns of $\bs{X}$ do not greatly affect
coverage, although larger correlations, as expected, can result in
serious undercoverage for certain
$\bs{\beta}$'s. \rev{As a comparison, we also simulated the Dicker CIs
  in the setting of Figures~\ref{correl_dense} and
  \ref{correl_sparse}, wherein coverage never exceeded 40\% for
  any $\bs{\beta}$ or correlation structure.} Figures~\ref{moments_sparse} and \ref{sparse_sparse}
show that when $\bs{\beta}$ is sparse, coverage is much more sensitive
to sparsity in $\bs{X}$, although if $\bs{X}$ is not sparse, coverage
remains robust to higher-order moments of the design
matrix. Figure~\ref{qq} demonstrates the crucial difference when
$\bs{X}$ is sparse by showing a few realizations of quantile-quantile
plots comparing the distribution of the entries of $\bs{V}_1$ to a
Gaussian distribution, for Bernoulli(0.1)- and
Bernoulli(0.001)-marginally-distributed $\bs{X}$.
\begin{figure}[!t]
  \centering
    \subfigure[]{\label{qqnormt3}
      \includegraphics[width=0.45\textwidth]{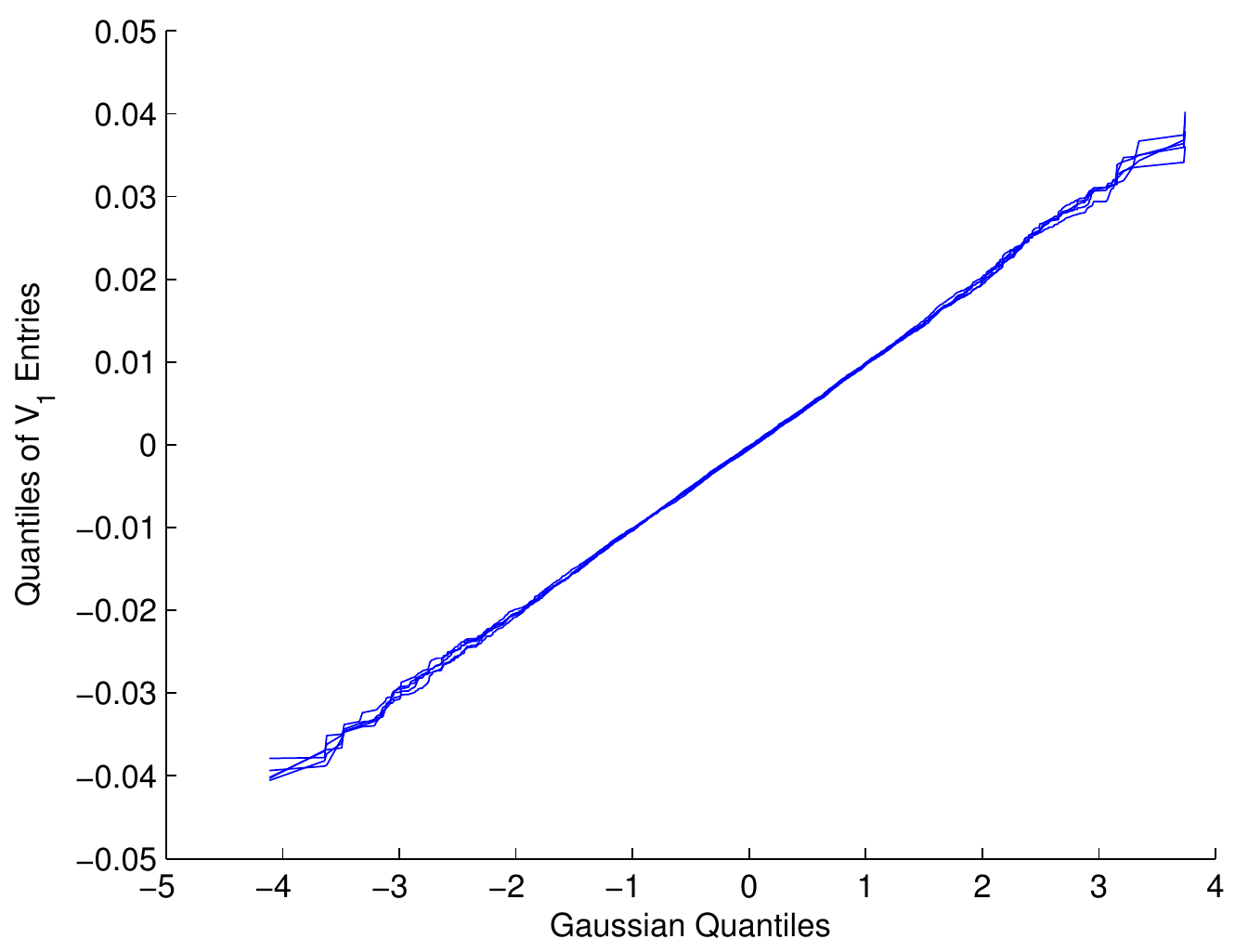}
    }
    \qquad
    \subfigure[]{\label{qqnormt5}
      \includegraphics[width=0.45\textwidth]{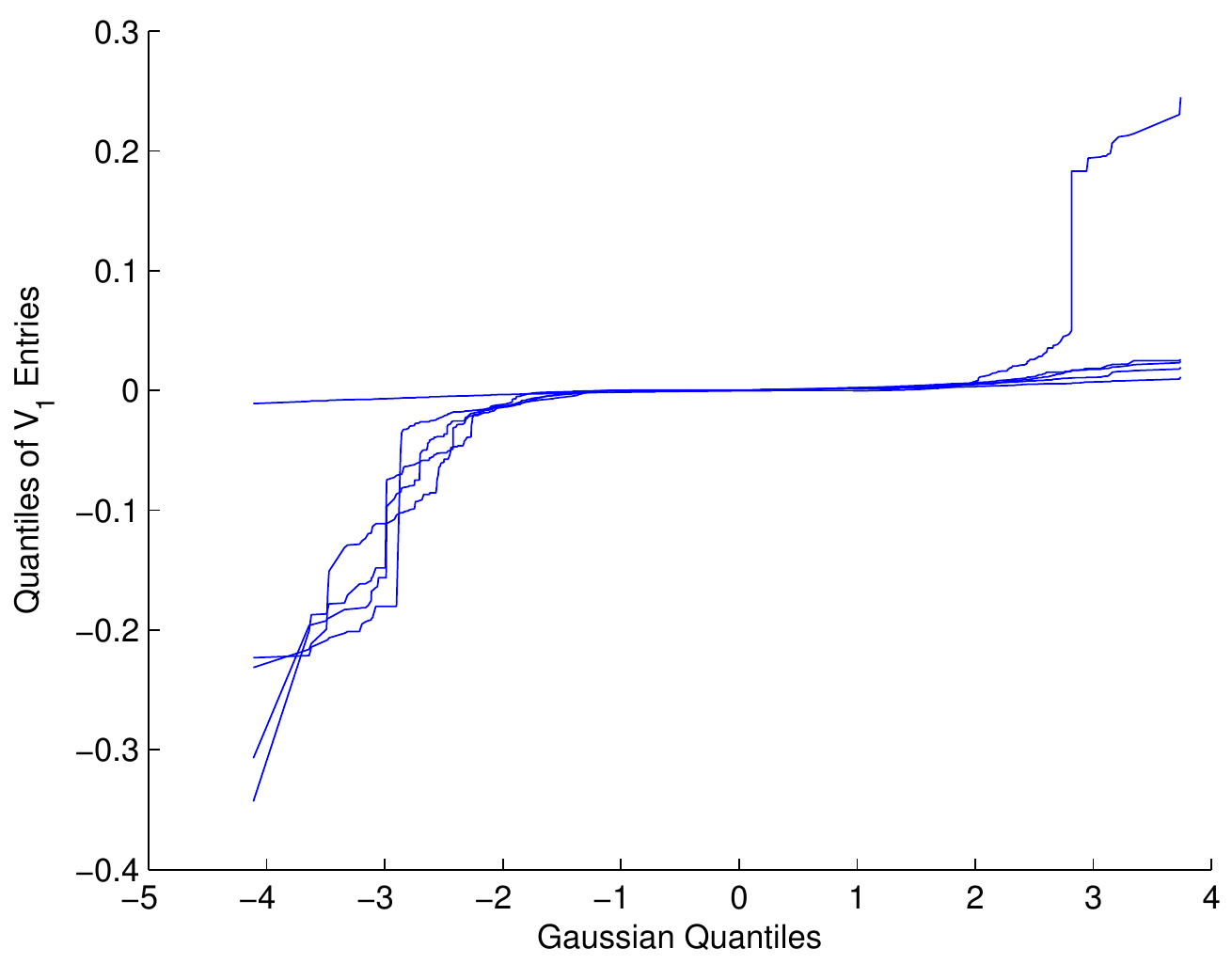}
    }
    \caption{Quantile-quantile plots measuring the Gaussianity of 5
      realizations of the entries of $\bs{V}_1$ for (a)
      Bernoulli(0.1)-distributed $\bs{X}$ and (b) Bernoulli(0.001)-distributed $\bs{X}$.}
\label{qq}
\end{figure}
The figure shows that the distribution for Bernoulli(0.1) is very
nearly Gaussian, but that this is far from the case for
Bernoulli(0.001), and thus it is the problem described at the end
  of the preceding paragraph that causes problems.

\subsection{2-Step Procedure}
\label{2step}
Note that in the variance upper-bound of Equation~\eqref{varUB}, the
unknown $\rho$ is maximized over to remove it from the equation. This
leads not only to conservative CIs, but suboptimal $\bs{w}^*$ as well,
since $\bs{w}^*$ are obtained by minimizing this upper-bound, as
opposed to the more accurate function of $\rho$. However by the end of
the EigenPrism procedure, we have produced estimates of both $\theta$
and $\theta^2+\sigma^2$, suggesting the possibility of a 2-step
plug-in procedure to remove the need for the upper-bound in
Equation~\eqref{varUB}. Explicitly, in the first step, we run the
EigenPrism procedure to obtain an estimate
$\hat{\rho}=T_2/(\|\bs{y}\|_2^2/n)$ of $\rho$. In the second step, we
re-run the procedure treating $\rho=\hat{\rho}$ as known, and thus
minimize the bound \eqref{varUB1} to compute $\bs{w}^*$. Although the
2-step procedure indeed produces shorter CIs than the EigenPrism
procedure, it does not achieve nominal coverage with the same
  consistency, as shown in Figure~\ref{fig2step}.

\begin{figure}[!t]
  \centering
    \subfigure[]{\label{twostepcoverage}
      \includegraphics[width=0.45\textwidth]{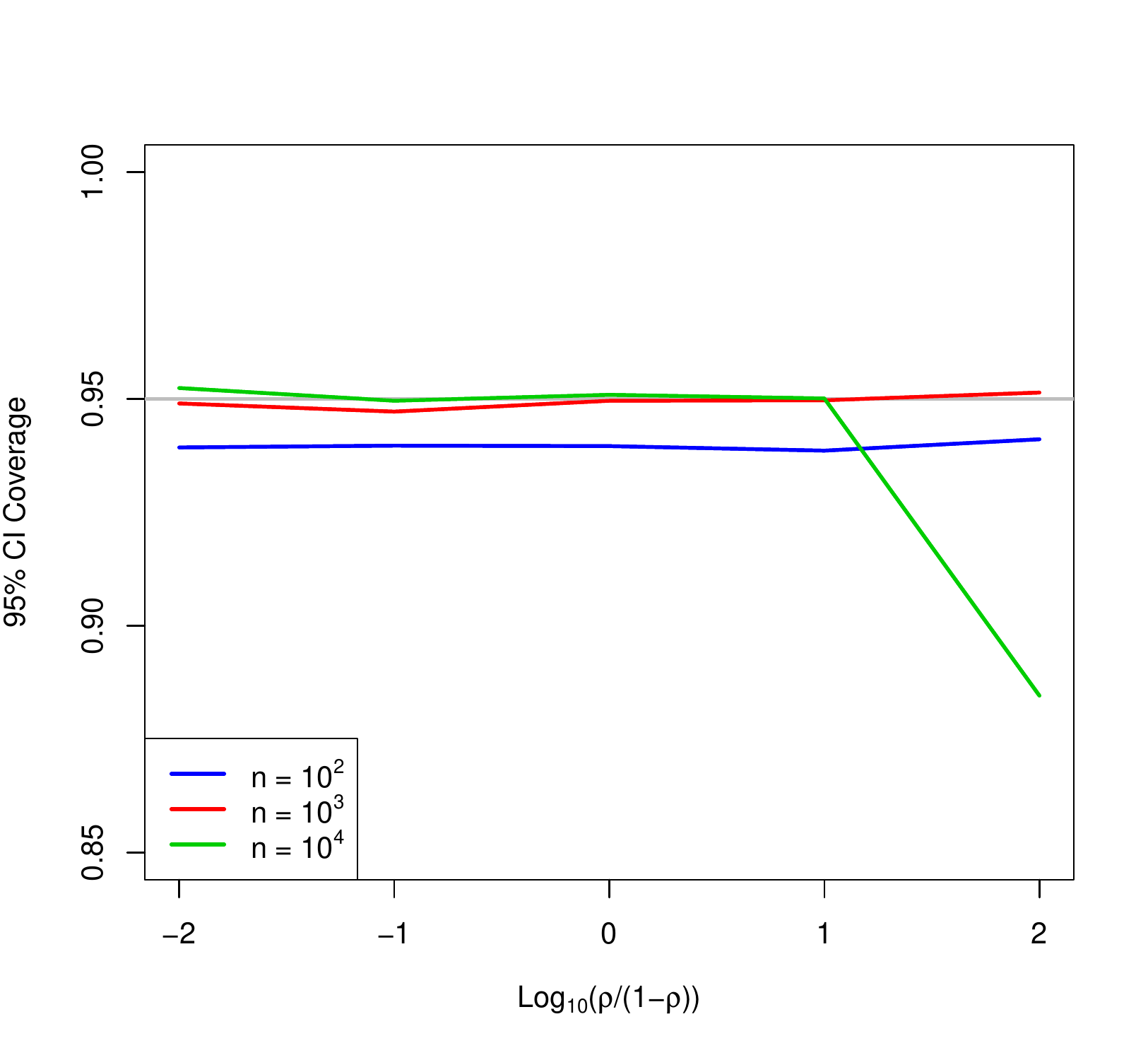}
    }
    \qquad
    \subfigure[]{\label{twostepratio}
      \includegraphics[width=0.45\textwidth]{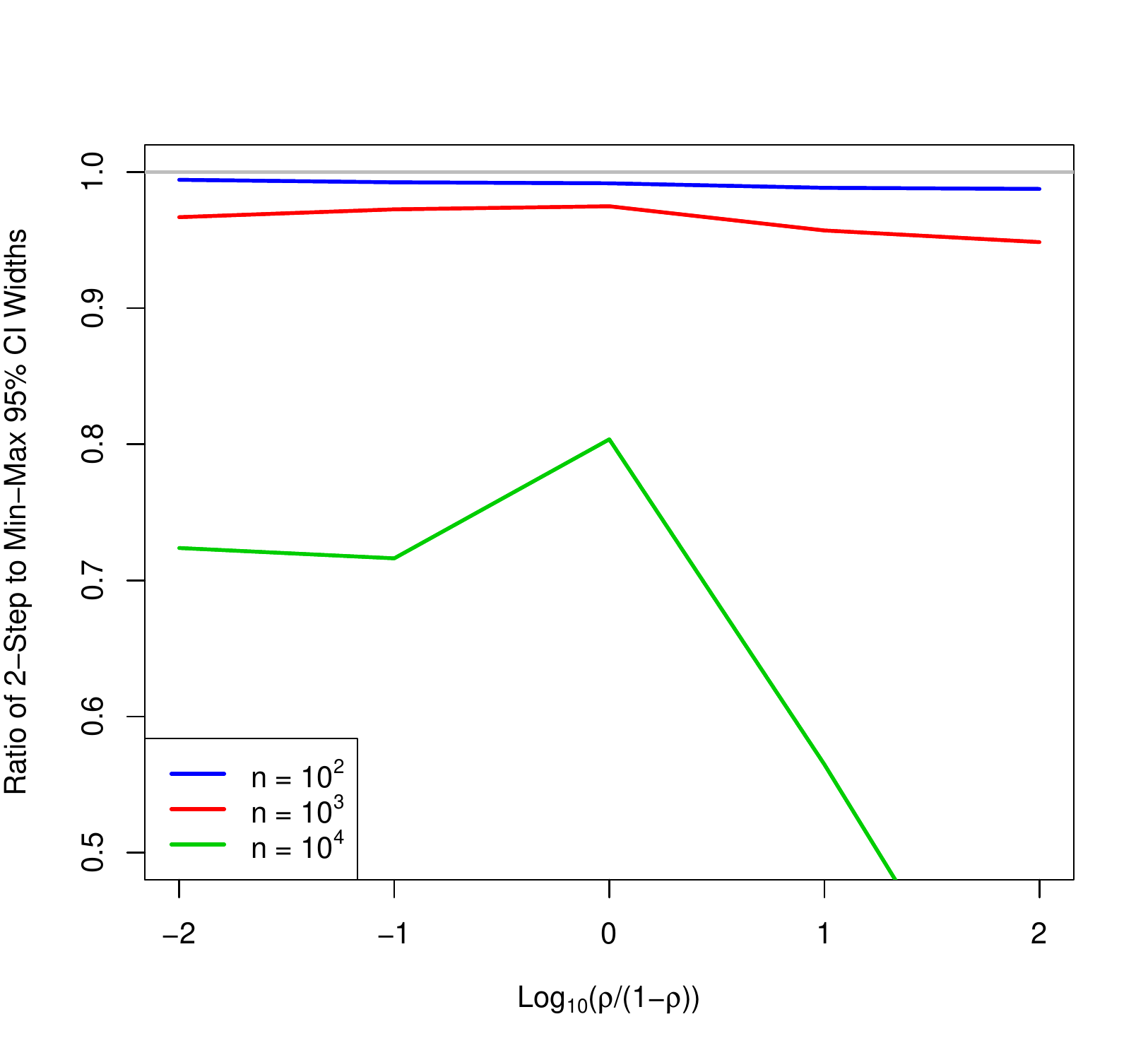}
    }
    \caption{(a) Coverage and (b) width relative to EigenPrism of 2-step
      confidence intervals for $p=10^4$ and $\theta^2+\sigma^2=10^4$ across a range of
      $\rho$ and $n$, with each point representing $10^4$
      simulations. Grey lines show (a) nominal coverage and (b)
      reference ratio of 1.}
\label{fig2step}
\end{figure}
There are two particularly surprising aspects of this plot. The first
is that the 2-step procedure produces substantial gains in width even
for $\rho$ values near 0 and 1. This is surprising because the
upper-bound \eqref{varUB} that is eliminated by the 2-step procedure
is tight when $\rho$ is nearly 0 or 1, however it is still not
exact. The slightly loose variance upper bound turns out to have an
optimizing $\bs{w}$ that is substantially different from the exact
variance formula. The second surprising feature is that the width
improvement is in fact \emph{smallest} for $\rho$ not near the
endpoints $0$ or $1$. This can be explained by the clipping at 0. For
$\rho\approx 0$, most CIs, both EigenPrism and 2-step, are cut nearly
in half by clipping, so the fractional width improvement achieved by
the 2-step procedure is fully realized. For $\rho\approx 1$, both
intervals are rarely clipped, and again the 2-step procedure realizes
its full width improvement. However, for $\rho$ not close to 0 or 1,
many EigenPrism CIs are only slightly shrunk by clipping, so that the
shorter 2-step intervals shorten the right side of the interval but
leave the unclipped left side about the same, so that much less than
the full width improvement is realized.

Although the 2-step procedure can provide substantial gains in width,
it loses the robustness of the EigenPrism procedure, as shown in the
slight undercoverage for $n=100$ and the substantial undercoverage for
large $\rho$ and $n=p$. Therefore, in practice, we recommend use of the 2-step
procedure instead of the EigenPrism procedure when $n\not\approx p$ or when the
statistician is confident that $\rho$ is not close to 1.

\section{Variance Decomposition in the Northern Finland Birth Cohort}
\label{nfbc}
We now briefly show the result of applying EigenPrism to a
dataset of SNPs and continuous phenotypes to perform
inference on the $\textsc{SNR} = \|\bs{\Sigma}^{1/2}\bs{\beta}\|_2^2/(\|\bs{\Sigma}^{1/2}\bs{\beta}\|_2^2 + \sigma^2)$. The data we use
comes from the Northern Finland Birth Cohort 1966 (NFBC1966)
\citep{sabatti2009,jarvelin2004}, made available through the dbGaP
database (accession number phs000276.v2.p1). The data consists of 5402
SNP arrays from subjects born in Northern Finland in 1966, as well
as a number of phenotype variables measured when the subjects were 31 years
old. After cleaning and processing the data (the details of which are
provided in Appendix~\ref{processing}), 328,934 SNPs remained. The
resulting $5402 \times 328,934$ design matrix $\bs{X}$ contained
approximately 58\% 0's (homozygous wild type), 34\% 1's
(heterozygous), and 8\% 2's (homozygous minor
allele).

In order to use EigenPrism directly, we would need to know
$\bs{\Sigma}$, as simply using $\bs{I}_p$ presents two possible
problems:
\begin{itemize}
\item[(1)] If $\bs{X}$ is not whitened before taking the SVD, the
  columns of $\bs{V}$ may be far from Haar-distributed, rendering
  our bias and variance computations incorrect.
\item[(2)] If $\bs{\Sigma}=\bs{I}_p$, then the ostensible target of
  our procedure is $\|\bs{\beta}\|_2^2/(\|\bs{\beta}\|_2^2 +
  \sigma^2)$, which may differ substantially from $\textsc{SNR} =
  \|\bs{\Sigma}^{1/2}\bs{\beta}\|_2^2/(\|\bs{\Sigma}^{1/2}\bs{\beta}\|_2^2
  + \sigma^2)$.
\end{itemize}
Unfortunately, the problem of estimating the covariance
matrix of a SNP array is extremely challenging (and the subject of much
current research) due to the fact that $n\ll p$, even if we use
outside data, so we prefer to avoid it here. In order to simply treat
the covariance matrix as diagonal, we must consider the two problems
above. There is a widely-held belief that the SNP locations that are
important for any given trait are relatively rare (see, for example,
\cite{yang2010common,Golan2011a}), and thus spaced far enough apart on
the genome to be treated as independent. This precludes problem (2)
above, since with nonzero coefficients spaced far apart, we have
$\|\bs{\Sigma}^{1/2}\bs{\beta}\|_2^2 \approx
\|\bs{\beta}\|_2^2$  (we take the columns of $\bs{X}$ to be standardized,
so the diagonal of $\bs{\Sigma}$ is all ones). For problem (1), we
know that far apart SNPs are very nearly independent, so we may expect
that the true $\bs{\Sigma}$ is roughly diagonal, and we already showed in
Section~\ref{robust} that the EigenPrism procedure is robust to some
small unaccounted-for covariances when constructing CIs for
$\|\bs{\beta}\|_2^2$. To ensure that problems (1) and (2) do not cause
EigenPrism to break down, we perform a series of diagnostics before
applying it to the real data.

Given the approximation of $\bs{\Sigma}$ as diagonal, we
first performed a series of simulations to ensure EigenPrism's accuracy
was not affected. Specifically, we ran the EigenPrism procedure (with
adjustments described in the paragraph below) on
artificially-constructed traits, but using the same standardized
design matrix $\bs{X}$ from the NFBC1966 data set. For 20 different
$\bs{\beta}$ vectors, we generated 500 independent Gaussian noise
realizations and recorded the coverage of 95\% EigenPrism CIs for
\textsc{SNR}. The noise variance was 1, and the $\bs{\beta}$'s were
chosen to have 300 nonzero entries with uniformly distributed
positions and all nonzero entries equal to $\sqrt{0.3/[(1-0.3)\cdot
  300]}$ (so that $\textsc{SNR} = 0.3$ if
$\bs{\Sigma}=\bs{I}_p$). Table~\ref{nfSim} shows the coverage over the
20 $\bs{\beta}$'s, and they are indeed all quite close to 95\%, even
though this simulation was \emph{conditional} on $\bs{X}$. Recomputing
the target \textsc{SNR} using other estimates of $\bs{\Sigma}$, such
as hard-thresholding the empirical covariance at 0.1, changed the
value of \textsc{SNR} very little, so that coverage was largely
unaffected.
\begin{table}[ht]\centering
\begin{tabular}{|r|c|c|c|c|c|c|c|c|c|c|c|}
\hline
Coverage & 90\% & 91\% & 92\% & 93\% & 94\% & 95\% & 96\% & 97\% & 98\% & 99\% & 100\% \\
\hline
Count & 1 & 0 & 0 & 2 & 1 & 0 & 8 & 7 & 1 & 0 & 0 \\
\hline
\end{tabular}
\caption{Coverage of 20 \textsc{SNR} 95\% CIs constructed for simulated
    traits using the NFBC1966 design matrix. Each coverage is an
    average over 500 simulations.}
\label{nfSim}
\end{table}

\begin{figure}[!t]
  \centering
  \includegraphics[width=0.45\textwidth]{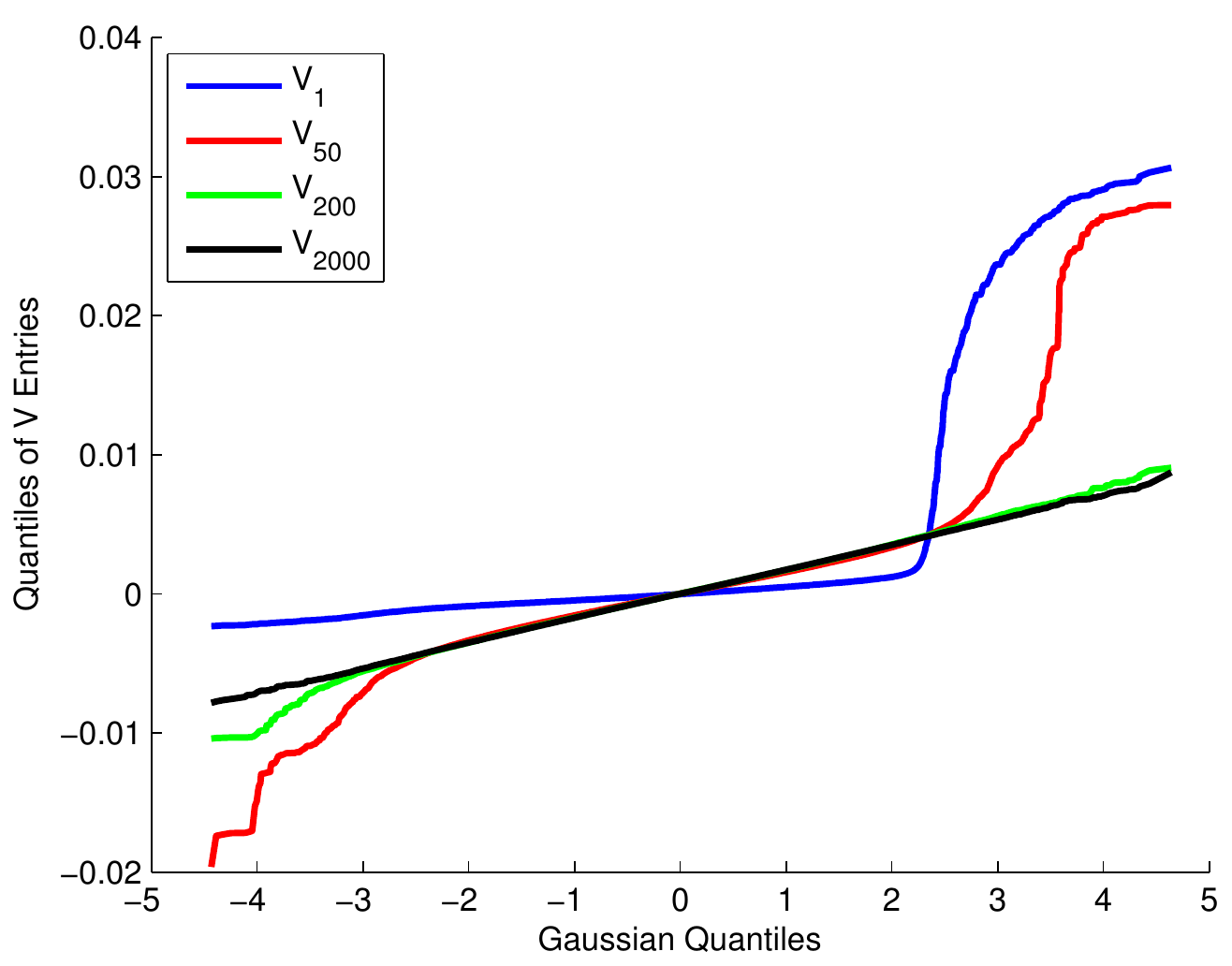}
  \caption{Distribution of the entries of some eigenvectors of the
    NFBC1966 design matrix.}
\label{nfV}
\end{figure}
A second diagnostic was to examine the columns of
$\bs{V}$ to check for Gaussianity related to the phenomenon mentioned in
Section~\ref{robust}. Indeed, we find that some of the columns of
$\bs{V}$ are quite non-Gaussian, as shown in
Figure~\ref{nfV}. However, this phenomenon is localized to only the
columns of $\bs{V}$ corresponding to the very largest
$\lambda_i$. Applying the unaltered EigenPrism procedure could cause
two problems. First, if the the first columns of
$\bs{V}$ are not Haar-distributed, $T_2$ could be biased and/or
higher-variance than our theory accounts for. Second, recalling the
interpretation of EigenPrism as a weighted regression of $z_i^2$ on
$\lambda_i$, the fact that the problematic eigenvectors correspond to
the largest eigenvalues means that they have high \emph{leverage},
which exacerbates any unwanted bias or variance they create. Luckily,
both problems can be remedied by running the EigenPrism
\textsc{SNR}-estimation procedure (with $\bs{\Sigma}=\bs{I}_p$ after
standardizing the columns of $\bs{X}$) with the added constraint to
the optimization program in Equation~\eqref{wopt} that the first
entries of $\bs{w}$ are equal to zero. Explicitly, as the non-Gaussianity
of the columns of $\bs{V}$ appears to dissipate after around the
100$^{th}$ column, we set $w_1 = \cdots = w_{100} = 0$. The choice of
100 is somewhat subjective, but we tried other values and obtained
very similar results. Because the resulting weights still obey the
original constraints, the estimator of $\|\bs{\beta}\|_2^2$ remains unbiased and the
variance upper-bounds remain valid. Although motivated by the
diagnostics from Section~\ref{robust}, this adjustment has the added
advantage of making the entire EigenPrism procedure completely independent of
the first 100 rows of $\bs{U}$. It has been shown that
the first rows of $\bs{U}$ are strongly related to the population
structure of the sample (for example, the first two principal
components correspond closely with subjects' geographic origin), so
constraining the first weights to be zero has the added effect of
controlling for population structure \citep{price2006}. As a final
note, by subtracting off the means of each column of $\bs{X}$, we
reduced $\bs{X}$'s rank by one, resulting in $\lambda_n=0$. As this is
not actually reflective of the distribution of $\bs{X}$, we also force
$w_n=0$ so that the last column of $\bs{V}$ and last row of $\bs{U}$
do not contribute to our estimate or inference.

Encouraged by the simulation results from Table~\ref{nfSim}, we proceeded to generate
EigenPrism CIs for the \textsc{SNR}s of the 9 traits
analyzed in \cite{sabatti2009}, as well as height (these 10 traits
were also analyzed in \cite{kang2010}). For each trait, transformation
and subject exclusion was performed
before computing \textsc{SNR}, following closely the procedures used
in \cite{sabatti2009, kang2010} (see Appendix~\ref{processing} for
details). Lastly, all non-height phenotype values were adjusted for
sex, pregnancy status, and oral contraceptive use, while height was
only adjusted for sex.
\begin{table}[ht]\centering
\begin{tabular}{|r|c|c|c|}
\hline
Phenotype Name & \# Samples & \textsc{SNR} 95\% CI (\%) & Point
Estimate (\%) \\
\hline
Triglycerides & 4644 & [3.1, 29.3] & 16.2 \\
\hline
HDL cholesterol & 4700 & [17.1, 42.9] & 30.0 \\
\hline
LDL cholesterol & 4682 & [27.7, 53.6] & 40.7 \\
\hline
C-reactive protein & 5290 & [5.6, 28.8] & 17.2 \\
\hline
Glucose & 4895 & [4.0, 28.9] & 16.5 \\
\hline
Insulin & 4867 & [0.0, 21.5] & 9.0 \\
\hline
BMI & 5122 & [8.9, 32.8] & 20.9 \\
\hline
Systolic blood pressure & 5280 & [7.8, 31.0] & 19.4 \\
\hline
Diastolic blood pressure & 5271 & [7.4, 30.7] & 19.0 \\
\hline
Height & 5306 & [46.0, 69.1] & 57.6 \\
\hline
\end{tabular}
\caption{CIs for heritability estimates 
  for each of the 10 continuous phenotypes considered, along with
  the number of samples used for each.}
\label{Htab}
\end{table}
Table~\ref{Htab} gives the point estimate and 95\% CI for the \textsc{SNR} of each
phenotype, as well as the number of subjects used. Recall that these are CIs for
the fraction of variance explained by the linear model consisting of
the given array of SNPs. Still, these CIs generally agree quite well
with heritability estimates in the literature
\citep{kang2010}. For instance, \cite[Supplementary Information]{kang2010} reports two
``pseudo-heritability'' estimates of 73.8\% and 62.5\% for height, and
27.9\% and 24.2\% for BMI, on the same data set. This is somewhat
remarkable given that they use a completely different statistical
procedure with different assumptions. In particular, while other works
in the heritability literature tend to treat $\bs{\beta}$ as random,
EigenPrism was motivated by a simple model with $\bs{\beta}$ fixed and
the rows of $\bs{X}$ random. We find this model more realistic, as
true genetic effects are not in fact random, but fixed. One could
argue the difference is not too important as long as the genetic
effects are approximately distributed as the random effects model
chosen, but such an assumption is impossible to verify in practice, as
the true effects are never observed. EigenPrism's assumptions, on the other
hand, are all on the design matrix, which is fully observed, leading
to checks and diagnostics that can be performed to ensure the
procedure will generate reasonable CIs.

\section{Discussion}
\label{discussion}
We have presented a framework for performing inference
on the $\ell_2$-norm of the coefficient vector in a linear regression
model. Although the resulting confidence intervals are asymptotic, we
show in extensive simulations that they achieve nominal coverage in
finite samples, without making any assumption on the structure or
sparsity of the coefficient vector, or requiring knowledge of
$\sigma^2$. In simulations, we are able to relax the restrictive
assumptions on the distribution of the design matrix and gain an
understanding of when our procedure is not appropriate. Applying this
framework to performing inference on $\ell_2$ regression error, noise level,
and genetic signal-to-noise ratio, we develop new procedures in all three
that are able to construct accurate CIs in situations not previously
addressed in the literature.

This work leaves open numerous avenues for further study. We briefly
introduced a 2-step procedure that provided substantially shorter CIs
than the EigenPrism procedure,
but had less-consistent coverage. If we could better understand that
procedure or come up with diagnostics for when it would
undercover, we could improve on the EigenPrism procedure. We also
explored in simulation a number of model failures that our procedure was robust (or
not) to, but further study could provide theoretical guarantees
on the coverage of the EigenPrism procedure for a broader class of random design
models. Section~\ref{robust} also briefly alluded to improved
robustness in a random effects framework, which we have not explored
further here. Finally, although in this work we consider a statistic that is
linear in the $z_i^2$, the framework and ideas of this work are not
intimately tied to this restriction, and there may exist statistics
that are nonlinear functions of the $z_i^2$ that give improved performance.

\section*{Acknowledgements}

We owe a great deal of gratitude to Chiara Sabatti for her patience in
explaining to us key concepts in statistical genetics and for her
guidance. We also thank Art Owen for sharing his unpublished notes
with us and for his constructive feedback, and Matthew Stephens and
Xiang Zhu for their helpful discussions on covariance estimation of
SNP data.  L.~J.~was partially supported by NIH training grant
T32GM096982. E.~C.~is partially supported by a Math + X Award from the
Simons Foundation.  The NFBC1966 Study is conducted and supported by
the National Heart, Lung, and Blood Institute (NHLBI) in collaboration
with the Broad Institute, UCLA, University of Oulu, and the National
Institute for Health and Welfare in Finland. This manuscript was not
prepared in collaboration with investigators of the NFBC1966 Study and
does not necessarily reflect the opinions or views of the NFBC1966
Study Investigators, Broad Institute, UCLA, University of Oulu,
National Institute for Health and Welfare in Finland and the
NHLBI. 

\bibliography{references}
\bibliographystyle{rss}

\appendix

\section{Inference for $\theta^2$ under non-Gaussian design with known
  variance}
\label{nonGaussKnownSigma}

The method of Section~\ref{para} also works asymptotically under more general
conditions than the Gaussianity assumptions of \eqref{normassump}. Let
$\bs{z} \sim (\bs{\mu}, \bs{\Sigma})$ denote the statement that
$\bs{z}$ has some distribution with mean $\bs{\mu}$ and covariance
matrix $\bs{\Sigma}$. Consider again the linear model \eqref{linmod}
but with relaxed assumptions,
\begin{equation*}
\bs{x}_i \iid (\bs{\mu}, \bs{I}_p - \bs{\mu}\bs{\mu}^{\top}), \qquad \varepsilon_i \iid (0, \sigma^2),
\end{equation*}
again with $\sigma^2$ known and $\bs{X}$ independent of
$\bs{\varepsilon}$. Under this model, we get that
\begin{equation*}
y_i^2 \iid
(\theta^2 + \sigma^2, v_1)
\end{equation*}
and the asymptotic distribution in \eqref{chi2distr} in turn becomes, by the CLT,
\begin{equation}
\label{normdistr}
\frac{1}{\sqrt{n}} \| \bs{y} \|_2^2 - \sqrt{n}\big(\theta^2 + \sigma^2\big)
\stackrel{\mathcal{D}}{\longrightarrow} N(0, v_1),
\end{equation}
as $n \rightarrow \infty$, where $v_1$ does not depend on $n$ but
does depend on the unknown $\bs{\beta}$, and is given by
\begin{equation*}
\begin{split}
v_1 & = \E{\varepsilon_i^4} + 4\sigma^2 \Big[\theta^2 - \big(\bs{\mu}^{\top}
\bs{\beta}\big)^2\Big] + 4\E{\varepsilon_i^3} \bs{\mu}^{\top} \bs{\beta} +
\mathbb{E}\left[\big(\bs{x}_i^{\top} \bs{\beta}\big)^4\right] \\
& \qquad - \sigma^4 - \theta^4
- \big(\bs{\mu}^{\top} \bs{\beta}\big)^4 + 2\theta^2 \big(\bs{\mu}^{\top} \bs{\beta}\big)^2 \\
\end{split}
\end{equation*}

In order to be less parametric, we can consider bootstrap
confidence intervals based on the above calculations. Corresponding to
\eqref{normdistr} we can get an unbiased statistic, 
\begin{equation}\label{eqn:T1}
\begin{split}
T_1 & := \frac{1}{n}\| \bs{y}\|_2^2 -
\sigma^2, \\
\E{T_1} & = \theta^2, \\
\sd(T_1) & = \sqrt{v_1/n}, \\
\end{split}
\end{equation}
whose distribution we may hope to be close to Gaussian. $T_1$ can be
bootstrapped (potentially with standard bias-correction and
acceleration) to obtain bootstrap CIs, nonparametrically dealing with
the unknown variance $v_1$. We ran simulations with $n=800$, $p=1500$,
$\bs{X}$ having i.i.d. Bernoulli(0.05) entries (the columns of $\bs{X}$ were
then standardized to have mean 0 and variance 1), $\theta^2 = \sigma^2
= 10$, and $\varepsilon_i$ i.i.d. $t_5$ (rescaled to have variance
10). We generated a single $\bs{\beta}$ uniformly on the
$\theta$-radius sphere and ran 1000 simulations (so that $\bs{\beta}$
did not change across simulations). Bias-corrected, accelerated
95\% bootstrap CIs achieved 93.8\% coverage (this is
within statistical uncertainty of the nominal 95\%, as a 95\% CI for
the CI coverage is $[0.923,0.953]$).

\section{Calculation of variance of EigenPrism estimator}
\label{minmaxvar}

In this section we calculate the variance of the statistic
$S=\sum_{i=1}^n w_iz_i^2$ when conditioning on $\bs{d}$.  Here we
treat $\bs{w}$ as fixed, but note that since we condition on $\bs{d}$,
this includes values of $\bs{w}$ that are calculated as a function of
$\bs{d}$, as in the EigenPrism method.

\begin{align*}
\var(S|\bs{d}) &=
\var\left(\left.\sum_{i=1}^nw_iz_i^2\right|\bs{d}\right)\\
&= \sum_{i=1}^nw_i^2\var\left(\left.z_i^2\right|\bs{d}\right) +
\sum_{\substack{i,j=1\\i\neq j}}^nw_iw_j\cov\left(\left.z_i^2,z_j^2\right|\bs{d}\right)\;.
\end{align*}

{We now calculate each term. Recall that $\lambda_i := d_i^2/p$ for $i=1,\dots,n$.} Then
\begin{align*}
\var\left(\left.z_i^2\right|\bs{d}\right)
&=\Ec{z_i^4}{\bs{d}}-\Ec{z_i^2}{\bs{d}}^2\\
&= \mathbb{E}\left[\left.(d_i\inner{\bs{V}_i}{\bs\beta}+\eps_i)^4\right|\bs{d}\right] - \left(\lambda_i \theta^2 + \sigma^2\right)^2\\
\intertext{Using $\eps_i\sim N(0,\sigma^2)$ {(and the fact that $\bs{\eps}\independent\bs{V}$),}}
&= \mathbb{E}\left[\left.(d_i\inner{\bs V_i}{\bs\beta})^4\right|\bs{d}\right] +
6\sigma^2\mathbb{E}\left[\left.(d_i\inner{\bs V_i}{\bs\beta})^2\right|\bs{d}\right] + 3\sigma^4 - \left(\lambda_i\theta^2 + \sigma^2\right)^2\\
\intertext{Using $\inner{\bs V_i}{\bs\beta}^2\sim \theta^2\cdot \mathsf{Beta}\left(\frac{1}{2},\frac{p-1}{2}\right)$,}
&= d_i^4\theta^4\cdot\frac{1\cdot 3}{p\cdot(p+2)} + 6\sigma^2 \lambda_i\theta^2 + 3\sigma^4 - \left(\lambda_i\theta^2 + \sigma^2\right)^2\\
&= 2\lambda_i^2\theta^4\frac{p-1}{p+2} + 4\sigma^2\lambda_i\theta^2 + 2\sigma^4
\end{align*}

Also, {for $i\neq j$},
\begin{align*}
\cov\left(\left.z_i^2,z_j^2\right|\bs{d}\right)
&= \Ec{z_i^2z_j^2}{\bs{d}} - \Ec{z_i^2}{\bs{d}}\Ec{z_j^2}{\bs{d}}\\
&= \mathbb{E}\left[\left.(d_i\inner{\bs V_i}{\bs\beta}+\eps_i)^2(d_j\inner{\bs V_j}{\bs\beta}+\eps_j)^2\right|\bs{d}\right] -  \left(\lambda_i\theta^2 + \sigma^2\right) \left(\lambda_j\theta^2 + \sigma^2\right)\\
\intertext{Using $\eps_i,\eps_j\iidsim N(0,\sigma^2)$ {(and the fact that $\bs{\eps}\independent\bs{V}$),}}
&= \Ec{d_i^2d_j^2\inner{\bs V_i}{\bs\beta}^2\inner{\bs
    V_j}{\bs\beta}^2}{\bs{d}} + \sigma^2\Ec{d_i^2\inner{\bs
    V_i}{\bs\beta}^2}{\bs{d}} + \sigma^2\Ec{d_j^2\inner{\bs
    V_j}{\bs\beta}^2}{\bs{d}}\\
&\qquad\qquad  + \sigma^4 - \left(\lambda_i\theta^2 + \sigma^2\right) \left(\lambda_j\theta^2 + \sigma^2\right)\\
\intertext{Using $\inner{\bs V_i}{\bs\beta}^2,\inner{\bs V_j}{\bs\beta}^2\sim \theta^2\cdot\mathsf{Beta}\left(\frac{1}{2},\frac{p-1}{2}\right)$,}
&= \Ec{d_i^2d_j^2\inner{\bs V_i}{\bs\beta}^2\inner{\bs
    V_j}{\bs\beta}^2}{\bs{d}} + \sigma^2\lambda_i\theta^2+
\sigma^2\lambda_j\theta^2 + \sigma^4\\
&\qquad\qquad - \left(\lambda_i\theta^2 + \sigma^2\right) \left(\lambda_j\theta^2 + \sigma^2\right)\\
&= d_i^2d_j^2\Ec{\inner{\bs V_i}{\bs\beta}^2\inner{\bs
    V_j}{\bs\beta}^2}{\bs{d}}  - \lambda_i\lambda_j\theta^4\\
&= d_i^2d_j^2\cov\left(\left.\inner{\bs V_i}{\bs\beta}^2,\inner{\bs V_j}{\bs\beta}^2\right|\bs{d}\right)\\
\intertext{Using $\left(\inner{\bs V_i}{\bs\beta}^2,\inner{\bs
    V_j}{\bs\beta}^2,\theta^2 - \inner{\bs V_i}{\bs\beta}^2 -
  \inner{\bs V_j}{\bs\beta}^2\right)\sim \theta^2\cdot\mathsf{Dirichlet}\left(\frac{1}{2},\frac{1}{2},\frac{p-2}{2}\right)$,}
&= \frac{-2}{p+2}\lambda_i\lambda_j\theta^4.
\end{align*}

Then,
\begin{align*}
\var\left(\left.S\right|\bs{d}\right) 
&=\sum_{i=1}^n w_i^2\left(2\lambda_i^2\theta^4\frac{p-1}{p+2} +
  4\sigma^2\lambda_i\theta^2 + 2\sigma^4\right) +
\sum_{\substack{i,j=1\\i\neq
    j}}^nw_iw_j\left(\frac{-2}{p+2}\lambda_i\lambda_j\theta^4\right)\\
&=\sum_{i=1}^n w_i^2\left(2\lambda_i^2\theta^4\frac{p}{p+2} +
  4\sigma^2\lambda_i\theta^2 + 2\sigma^4\right) +
\sum_{i=1}^n\sum_{j=1}^nw_iw_j\left(\frac{-2}{p+2}\lambda_i\lambda_j\theta^4\right)\\
&=2 \sigma^4 \sum_{i=1}^n w_i^2 + 4 \sigma^2
\theta^2 \sum_{i=1}^n w_i^2 \lambda_i + 2 \theta^4 \left[\frac{p}{p+2} \sum_{i=1}^n w_i^2 \lambda_i^2 - \frac{\left(\sum_{i=1}^n w_i \lambda_i\right)^2}{p+2}\right].\\
\end{align*}

\rev{
\section{Proof of Asymptotic Normality of $T_2$ and $T_3$}
\label{asymptoticnormality}
\begin{proof}
First consider $T_2(\bs{y},\bs{X})$ as a deterministic function of the
random $\bs{y}$ and $\bs{X}$. Then for any constant $c$,
\begin{equation}
\label{equivariance}
T_2(c\bs{y},\bs{X}) = c^2T_2(\bs{y},\bs{X}).
\end{equation}
Note that
$c\bs{y}$ also follows a linear model, only with $\theta^2$
replaced by $c^2\theta^2$ and $\sigma^2$ replaced by
$c^2\sigma^2$. Thus by taking $c=1/\max\{\theta,\sigma\}$ we may
treat $\theta^2$ and $\sigma^2$ as belonging to $[0,1]$ in order to
prove asymptotic normality of $T_2(c\bs{y},\bs{X})$, which by
Equation~\eqref{equivariance} implies asymptotic normality of
$T_2(\bs{y},\bs{X})$. The same argument holds for $T_3$, and so
without loss of generality, in the
remainder of the proof we assume $\theta^2$ and $\sigma^2$ are both
bounded. We have assumed $\max\{\theta^2,\sigma^2\}>0$,
as the case $\theta^2=\sigma^2=0$ is immediately identifiable because
$\bs{y}\equiv\bs{0}$, and trivial.

Recall that because $V$ is Haar-distributed,
\begin{equation*}
(\bs{V}_1^{\top}\bs{\beta}, \dots, \bs{V}_n^{\top}\bs{\beta})
\stackrel{d}{=} \theta/\|\bs{u}\| \cdot (u_1,\dots,u_n),
\end{equation*}
where $\bs{u} \sim N(\bs{0},\bs{I}_p)$. From this, we can rewrite
$T_2$ as:
\begin{align}
T_2 - \EE{T_2} &= \sum_{i=1}^n w_i
\left(\sqrt{\lambda_i}\theta\frac{u_i}{\|\bs{u}\|/\sqrt{p}}+\varepsilon_i\right)^2
- \sum_{i=1}^n w_i\left(\lambda_i\theta^2+\sigma^2\right)\nonumber\\
&= \frac{1}{\|\bs{u}\|^2/p} \cdot \sum_{i=1}^n
w_i\left(\sqrt{\lambda_i}\theta u_i +\varepsilon_i +
  \varepsilon_i\left(\|\bs{u}\|/\sqrt{p}-1\right)\right)^2-
\left(1+\frac{1}{\|\bs{u}\|^2/p}-\frac{1}{\|\bs{u}\|^2/p}\right)\sum_{i=1}^n w_i\left(\lambda_i\theta^2+\sigma^2\right)\nonumber\\
\label{eq:mainT2}&= \frac{1}{\|\bs{u}\|^2/p} \cdot \left(\sum_{i=1}^n
w_i\left(\sqrt{\lambda_i}\theta u_i +\varepsilon_i\right)^2 -
\sum_{i=1}^n w_i\left(\lambda_i\theta^2+\sigma^2\right)\right)\\ 
\label{eq:secondT2}&\qquad + \frac{2}{\|\bs{u}\|^2/p}
\left(\|\bs{u}\|/\sqrt{p}-1\right)\sum_{i=1}^n
w_i\left(\sqrt{\lambda_i}\theta u_i\varepsilon_i +
  \varepsilon_i^2-\sigma^2\right)\\ 
\label{eq:thirdT2}&\qquad + \frac{2\sigma^2}{\|\bs{u}\|^2/p}
\left(\|\bs{u}\|/\sqrt{p}-1\right)\sum_{i=1}^n w_i\\ 
\label{eq:fourthT2}&\qquad +
\frac{1}{\|\bs{u}\|^2/p}\left(\|\bs{u}\|/\sqrt{p}-1\right)^2\sum_{i=1}^n
w_i\left(\varepsilon_i^2-\sigma^2\right)\\ 
\label{eq:fifthT2}&\qquad +\frac{\sigma^2}{\|\bs{u}\|^2/p}
\left(\|\bs{u}\|/\sqrt{p}-1\right)^2\sum_{i=1}^n w_i\\ 
\label{eq:sixthT2}&\qquad + \left(1-\frac{1}{\|\bs{u}\|^2/p}\right)\sum_{i=1}^n w_i
\left(\lambda_i\theta^2+\sigma^2\right).
\end{align}
Our goal is to show that the right-hand side of \eqref{eq:mainT2} converges to
Gaussian, while \eqref{eq:secondT2}--\eqref{eq:sixthT2} each converge
to zero in probability. In particular, using certain probabilistic
properties of the $\lambda_i$'s and $w_i$'s (which are independent of the
other random variables), we will show convergence \emph{conditional}
on the $\lambda_i$ and $w_i$. We first prove the result for $T_2$ and then
explain the (minor) changes needed to prove the same for $T_3$ (for
which Equation~\eqref{eq:mainT2}--\eqref{eq:sixthT2} also holds).

Before either, however, we need a few tools, including the following Lemma:
\begin{lemma}\label{wbound}
For both $T_2$ and $T_3$, there exist constants $a$ and $b$ such that,
\[ \PP{\forall i, \,n |w_i| \le
  \frac{a+b\lambda_i}{\min\{\lambda_i^2,1\}}} \longrightarrow 1.\]
\proof We defer the proof to the end of this section.
\end{lemma}
Note that by convergence of the moments of the $\lambda_i$ to those of the
Mar{\v c}enko--Pastur (MP) distribution, Lemma~\ref{wbound} implies that 
\begin{equation}\label{w3bound}
\sum_{i=1}^n
|w_i|^3\lambda_i^r \in O_p(n^{-2})
\end{equation}
for any $r\in \mathbb{R}$. Note also that by the Cauchy-Schwarz
inequality,
\begin{equation}\label{w2bound}
\sum_{i=1}^nw_i^2 \lambda_i^r \ge
\frac{\left(\sum_{i=1}^n w_i\lambda_i\right)^2}{\sum_{i=1}^n\lambda_i^{2(1-r)}} =
\frac{1}{\sum_{i=1}^n\lambda_i^{2(1-r)}} \in \Omega_p(n^{-1})
\end{equation}
for any $r\in \mathbb{R}$. Finally, note that
$\|\bs{u}\|/\sqrt{p}\stackrel{p}{\longrightarrow} 1$.

Starting from the bottom, \eqref{eq:sixthT2} converges in
probability to zero because $\left(1-\frac{1}{\|\bs{u}\|^2/p}\right)
\stackrel{p}{\longrightarrow} 0$ and $\sum_{i=1}^n w_i
\left(\lambda_i\theta^2+\sigma^2\right) = \theta^2$, a
constant. \eqref{eq:fifthT2} and \eqref{eq:thirdT2}
equal zero because $\sum_{i=1}^n w_i = 0$. 

In \eqref{eq:fourthT2}, we seek to show that
$\sum_{i=1}^nw_i\left(\varepsilon_i^2-\sigma^2\right)$
converges in distribution, so that by Slutsky's Theorem,
\eqref{eq:fourthT2} converges to zero in probability. The summands
$w_i\left(\varepsilon_i^2-\sigma^2\right)$ are independent and
mean-zero with variance $2\sigma^4 w_i^2$. By Lyapunov's central
limit theorem \citep[p. 362]{billingsley1995}, we just need to establish the Lyapunov condition:
\begin{equation*}
\frac{\sum_{i=1}^n\EE{|w_i\left(\varepsilon_i^2-\sigma^2\right)|^3}}{\left(\sum_{i=1}^n\var\left(w_i\left(\varepsilon_i^2-\sigma^2\right)\right)\right)^{3/2}}
\propto\frac{\sum_{i=1}^n|w_i|^3}{\left(\sum_{i=1}^nw_i^2\right)^{3/2}}
\in O_p(n^{-1/2}),
\end{equation*}
where the $\in$ follows from \eqref{w3bound} and \eqref{w2bound}.

In \eqref{eq:secondT2}, we similarly seek to show that the
sum converges in distribution, allowing us to again use Slutsky's
Theorem to show \eqref{eq:secondT2} converges to zero in
probability. The argument is nearly the same as that for
\eqref{eq:fourthT2}, using various different values of $r$ in
\eqref{w3bound} and \eqref{w2bound} to establish the Lyapunov
condition. 

Lastly for \eqref{eq:mainT2}, by Slutsky's Theorem \citep[p. 433]{lehmann2005b}, it
suffices to show that $\sum_{i=1}^n
w_i\left(\sqrt{\lambda_i}\theta u_i +\varepsilon_i\right)^2 -
\sum_{i=1}^n w_i\left(\lambda_i\theta^2+\sigma^2\right)$ converges in
distribution to a Gaussian random variable, which can again be established
using Lyapunov's central limit theorem in nearly the same way as in
the argument for \eqref{eq:fourthT2}. Note that the resulting
variance expression $\sum_{i=1}^n
\var\left(w_i\left(\sqrt{\lambda_i}\theta u_i
    +\varepsilon_i\right)^2\right) = \sum_{i=1}^n
2\left(\lambda_i\theta^2+\sigma^2\right)^2$ is not identical to the
variance of $S$ in \eqref{varExact}, but the quotient of the
two expressions converges to 1 as $n,p \longrightarrow \infty$.

Only a few changes to the above proof are needed for establishing
asymptotic normality of $T_3$. First, an analogue to
Equation~\eqref{w2bound} can be shown:
\begin{equation}\label{w2boundT3}
\sum_{i=1}^nw_i^2 \lambda_i^r \ge
\frac{\left(\sum_{i=1}^n w_i\right)^2}{\sum_{i=1}^n\lambda_i^{-2r}} =
\frac{1}{\sum_{i=1}^n\lambda_i^{-2r}} \in \Omega_p(n^{-1}).
\end{equation}

Next, in each of \eqref{eq:thirdT2},\eqref{eq:fifthT2}, and
\eqref{eq:sixthT2}, the sum equals a constant while the coefficient in
front of the sum converges in probability to zero. The arguments for
\eqref{eq:mainT2}, \eqref{eq:secondT2}, and \eqref{eq:fourthT2} take
the same form as for $T_2$ except using \eqref{w3bound} and
\eqref{w2boundT3} instead of \eqref{w2bound} to establish the Lyapunov condition.
\end{proof}
\begin{proof}[Proof of Lemma~\ref{wbound}]
We start by slightly rewriting the optimization program
$\mathcal{P}_1$:
\begin{equation}\label{P1rewrite}
\argmin_{\bs{w}\in\mathbb{R}^n} \; t \quad \text{such that}  \,
\sum_{i=1}^n w_i^2 \le t, \;\sum_{i=1}^n w_i^2 \lambda_i^2 \le t, \;
\sum_{i=1}^n w_i = 0, \; \sum_{i=1}^n w_i \lambda_i = 1.
\end{equation}
By the Karush-Kuhn-Tucker conditions for \eqref{P1rewrite}, the
gradient of the Lagrangian with respect to $(t,w_1,\dots,w_n)$ vanishes, i.e.,
\begin{equation}\label{dualvars}
1-\delta_1-\delta_2=0,
\end{equation}
\begin{equation}\label{lagrange}
w_i\left(2\delta_1 + 2\delta_2\lambda_i^2\right) + \kappa_1 + \kappa_2\lambda_i=0,
\end{equation}
where $\delta_1\ge 0$ and $\delta_2\ge 0$ are the
dual variables corresponding to the inequalities and $\kappa_1$ and
$\kappa_2$ are the dual variables corresponding to the
equalities. Rearranging Equation~\eqref{lagrange},
\begin{equation}\label{lagrangew}
w_i = \frac{-\kappa_1 -\kappa_2\lambda_i}{2\delta_1+2\delta_2\lambda_i^2}.
\end{equation}
By Equation~\eqref{dualvars} and dual positivity constraints, we have
$\delta_1,\delta_2\in [0,1]$. Observe that
\begin{equation*}
\min_{\delta_1\in [0,1],\delta_2=1-\delta_1} \delta_1+\delta_2\lambda_i^2 = \min\{\lambda_i^2,1\},
\end{equation*}
establishing a lower-bound on the denominator. Now it suffices to show that
$|\kappa_1|,|\kappa_2|\in O_p(1/n)$. 

Multiplying
Equation~\eqref{lagrange} by $w_i$ and summing over $i$,
\begin{equation}\label{sumwi1}
2\delta_1\sum_{i=1}^nw_i^2 + 2\delta_2\sum_{i=1}^nw_i^2\lambda_i^2
+ \kappa_1\sum_{i=1}^nw_i + \kappa_2\sum_{i=1}^nw_i\lambda_i = 0.
\end{equation}
By recalling that Equation~\eqref{varbound}
established that
$\max\{\sum_{i=1}^nw_i^2,\sum_{i=1}^nw_i^2\lambda_i^2\}\in O_p(1/n)$
and the constraints $\sum_{i=1}^nw_i=0$ and $\sum_{i=1}^nw_i\lambda_i=1$,
we have that $|\kappa_2| \in O_p(1/n)$. Next, by just summing Equation~\eqref{lagrange} over $i$,
\begin{equation}\label{sumwi2}
2\delta_1\sum_{i=1}^nw_i + 2\delta_2\sum_{i=1}^nw_i\lambda_i^2
+ n\kappa_1 + \kappa_2\sum_{i=1}^n\lambda_i = 0.
\end{equation}
By Cauchy-Schwarz,
$\sum_{i=1}^nw_i\lambda_i^2 \le
\sqrt{\sum_{i=1}^nw_i^2\sum_{i=1}^n\lambda_i^4}\in O_p(1)$, and
using that $|\kappa_2| \in O_p(1/n)$ and $\sum_{i=1}^n\lambda_i\in
O_p(n)$, we find that $|\kappa_1| \in O_p(1/n)$ and the Lemma is proved for $T_2$.

To see the same result for $T_3$, first note that rewriting
$\mathcal{P}_2$ analogously to \eqref{P1rewrite} gives the same gradient for the Lagrangian,
so that Equations~\eqref{dualvars} and \eqref{lagrange} still hold
with the same implications for the denominator of $w_i$ in
Equation~\eqref{lagrangew}, so all that remains is again showing that
$|\kappa_1|,|\kappa_2|\in O_p(1/n)$.

We will need an analogue to Equation~\eqref{varbound} for $T_3$ to
show that $\max\{\sum_{i=1}^nw_i^2,\sum_{i=1}^nw_i^2\lambda_i^2\}\in
O_p(1/n)$. The proof of Equation~\eqref{varbound} can be found in
Appendix~\ref{CFvarUB}, and follows from the construction of a simple
set of weights $\tilde{w_i}$ satisfying the constraints of
$\mathcal{P}_1$. By considering instead the set of weights 
\[\breve{w}_i :=\frac{\lambda_i^{-1}}{\sum_{j=1}^{n/2}\lambda_j^{-1} -
  \sum_{j=n/2+1}^n \lambda_j^{-1}}\cdot  \begin{cases}+1,&\text{ for
  }i\leq n/2,\\-1,&\text{ for }i>n/2,\end{cases}\] satisfying the
constraints of $\mathcal{P}_2$, one can follow the
same steps to establish $\max\{\sum_{i=1}^nw_i^2,\sum_{i=1}^nw_i^2\lambda_i^2\}\in
O_p(1/n)$ for $T_3$. Using this and the constraints of
$\mathcal{P}_2$, Equation~\eqref{sumwi1} establishes $|\kappa_1|\in
O_p(1/n)$. Using this result and the same methods as for $T_2$,
Equation~\eqref{sumwi2} establishes $|\kappa_2|\in O_p(1/n)$, and the
Lemma is proved.
\end{proof}
}

\section{Variance upper-bound for $T_2$}
\label{CFvarUB}

{In this section we derive the upper bound~\eqref{varbound} on the
variance of the statistic $T_2$. For simplicity we assume that $n$ is even.}

{We begin by constructing a vector of weights $\tilde{\bs{w}}$:
\[\tilde{w}_i :=\frac{1}{\sum_{j=1}^{n/2}\lambda_j - \sum_{j=n/2+1}^n \lambda_j}\cdot  \begin{cases}+1,&\text{ for }i\leq n/2,\\-1,&\text{ for }i>n/2.\end{cases}\]
Note that $\tilde{\bs{w}}$} satisfies the constraints of the
optimization problem~\eqref{wopt}, and thus $\var(T_2)$ is
upper-bounded by Equation~\eqref{varUB} with $\tilde{\bs{w}}$ plugged
in. A second key observation is that we know from random matrix theory
that for $n,p \rightarrow \infty$ and $n/p \rightarrow \gamma \in
(0,1)$, the distribution of rescaled eigenvalues,
$\lambda_i$, converges to the MP distribution with
parameter $\gamma$.

Recalling the definitions of
  $A_{\gamma},B_{\gamma}$ given in~\eqref{eqn:ABgamma}, this implies
  that
\[\frac{1}{n}\sum_{i=1}^{n}\lambda_i\cdot(\one{i\leq n/2}-\one{i>n/2}) \rightarrow A_{\gamma}\]
and
\[\frac{1}{n}\sum_{i=1}^n\lambda_i^2 \rightarrow B_{\gamma}\;.\]

Together with the definition of $\tilde{\bs{w}}$, these imply that as $n,p \rightarrow \infty$ with $n/p
\rightarrow \gamma \in (0,1)$,
\begin{equation*}
\begin{split}
n\sum_{i=1}^n\tilde{w}_i^2 =&\, \frac{n\cdot n}{\left\{n \cdot
    \left[\frac{1}{n}\sum_{i=1}^{n}\lambda_i\cdot(\one{i\leq
        n/2}-\one{i>n/2})\right]\right\}^2} \rightarrow 
  \frac{n\cdot n}{\left(nA_{\gamma}\right)^2}= \frac{1}{A_{\gamma}^2}, \\
n\sum_{i=1}^n\tilde{w}_i^2\lambda_i^2 =&\,
\frac{n\cdot n\cdot\left[\frac{1}{n}\sum_{i=1}^n\lambda_i^2\right]}{\left\{n\cdot \left[\frac{1}{n}\sum_{i=1}^{n}\lambda_i\cdot(\one{i\leq n/2}-\one{i>n/2})\right]\right\}^2}
\rightarrow \frac{n\cdot n\cdot B_{\gamma}}{\left(nA_{\gamma}\right)^2} = \frac{B_{\gamma}}{A_{\gamma}^2}, \\
\end{split}
\end{equation*}
which in turn implies
\begin{equation}
\sqrt{n} \cdot \frac{\sqrt{\var(T_2)}}{\theta^2+\sigma^2} \le \sqrt{2 \cdot
\max\left(n\sum_{i=1}^n\tilde{w}_i^2,
  n\sum_{i=1}^n\tilde{w}_i^2\lambda_i^2\right)} \rightarrow \sqrt{2}
\cdot \max\left(\frac{1}{A_{\gamma}},\frac{\sqrt{B_{\gamma}}}{A_{\gamma}}\right).
\end{equation}

\section{Proof of Theorem~\ref{thm:lambdas}}
\label{lambdasone}
\begin{proof}
For this proof, we use Le Cam's method (see e.g.~\citet[Lemma 1]{yu1997assouad}), which states that 
\[\mathbb{P}_{\bs{Z}\sim P_0}\left[\psi(\bs{Z})=1\right] + \mathbb{P}_{\bs{Z}\sim P_1}\left[\psi(\bs{Z})=0\right] \geq 1 - \norm{P_0-P_1}_{\mathsf{TV}}\;,\]
where $\norm{\cdot}_{\mathsf{TV}}$ is the total variation norm:
\[\norm{P_0-P_1}_{\mathsf{TV}} = \sup_{\Aset\subseteq \R^n}\left|\Pp{\bs{Z}\sim P_0}{\bs{Z}\in \Aset}-\Pp{\bs{Z}\sim P_1}{\bs{Z}\in \Aset}\right|\;,\]
where the supremum is taken over Lebesgue-measurable sets.

We begin by constructing a related distribution $Q_1$:
\begin{equation}\label{eqn:gaussian_distr} \bs{W} = \theta \cdot \bs{D V}^{\top} \bs{a} \cdot r + \sigma \cdot \bs{\varepsilon}\;,\end{equation}
where $\theta=\sigma=\frac{1}{\sqrt{2}}$, and where $r\sim \chi_p/\sqrt{p}$ is independent from $\bs{V},\bs{\varepsilon}$. We will bound
\[ \norm{P_0-P_1}_{\mathsf{TV}} \leq \norm{P_0-Q_1}_{\mathsf{TV}}  + \norm{P_1-Q_1}_{\mathsf{TV}} \;.\]

First, we use the fact that $r$ concentrates tightly near $1$ for the
following bound:
\begin{align*}
\EE{\norm{\theta \cdot \bs{D V}^{\top} \bs{a} \cdot r - \theta\cdot \bs{DV}^{\top} \bs{a}}_2}
&= \mathbb{E}\left[\EEst{\norm{\theta \cdot \bs{D V}^{\top} \bs{a} \cdot r - \theta\cdot \bs{DV}^{\top} \bs{a}}_2}{r,\bs{V}}\right]\\
&= \EE{\norm{\theta \cdot \bs{DV}^{\top} \bs{a}}_2\cdot \left|r-1\right|}\\
&\leq \theta\cdot \sqrt{\EE{\bs{a}^{\top} \bs{V} \bs{D}^2 \bs{V}^{\top} \bs{a}}}\cdot \sqrt{\mathbb{E}\left[(r-1)^2\right]}\\
&= \theta\cdot  \sqrt{\EE{\bs{a}^{\top} \bs{V} \bs{D}^2 \bs{V}^{\top} \bs{a}}}\cdot \frac{1}{\sqrt{p}}\sqrt{\rev{\EE{\chi^2_p}} - 2\sqrt{p}\cdot \EE{\chi_p}+p}
\end{align*} 
Using the fact that $\EE{\chi^2_p}=p$ and $\EE{\chi_p}\geq\sqrt{p} -
\frac{1}{4\sqrt{p}}$, and that $\EE{(\bs{V}_i^{\top}
  \bs{a})^2}=\frac{1}{p}$ for each $i=1,\dots,n$,
\begin{align*}
\EE{\norm{\theta \cdot \bs{D V}^{\top} \bs{a} \cdot r - \theta\cdot \bs{DV}^{\top} \bs{a}}_2}
&\leq \theta\cdot \sqrt{\sum_i \frac{D_{ii}^2}{p}}\cdot \frac{1}{\sqrt{2p}}\\
&= \frac{\theta\sqrt{n}}{\sqrt{2p}}\;,
\end{align*}
since $\frac{1}{p}\sum_i D_{ii}^2=n$.
Next, for any measurable set $\Aset\subseteq\R^n$, we have
\begin{align*}
&\left|\Pp{\bs{W}\sim Q_1}{\bs{W}\in \Aset} - \Pp{\bs{Z}\sim P_1}{\bs{Z}\in \Aset}\right|\\
&=\left|\mathbb{E}\left[\PPst{\bs{W}\in \Aset}{r,\bs{V}}-\PPst{\bs{Z}\in \Aset}{r,\bs{V}}\right]\right|\\
&\leq\mathbb{E}\left[\left|\PPst{\bs{W}\in \Aset}{r,\bs{V}}-\PPst{\bs{Z}\in \Aset}{r,\bs{V}}\right|\right]\\
&=\mathbb{E}\left[\left|\PPst{\theta \cdot \bs{D} \bs{V}^{\top} \bs{a} \cdot r + \sigma \cdot \bs{\varepsilon}\in \Aset}{r,\bs{V}}-\PPst{\theta \cdot \bs{D} \bs{V}^{\top} \bs{a} + \sigma \cdot \bs{\varepsilon}\in \Aset}{r,\bs{V}}\right|\right]\\
&\leq \mathbb{E}\left\{\mathbb{E}\left[\left.\norm{N\left(\theta \cdot \bs{D} \bs{V}^{\top} \bs{a}\cdot r,\sigma^2\ident_n\right)-N\left(\theta \cdot \bs{D} \bs{V}^{\top} \bs{a},\sigma^2\ident_n\right)}_{\mathsf{TV}}\right|r,\bs{V}\right]\right\}\\
\intertext{Using the fact that $\norm{N(\bs{\mu},\sigma^2\ident_n)-N(\bs{\mu}',\sigma^2\ident_n)}_{\mathsf{TV}}\leq \frac{\norm{\bs{\mu}-\bs{\mu}'}_2}{\sqrt{2\pi\sigma^2}}$ for any fixed $\bs{\mu},\bs{\mu}',\sigma^2$,}
&\leq \mathbb{E}\left[\EEst{\frac{\norm{\theta \cdot \bs{D} \bs{V}^{\top} \bs{a} \cdot r - \theta\cdot \bs{D}\bs{V}^{\top} \bs{a}}_2}{\sqrt{2\pi\sigma^2}}}{r,\bs{V}}\right]\\
&\leq \frac{1}{\sqrt{2\pi\sigma^2}}\cdot \frac{\theta\sqrt{n}}{\sqrt{2p}}\;,
\end{align*}
where the last step uses our calculations above.  Since this is true
for any $\Aset\subset\R^n$, and since
$\theta=\sigma=\frac{1}{\sqrt{2}}$ by assumption under the
distribution $P_1$, we have
\[\norm{P_1-Q_1}_{\mathsf{TV}}\leq \sqrt{\frac{n/p}{4\pi}}\;.\]

Next, we bound $\norm{P_0-Q_1}_{\mathsf{TV}}$. By Pinsker's inequality,
\[\norm{P_0-Q_1}_{\mathsf{TV}} \leq \sqrt{\frac{1}{2} \mathsf{KL}\left(Q_1 \| P_0\right)}\;,\]
where $\mathsf{KL}\left(\cdot \| \cdot\right)$ is the Kullback-Leibler divergence.
Note that the distributions $P_0$ and $Q_1$ can be reformulated as
\[P_0: Z_i\indsim N(0,1)\]
and 
\[Q_1: Z_i\indsim N\left(0, \frac{\lambda_i + 1}{2}\right)\;.\]
 Writing $p_0(\cdot)$ and $q_1(\cdot)$
to be the densities of the distributions $P_0$ and $Q_1$, respectively, we have
\begin{align*}
 \mathsf{KL}\left(Q_1 \| P_0\right)
 &=\mathbb{E}_{\bs{Z}\sim Q_1}\left\{\log\left[\frac{q_1(\bs{Z})}{p_0(\bs{Z})}\right]\right\}\\
 &=\mathbb{E}_{\bs{Z}\sim Q_1}\left[\log\left(\frac{\prod_{i=1}^n \frac{1}{\sqrt{2\pi\left(\frac{\lambda_i+1}{2}\right)}}e^{-\frac{Z_i^2}{\lambda_i+1}}}{\prod_{i=1}^n \frac{1}{\sqrt{2\pi}}e^{-\frac{Z_i^2}{2}}}\right)\right]\\
 &=\mathbb{E}_{\bs{Z}\sim Q_1}\left\{- \frac{1}{2}\sum_i \left[\log\left(\frac{\lambda_i+1}{2}\right)+Z_i^2\cdot\left(\frac{2}{\lambda_i+1} -1\right)\right]\right\}\\
\intertext{Since $\Ep{\bs{Z}\sim Q_1}{Z_i^2}=\frac{\lambda_i+1}{2}$,}
 &=-\frac{1}{2}\sum_i \left[\log\left(\frac{\lambda_i+1}{2}\right)+\left(1 -\frac{\lambda_i+1}{2}\right)\right]\\
 \intertext{Using the fact that
$\log(x) \geq (x-1) - 2(x-1)^2$ for all $x\geq \frac{1}{2}$,}
 &\leq -\frac{1}{2}\sum_i \left[\frac{\lambda_i-1}{2} - 2\left(\frac{\lambda_i-1}{2} \right)^2+\left(1 -\frac{\lambda_i+1}{2}\right)\right]\\
 &=\frac{1}{4}\sum_i (\lambda_i-1)^2\;.
\end{align*}
Combining everything, we have
\[\norm{P_0-Q_1}_{\mathsf{TV}} \leq \sqrt{\frac{1}{2} \mathsf{KL}\left(Q_1 \| P_0\right)}\leq \sqrt{\frac{1}{2}\cdot \frac{1}{4}\sum_i (\lambda_i-1)^2}\;,\]
and so
\[ \norm{P_0-P_1}_{\mathsf{TV}} \leq \sqrt{\frac{1}{8}\sum_i (\lambda_i-1)^2} +  \sqrt{\frac{n/p}{4\pi}}\;.\]
\end{proof}

\section{CVX code for computing the weight vector}
\label{cvx}
The following snippet of code was used with MATLAB Version 8.1
(R2013a) and CVX Version 2.1, Build 1085 on a 64-bit Linux
OS. The eigenvalues $\lambda_i$ are represented by the column vector
\texttt{lambda}, \texttt{t} corresponds to $2\val(\mathcal{P}_1)$, and the resulting
vector \texttt{w} corresponds to $\bs{w}^*$.

\begin{verbatim}
cvx_begin
variable t
variable w(n)
minimize t
subject to
    sum(w) == 0;
    sum(w .* lambda) == 1;
    norm([w; (t/2-1)/2])  <= (t/2+1)/2;
    norm([w .* lambda; (t/2-1)/2]) <= (t/2+1)/2;
cvx_end
\end{verbatim}

\section{Bayesian model}
\label{appbayes}
The Bayesian model is given explicitly as follows ($M$, $\bs{Z}$,
$\sigma^2$, and $\bs{\varepsilon}$ are all independent of one another):
\begin{equation}
\label{bayes}
\begin{split}
M \sim & \,\text{Exponential}(\lambda), \\
\bs{Z} \sim & \,N(0, I_p), \\
\bs{\beta} = & \,\sqrt{M} \bs{Z}, \\
\frac{1}{\sigma^2} \sim & \,\text{Gamma}(A,B), \\
\bs{\varepsilon} \sim & \,N(0, \sigma^2I_p), \\
\bs{y} = & \,\bs{X}\bs{\beta} + \bs{\varepsilon}, \\
\end{split}
\end{equation}
where the values for the parameters used were $\lambda =
\frac{50,000}{p}$ (so $\theta^2 \approx \text{Exponential}(1/2000)$),
$A = 14$, $B = \frac{1}{20,000}$, and we have used the shape/scale
parameterization of the Gamma distribution, as opposed to the
shape/rate parameterization. Figure~\ref{fig:priors} shows the
resulting priors for $\theta^2$, $\sigma^2$, and $\rho =
\frac{\theta^2}{\theta^2 + \sigma^2}$.
\begin{figure}[t!]
\centering
\includegraphics[trim=0cm 0cm 0cm 0cm, clip=true, width=16cm]{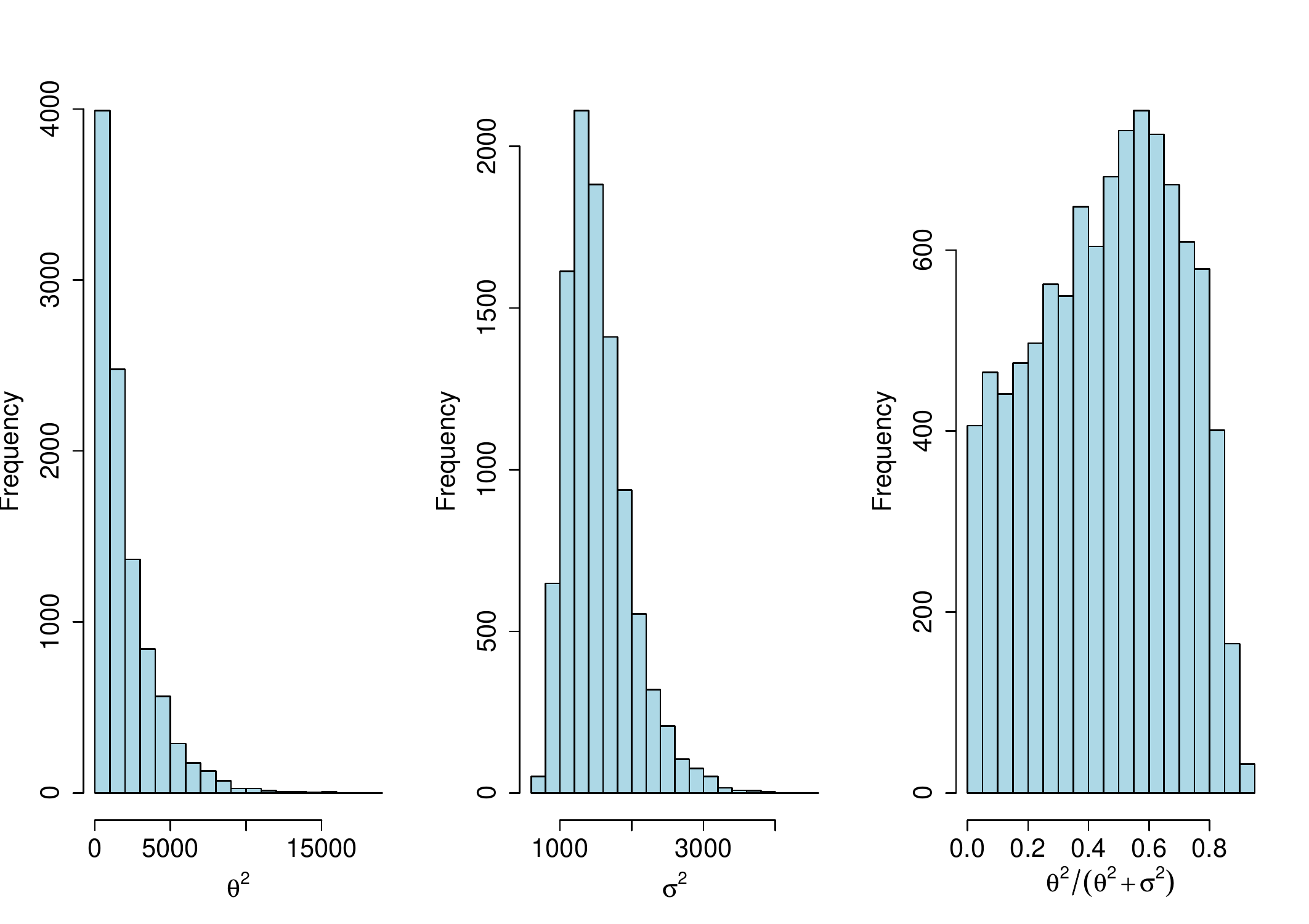}
\caption{Priors from model~\eqref{bayes}.}
\label{fig:priors}
\end{figure}

Note also that, although not shown, the posteriors achieved under this
setup were all unimodal, so that the equal-tailed credible intervals
were very close to the minimum-length credible intervals. We used
equal-tailed credible intervals to give fair comparison with the
EigenPrism CIs, which are also equal-tailed. The interval widths plotted
all have nominal coverage of 80\%. BCI endpoints were estimated by
empirical quantiles of posterior draws from a Gibbs sampler, and thus
we were able to much more accurately estimate the 10th and 90th
percentiles than, say, the 2.5th and 97.5th percentiles.

\section{Construction of correlated-column covariance matrices}
\label{corrMat}
Dense 10\% Correlations used a covariance matrix with ones on the
diagonal and 0.1's as all the other entries. The Sparse $100\cdot P\%$
Correlations used alternating $P$ and $-P$ as off-diagonal entries in
a correlation matrix, then
projected that matrix into the positive semidefinite cone and
reset the diagonal entries to 1. The resulting matrix has
approximately 1/4 of its entries equal to $P$, 1/2 of its entries
equal to 0, and 1/4 of its entries equal to $-2\times 10^{-4}\cdot (1-P)$. 

\section{Processing of NFBC1966 dataset}
\label{processing}
Genotype features from the original data set were removed if they met any of the
following conditions:
\begin{itemize}
\item Not a SNP (some were, e.g., copy number variations)
\item Greater than 5\% of values were missing
\item All nonmissing values belonged to the same nucleotide
\item SNP location could not be aligned to the genome
\item A $\chi^2$ test rejected Hardy-Weinberg equilibrium at the 0.01\%
  level
\item On chromosome 23 (sex chromosome)
\end{itemize}
The remaining missing values were assumed to take the major allele
value (thus were coded as 0's in the pre-centered design
matrix). 

For each trait, further processing was performed on the
subjects. Triglycerides, BMI, insulin, and glucose were all
log-transformed. C-reactive protein was also log-transformed after
adding 0.002 mg/l (half the detection limit) to 0 values. Subjects
were excluded from the triglycerides, HDL and LDL cholesterol,
glucose, and insulin analyses if they were on diabetic medication or
had not fasted before blood collection (or if either value was
missing). Further subjects were excluded from the triglycerides, HDL
and LDL cholesterol analyses if they were pregnant or if their
respective phenotype measurement was more than three standard
deviations from the mean, after correcting for sex, oral contraceptive
use, and pregnancy. Subjects whose weight was not directly measured
were excluded from BMI analysis. Of course any missing values in each
phenotype were also excluded.

\end{document}